\documentclass[11pt,titlepage=off,abstract=on]{scrartcl}
\usepackage[utf8]{inputenc}
\usepackage[
    paper=letterpaper,
    left=1in,
    right=1in,
    bottom=1in, 
    top=1in,
    footskip=0.33in,
    bindingoffset=0mm
]{geometry}

\usepackage[dvipsnames]{xcolor}
\definecolor{color3}{RGB}{140,29,49}
\definecolor{color2}{RGB}{31,107,76}
\definecolor{color1}{RGB}{29,66,91}
\usepackage[colorlinks=true,linkcolor=color1,urlcolor=color2,citecolor=color3,anchorcolor=green,pdfusetitle]{hyperref}

\let\latextextsuperscript\textsuperscript
\AtBeginDocument{\let\textsuperscript\latextextsuperscript}

\usepackage{amsmath,amsthm,amsfonts,mathtools,stmaryrd}
\usepackage{newtxtext,newtxmath,microtype,bbm,dsfont}
\usepackage{booktabs,multirow,makecell}

\usepackage[font=normalfont,labelfont=bf,format=plain]{caption}

\input{braket.tex}

\usepackage{algpseudocode,algorithm}
\algnewcommand\Yield{\textbf{yield\ }}
\algnewcommand\Break{\textbf{break}}

\usepackage{cleveref}
\crefname{lemma}{lemma}{lemmas}
\crefname{proposition}{proposition}{propositions}
\crefname{definition}{definition}{definitions}
\crefname{theorem}{theorem}{theorems}
\crefname{conjecture}{conjecture}{conjectures}
\crefname{corollary}{corollary}{corollaries}
\crefname{example}{example}{examples}
\crefname{section}{section}{sections}
\crefname{appendix}{appendix}{appendices}
\crefname{figure}{fig.}{figs.}
\crefname{equation}{eq.}{eqs.}
\crefname{table}{table}{tables}
\crefname{item}{property}{properties}
\crefname{remark}{remark}{remarks}

\newtheorem{theorem}{Theorem}
\newtheorem{definition}[theorem]{Definition}

\newtheorem{lemma}[theorem]{Lemma}

\usepackage[backend=bibtex8,style=alphabetic,sorting=nyt,doi=false,isbn=false,url=false,maxbibnames=20,maxcitenames=2]{biblatex}

\newbibmacro{string+doi}[1]{\iffieldundef{doi}{#1}{\href{https://dx.doi.org/\thefield{doi}}{#1}}}
\DeclareFieldFormat{title}{\usebibmacro{string+doi}{\mkbibemph{#1}}}
\DeclareFieldFormat[article]{title}{\usebibmacro{string+doi}{\mkbibquote{#1}}}
\DeclareFieldFormat[incollection]{title}{\usebibmacro{string+doi}{\mkbibquote{#1}}}

\newcommand{\bC}{\field{C}}
\newcommand{\bZ}{\field{Z}}
\newcommand{\bN}{\field{N}}
\newcommand{\bF}{\field{F}}
\newcommand{\bR}{\field{R}}

\newcommand{\cA}{\mathcal{A}}

\newcommand{\cD}{\mathcal{D}}

\newcommand{\cG}{\mathcal{G}}
\newcommand{\cH}{\mathcal{H}}

\newcommand{\cM}{\mathcal{M}}
\newcommand{\cN}{\mathcal{N}}

\newcommand{\cP}{\mathcal{P}}

\newcommand{\cS}{\mathcal{S}}

\newcommand{\ox}{\otimes}
\newcommand{\one}{\mathbbm{1}}
\DeclareMathOperator{\tr}{tr}
\newcommand{\id}{\mathrm{id}}
\newcommand{\eps}{\varepsilon}
\DeclareMathOperator{\supp}{supp}
\newcommand{\LC}{\mathrm{LC}}
\newcommand{\bX}{\bar{X}}
\newcommand{\barZ}{\bar{Z}}
\newcommand{\bz}{\bar{0}}
\newcommand{\bo}{\bar{1}}

\newcommand{\ii}{\mathbbm i}
\newcommand{\s}[1]{\mathrm{#1}}

\usepackage{xspace,xparse}
\usepackage{graphicx,graphbox}
\usepackage{grffile}

\DeclareDocumentCommand{\graph}{s m O{\textwidth}}{%
\IfValueTF{#1}{%
\includegraphics[align=c,width=#3]{graphs/#2.pdf}
}{%
\includegraphics[width=#3]{graphs/#2.pdf}
}
}

\DeclareDocumentCommand{\$}{m}{\texttt{#1}\xspace}
\usepackage{xcolor,tikz}
\usepackage{subcaption}

\usetikzlibrary{shapes}

\DeclareDocumentCommand{\surfacePlot}{ m m O{1} O{1} o }{%
\begin{tikzpicture}[
    x=#1,y=#1,
    label/.style={
        scale=#3,
        circle,
        fill=white,
        fill opacity=.4,
        text opacity=1,
        draw=none
    },
    plotlabel/.style={
        scale=#3,
        rectangle,
        fill=white,
        fill opacity=.4,
        text opacity=1,
        draw=none,
        align=center
    }
]
    \node at (0,0) {\includegraphics[width=#1]{#2}};
    \node[label] at (.045,.46) {$z$};
    \node[label] at (-.43,-.205) {$x$};
    \node[label] at (.43,-.205) {$y$};
    \IfValueT{#5}{%
        \node[plotlabel] at (0,-0.35) {#5};
    }
    #4
\end{tikzpicture}
}
\DeclareDocumentCommand{\subplotM}{ m m O{1} O{vs.single.png} o }{%
\begin{subfigure}[b]{#2\textwidth}
    \surfacePlot{\textwidth}{graphs/#1.#4}[#3][1][#5]
\end{subfigure}
}
\DeclareDocumentCommand{\subplotA}{ m o }{%
\subplotM{#1}{0.32}[.9][vs.single.png][\IfValueTF{#2}{#2}{\${#1}}]
}
\DeclareDocumentCommand{\subplotAA}{ m o }{\subplotM{#1}{0.32}[.9][png][#2]}
\DeclareDocumentCommand{\subplotAAsmall}{ m o }{\subplotM{#1}{0.24}[.8][png][#2]}
\DeclareDocumentCommand{\subplotB}{ m o }{%
\subplotM{#1}{0.15}[.7][vs.single.png][\IfValueTF{#2}{#2}{\${#1}}]
}
\DeclareDocumentCommand{\subplotCC}{ m o }{\subplotM{#1}{0.48}[.9][png][#2]}

\DeclareDocumentCommand{\subplotD}{ o m }{%
\subplotM{#2}{0.18}[.7][vs.single.png][\IfValueTF{#1}{#1}{\${#2}}]
}

\DeclareRobustCommand{\smallColorScale}{%
    \vspace*{-2mm}\hspace*{-1mm}\begin{tikzpicture}[x=179.5,y=179.5,rotate=270,transform shape]
    \node[] at (-.6,0) {\includegraphics[height=7cm,width=0.3cm]{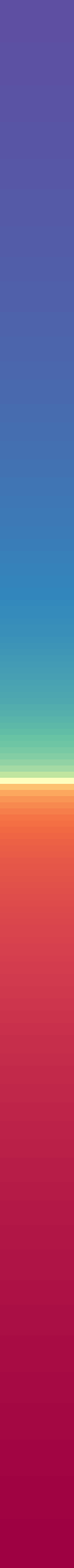}};
        \foreach \i/\l in {-4/-0.1,-3/,-2/-0.05,-1/,0/0,1/,2/0.05,3/,4/0.1}{
            \draw (-0.64,\i*.552/4) -- (-0.61,\i*.552/4) node[overlay,above,rotate=90,yshift=5] {\footnotesize{$\l$}};
        }
    \end{tikzpicture}
}

\usepackage{pgfplots}

\usepackage[section]{placeins} 
\usepackage{setspace}

\bibliography{references}

\setkomafont{title}{\normalfont\bfseries}
\title{Error Thresholds for\\ Arbitrary Pauli Noise}

\usepackage{authblk}

\author[1]{Johannes Bausch}
\affil[1]{CQIF, DAMTP, University of Cambridge, UK. \texttt{jkrb2@cam.ac.uk}}
\author[2,3]{Felix Leditzky}
\affil[2]{Department of Mathematics \& IQUIST, University of Illinois at Urbana-Champaign, USA. \texttt{leditzky@illinois.edu}}
\affil[3]{JILA \& CTQM, University of Colorado Boulder, USA.}
\usepackage{etoolbox}
\patchcmd\abstract{\small}{\normalfont}{}{}

\date{\today}

\definecolor{envcolor}{rgb}{0,0.72,0.5}

\newcommand\sys{_\mathrm{A}}
\newcommand\env{_\mathrm{R}}
\newcommand\tot{_\mathrm{AR}}
\newcommand\systosys{_{\mathrm A \shortrightarrow \mathrm A}}
\newcommand\systoenv{_{\mathrm A \shortrightarrow \mathrm R}}
\newcommand\envtosys{_{\mathrm R \shortrightarrow \mathrm A}}
\newcommand\envtoenv{_{\mathrm R \shortrightarrow \mathrm R}}

\DeclareMathOperator{\Aut}{Aut}
\DeclareMathOperator{\Stab}{Stab}
\DeclareMathOperator{\Can}{Can}

\newcommand\field{\mathds}
\newcommand\alg{\textsc}
\newcommand\tool{\texttt}

\DeclareDocumentCommand{\modfrac}{m m}{%
    \raisebox{3pt}{$#1$\!}\big/\raisebox{-3pt}{$#2$}
}

\begin{document}
\thispagestyle{empty}
\maketitle

\begin{abstract}
    The error threshold of a one-parameter family of quantum channels is defined as the largest noise level such that the quantum capacity of the channel remains positive.
    This in turn guarantees the existence of a quantum error correction code for noise modeled by that channel.
    Discretizing the single-qubit errors leads to the important family of Pauli quantum channels;
    curiously, multipartite entangled states can increase the threshold of these channels beyond the so-called hashing bound, an effect termed superadditivity of coherent information.
    In this work, we divide the simplex of Pauli channels into one-parameter families and compute numerical lower bounds on their error thresholds.
    We find substantial increases of error thresholds relative to the hashing bound for large regions in the Pauli simplex corresponding to biased noise, which is a realistic noise model in promising quantum computing architectures.
    The error thresholds are computed on the family of graph states, a special type of stabilizer state.

    In order to determine the coherent information of a graph state, we devise an algorithm that exploits the symmetries of the underlying graph resulting in a substantial computational speed-up.
    This algorithm uses tools from computational group theory and allows us to consider symmetric graph states on a large number of vertices.
    Our algorithm works particularly well for repetition codes and concatenated repetition codes (or cat codes), for which our results provide the first comprehensive study of superadditivity for arbitrary Pauli channels.
    In addition, we identify a novel family of quantum codes based on tree graphs.
    The error thresholds of these tree graph states outperform repetition and cat codes in large regions of the Pauli simplex, and hence form a new code family with desirable error correction properties.
\end{abstract}

\clearpage
\linespread{1.00} 
\section{Introduction and Summary of Results}
\subsection{Background}

It is widely believed that quantum error correction is a necessary requirement for quantum computers to maintain coherence in computations.
Quantum error-correcting codes protect logical quantum information from environmental noise by encoding it in larger physical systems in a redundant way.
Based on the particular architecture of the quantum computing device different mathematical models, or \emph{quantum channels}, are used to describe this environmental noise.
The analysis of quantum error correction is greatly simplified by the fact that the full unitary group of possible errors affecting a physical qubit can be discretized to the group of Pauli errors \cite{Shor95,Steane96}: a bit flip $X$, a phase flip $Z$, or a combined bit-phase flip $Y$, where $X$, $Y$ and $Z$ are the well-known Pauli matrices.
The most extensively studied noise model is depolarizing noise, in which the three Pauli errors occur with equal probabilities.
However, in many promising quantum computing architectures such as superconducting qubits \cite{superconducting}, quantum dots \cite{quantumdots}, or trapped ions \cite{trappedions}, the actual noise model is biased towards dephasing noise consisting of only one type of Pauli error.
These channels are examples of so-called Pauli channels (see \Cref{eq:pauli-channel-main}), on which we focus in this paper.

A common assumption in the noise model is that each physical qubit is independently affected by the same type of noise modeled by a quantum channel, that is, no memory effects between qubits are considered.
In this model, commonly referred to as the independent and identically distributed (i.i.d.) noise model, the ability to protect a quantum system (and the quantum information stored in it) from environmental noise is quantified by the \emph{quantum capacity} $Q(\cN)$ of a quantum channel $\cN$.
Positive quantum capacity of a particular quantum channel implies that error correction against the i.i.d.~noise modeled by that channel is in principle possible.
For a one-parameter family $x\mapsto \cN_x$ of quantum channels (such as depolarizing noise above), this leads to the question of determining the \emph{error threshold} (or noise threshold), defined as the supremum over all noise levels $x$ such that the quantum capacity remains positive,  $Q(\cN_x)>0$.
From the perspective of quantum error correction, an error threshold of $x_0$ for the i.i.d.~noise model defined by $x\mapsto \cN_x$ means that for all noise levels $x\leq x_0$ there exists a quantum error-correcting code whose decoding error can be made arbitrarily small by increasing the blocklength of the code.

The quantum capacity theorem \cite{Sch96,SN96,Llo97,Sho02,Dev05} gives a regularized entropic formula for the quantum capacity.
One consequence of this theorem is that positive quantum capacity of a channel can be certified by finding an entangled state that retains coherence after part of it has been affected by the noise.
Here, the coherence is measured by an entropic quantity called the \emph{coherent information}.
Evaluating the coherent information for one copy of a Pauli channel leads to a quantum capacity threshold given by the so-called \emph{hashing bound} (\cite{BDSW96}; see \Cref{sec:numerical-methods} for a discussion).
There are examples of topological quantum error-correcting codes achieving this hashing bound for certain Pauli channels \cite{fujii2012error}.
However, in contrast to Shannon's classical theory of information \cite{Sha48}, it does \emph{not} suffice to compute the coherent information with respect to one copy of the channel, since entangled multipartite quantum states acted upon by i.i.d.~copies of a quantum channel may maintain more coherence \cite{SS96,DSS98,SS07,FW08,CEMOPS15,LLS18}.
This effect is called \emph{superadditivity} of coherent information, and leads to an increased threshold of certain Pauli channels beyond the hashing bound.
The superadditivity is a direct consequence of the multipartite entanglement present in the quantum state.
However, finding suitable multipartite entangled states is challenging due to the exponential scaling of the dimension of the underlying Hilbert space.

\subsection{Summary of Main Results}
In this work we study the quantum capacity thresholds of qubit Pauli channels 
\begin{align}
    \rho\mapsto p_0 \rho + p_1 X\rho X + p_2 Y\rho Y + p_3 Z\rho Z
    \label{eq:pauli-channel-main}
\end{align}
where $(p_0,p_1,p_2,p_3)$ is a probability distribution, corresponding to the Pauli error model in terms of the Pauli matrices $X,Y,Z$ introduced above.
To this end, we evaluate the coherent information of a Pauli channel acting on graph states \cite{Briegel01,Hein04}, a special type of stabilizer states \cite{Got97}.
We find large threshold increases in the entire Pauli channel simplex, in particular for biased noise models where one type of Pauli error dominates over the other ones.
This noise regime is of particular interest as many proposed architectures for quantum computation---such as superconducting qubits \cite{superconducting}, quantum dots \cite{quantumdots}, or trapped ions \cite{trappedions,Ospelkaus2008}---have noise biased towards one of the Pauli errors.
We identify a novel family of quantum codes defined in terms of tree graphs with superadditive coherent information and an increased error threshold.
In large regions of the Pauli channel simplex these new codes outperform the known families of (concatenated) repetition codes.

To obtain our results, we make use of the special error-correcting properties of graph states that simplify the computation of the coherent information under Pauli noise \cite{HDB05}.
We devise an algorithm that exploits the symmetries of the underlying graph, providing a substantial speed-up in the computation of the coherent information for suitable graphs.
The algorithm relies on tools from computational group theory, in particular strong generating sets for permutation groups, and makes heavy use of the theory of group actions.
The resulting speed-up allows us to consider quantum states on a large number of constituent systems.
Furthermore, we obtain algebraic expressions for the coherent information in terms of the probability distribution $(p_0,p_1,p_2,p_3)$ characterizing the Pauli channel in \eqref{eq:pauli-channel-main}, which can then easily be evaluated numerically for the entire Pauli channel simplex.

Our algorithm is particularly well-suited for the class of repetition codes and concatenated repetition codes (or \emph{cat codes}).
These codes are examples of \emph{degenerate} codes where multiple errors are mapped to the same syndrome, and have long been known to increase the thresholds of certain Pauli channels such as the depolarizing channel or the independent bit and phase flip channel \cite{SS96,DSS98,SS07,FW08}.
Our paper comprises the first comprehensive study of the thresholds of these degenerate codes for \emph{arbitrary} Pauli channels.
We also identify a novel type of quantum codes given by tree graph states.
These codes outperform the best repetition and cat codes in large regions in the Pauli channel simplex, and provide new examples of quantum codes with favorable error-correcting properties.

The speed-up in computing the coherent information of graph states resulting from exploiting graph state symmetries is significant for the code families mentioned above.
For e.g.\ repetition codes, the runtime is reduced from $\Omega(4^k)$ to $\mathrm O(k^4)$ in the number of system qubits $k$, and similar speed-ups are obtained for cat codes and tree codes.

While superadditivity of coherent information has been studied for specific examples of Pauli channels such as unbiased depolarizing noise, the BB84 channel (corresponding to independent bit and phase flips), and 2-Pauli noise (where only two out of the three possible Pauli errors occur) \cite{DSS98,SS07,FW08}, we consider arbitrary Pauli channels in the whole channel simplex and find codes based on graph states with much more significant superadditivity; these findings are summarized in  \Cref{tab:summary-num-thresholds} in \Cref{sec:results} and \Cref{fig:summary-thresholds} in \Cref{sec:discussion}.
Furthermore, our analysis allows us to determine what type of code (or graph family) performs best for a given i.i.d.\ Pauli noise model, as indicated by the colored regions in \Cref{fig:code-familiy-comparison} in \Cref{sec:discussion}.

In summary, our results provide a comprehensive study of error thresholds for Pauli channels, and we identified a new code family with desirable error-correcting properties based on tree graphs.
Further use of our techniques to exploit graph symmetries may lead to the design of more examples of interesting error-correcting codes for realistic biased Pauli noise, which is typically found in promising noisy intermediate-scale quantum (NISQ) devices---we refer to \Cref{sec:nisq} for a discussion of the relevance of our results for this regime.

\subsection{Structure of this Paper}

Our paper is structured as follows.
\Cref{sec:preliminaries} contains necessary background information about the quantum capacity of a quantum channel, superadditivity of coherent information, stabilizer and graph states, and their coherence properties under Pauli noise.
In \Cref{sec:symmetries} we derive our main technical tool, an algorithm to compute the coherent information of graph states under Pauli noise that exploits the symmetries of the graph (\Cref{alg:is-canonical,alg:canonical-colorings,alg:symm-lambda}).
In \Cref{sec:numerical-methods} we describe the numerical methods used throughout our study.
\Cref{sec:results} contains the main results of this paper, a detailed analysis of the error thresholds of Pauli channels using graph states, and an analysis of rates close to the noise threshold.
We discuss repetition codes (\Cref{sec:rep-codes}), concatenated repetition codes (\Cref{sec:cat-codes}), and a novel code family based on tree graphs (\Cref{sec:tree-codes}).
We also carry out an exhaustive search over connected graphs of up to eight vertices in \Cref{sec:exhaustive}.
We conclude in \Cref{sec:discussion} with a discussion of our results and how they are relevant for realistic noise models in near-term quantum devices, and we mention open questions that are subject to further study.

%
%
%

\section{Preliminaries}\label{sec:preliminaries}

\subsection{Notation}

We briefly define the notation used throughout this paper.
A quantum system $A$ is associated with a finite-dimensional Hilbert space $\cH_\s{A}$, and multipartite systems $ABC\dots$ are described by the tensor product $\cH_\s{A}\ox\cH_\s{B}\ox\cH_\s{C}\ox\dots$ of the respective Hilbert spaces.
We denote by $|A|=\dim\cH_\s{A}$ the dimension of a quantum system $A$.
A quantum state (or density operator) $\rho_\s{A}$ on $A$ is a linear positive semidefinite operator with unit trace.
A pure quantum state $\psi_\s{A}$ is a quantum state of rank 1, which can be associated with a normalized vector (or ket) $\ket\psi_\s{A}\in\cH_\s{A}$ such that $\psi_\s{A} = \ketbra \psi_\s{A}$.
A quantum channel $\cN\colon A\to B$ is a linear, completely positive and trace-preserving map from the algebra of linear operators on $\cH_\s{A}$ to the algebra of linear operators on $\cH_\s{B}$.
We denote by $\one_\s{A}$ the identity operator on $\cH_\s{A}$, and by $\id_\s{A}\colon A\to A$ the identity map on the algebra of linear operators on $\cH_\s{A}$.
All logarithms are taken to base $2$.
For $n\in\bN$, we use the notation $[n]\coloneqq \lbrace 1,\dots,n\rbrace$.
For a vertex set $V=[n]$ and a subset $U\subset V$, the notation $O^U$ indicates that the single-qubit operator $O$ acts on the vertices contained in $U$, e.g., for $U=\{ 2,3\}$ and $|V|=5$ we have $O^U = \one_1\ox O_2\ox O_3\ox \one_4 \ox \one_5$.
For an arbitrary set $X$ we identify subsets $Y\subset X$ with binary vectors $\tilde{X}\in\bF_2^{|X|}$, where $\tilde{X}_i = 1$ if $i\in X$ and $0$ otherwise.
In slight abuse of notation, we use the same symbol for subsets and binary vector representations.

\subsection{Quantum Capacity}\label{sec:quantum-capacity}
In information-theoretic terms, the quantum capacity $Q(\cN)$ of a quantum channel $\cN\colon A\to B$ characterizes the channel's ability to faithfully transmit quantum information.
Here, we define it operationally through the task of entanglement generation \cite{Dev05}.
In this task, Alice prepares a bipartite state $\rho_{\s{R}\s{A}^n}$ in her lab, and sends the $\s{A}^n$-part to Bob through $n$ copies of the quantum channel $\cN$. 
Bob applies a local decoding operation $\cD$ to his share of the state in order to bring the total state $(\id\env \ox \cD\circ \cN^{\ox n})(\rho_{\s{R}\s{A}^n})$ close to a maximally entangled state $\Phi^+$ up to some error $\eps$ defined in terms of a suitable distance measure.
Informally, the quantum capacity is defined as the (logarithm of the) largest Schmidt Rank of $\Phi^+$ per channel copy such that $\eps$ vanishes asymptotically in the limit $n\to\infty$.

Formally, let $\eps\in(0,1)$ and $\ket{\Phi^+}_\s{RR'} = |R|^{-1/2} \sum_{i=0}^{|R|-1} \ket i\env\ket{i}_\s{R'}$ be a maximally entangled state on a bipartite quantum system.
An \emph{$\eps$-entanglement generating code} for a quantum channel $\cN\colon A\to B$ is a pair $(\rho_\s{RA},\cD)$ consisting of an arbitrary quantum state $\rho_\s{RA}$ and a decoding quantum channel $\cD\colon B\to R'$ such that
\begin{align*}
F\left(\Phi^+_\s{RR'}, (\id_\s{R} \ox \cD\circ \cN)(\rho_\s{RA}) \right) \geq 1-\eps,
\end{align*}
where $F(\rho,\sigma) = \|\sqrt{\rho}\sqrt{\sigma}\|_1$ denotes the fidelity of two quantum states $\rho$ and $\sigma$.
Let $r(n,\eps)$ denote the largest dimension $|R|$ of $R$ such that there exists an $\eps$-entanglement generating code for $\cN^{\ox n}$.
Then the quantum capacity $Q(\cN)$ is defined as
\begin{align*}
Q(\cN) = \inf_{\eps> 0} \liminf_{n\to\infty} \frac{1}{n}\log r(n,\eps).
\end{align*}

A positive quantum capacity means that it is theoretically possible to faithfully transmit quantum information through the channel at a positive rate.
Alternatively, $Q(\cN)>0$ means that a quantum system can be protected from environmental noise modeled by $\cN$ assuming that the noise acts independently and identically distributed (i.i.d.) on each qubit.
Both points of view lead to the fundamental questions of determining the value of $Q(\cN)$ for a given noise model $\cN$, and more generally whether we have $Q(\cN)>0$.

In information-theoretic terms, these questions have found a partial answer in the following \emph{coding theorem} providing a regularized entropic formula for the quantum capacity \cite{Sch96,SN96,Llo97,Sho02,Dev05}:
\begin{align}
Q(\cN) = \sup_{n\in\bN} \frac{1}{n} I_c(\cN^{\ox n}),\label{eq:quantum-capacity}
\end{align}
where the \emph{channel coherent information} is defined as
\begin{align}
I_c(\cN) &= \max_{\ket{\psi}_\s{RA}} I_c(\psi,\cN) = \max_{\ket{\psi}_\s{RA}} \left[ S(\cN(\psi\sys)) - S(\cN(\psi_\s{RA})) \right],
\label{eq:coherent-information}
\end{align}
with $S(\rho)=-\tr\rho\log\rho$ denoting the von Neumann entropy of a density operator $\rho$.
In this paper, we refer to the pure input state $\ket{\psi}_\s{RA}$ in \Cref{eq:coherent-information} simply as a \emph{quantum code}.\footnote{\label{fn:quantum-code}The quantum state $\psi_\s{RA}$ appearing on the right-hand side of \Cref{eq:coherent-information} is directly related to an actual entanglement generation code as follows: Assume that the state $|\psi_n\rangle_\s{RA^n}\in \cH_\s{R}\otimes \cH_\s{A}^{\otimes n}$ has positive coherent information with respect to $n$ copies of the channel $\cN$, i.e., $I_c(\psi_n,\cN^{\otimes n}) > 0$.
To obtain an entanglement generation code, Alice prepares $k$ independent and identically distributed (i.i.d.)~copies of $\psi_n$, for each copy sharing the $\s{A}^n$ systems with Bob through the channel $\cN^{\otimes n}$, after which they share the state $\sigma_{\s{RB}^n}^{\otimes k} = \left[(\id_\s{R}\otimes \cN^{\otimes n})(\psi_n)\right]^{\otimes k}$. The parameters $n$ and $k$ are referred to as the \emph{inner} and \emph{outer} blocklength of the code, respectively. Provided that the outer blocklength $k$ is large enough, an entanglement generation code with rate arbitrarily close to $\frac{1}{n}I_c(\psi_n,\cN^{\otimes n})$ and arbitrarily small error can be obtained by suitably choosing a random subspace in a typical subspace defined in terms of the shared state $\sigma_{\s{RB}^n}^{\otimes k}$ \cite{HHWY08}.}

The regularization in \Cref{eq:quantum-capacity} is necessary since the channel coherent information is known to be superadditive: there exist channels $\cN$ and quantum codes $\psi_{\s{R}\s{A}^n}$ such that \cite{SS96,DSS98,SS07,FW08,CEMOPS15,LLS18}
\begin{align}
\frac{1}{n}I_c(\psi,\cN^{\ox n}) > I_c(\cN).
\label{eq:superadditivity}
\end{align}

A well-known quantum channel exhibiting superadditivity of coherent information as in \Cref{eq:superadditivity} is the qubit depolarizing channel \cite{SS96,DSS98}, defined for $p\in[0,1]$ as
\begin{align}\label{eq:depolarizing-channel}
\cD_p(\rho) = (1-p)\rho + \frac{p}{3}(X\rho X + Y\rho Y + Z\rho Z),
\end{align}
where $X,Y,Z$ are the Pauli matrices.
The depolarizing channel belongs to the class of \emph{Pauli channels} $\cN_{\mathbf{p}}$, defined in terms of a probability distribution $\mathbf{p} = (p_0,p_1,p_2,p_3)$ as
\begin{align}
\cN_{\mathbf{p}}(\rho) = p_0 \rho + p_1 X\rho X + p_2 Y\rho Y + p_3 Z\rho Z.\label{eq:pauli-channel}
\end{align}
Superadditivity of coherent information has also been demonstrated for other channels in this class, such as the BB84 channel with $\mathbf{p} = ((1-p)^2,p-p^2,p^2,p-p^2)$ or the Two-Pauli channel with $\mathbf{p} = (1-p,p/2,0,p/2)$ \cite{SS07,FW08}.
One of the main goals of this paper is a comprehensive study of superadditivity for the class of Pauli channels.
We note here that for dephasing-type noise of the form $\cN^D_p(\rho)= (1-p)\rho + p D\rho D$ with $D\in\lbrace X,Y,Z\rbrace$, the coherent information is additive and hence equal to the quantum capacity \cite{DS05}.

Superadditivity of the coherent information renders the quantum capacity of general quantum channels intractable to compute.
As a result, the quantum capacity of Pauli channels is unknown except for a few special cases, and one needs to resort to finding lower and upper bounds on $Q(\cN)$, e.g., in order to certify that $Q(\cN) >0$.
A helpful tool in this task is the direct part of the coding theorem \Cref{eq:quantum-capacity}, which states that
\begin{align}
Q(\cN) \geq \max_{\ket{\psi}_{\s{R}\s{A}^n}} \frac{1}{n} I_c(\psi_{\s{R}\s{A}^n},\cN^{\ox n})\quad\text{for any $n\in\bN$}.\label{eq:coh-info-lower-bound}
\end{align}
Note that this inequality, a direct consequence of \Cref{eq:quantum-capacity,eq:coherent-information}, allows us to obtain lower bounds on the quantum capacity of a channel by testing single input states or quantum codes. 
This is because a particular quantum code achieving positive coherent information on the right-hand side of \Cref{eq:coh-info-lower-bound} gives rise to an entanglement generation code via the procedure outlined in \Cref{fn:quantum-code}, and thus provides a lower bound on the quantum capacity equal to the best entanglement generation rate achievable for this particular channel.
In particular, for a one-parameter family $x\mapsto \cN_x$ of channels, \Cref{eq:coh-info-lower-bound} allows us to find a lower bound on the \emph{threshold} $x_0$ defined in the introduction, i.e., the supremum over all $x$ such that $Q(\cN_x) >0$.

Unfortunately, solving the optimization problem on the right-hand side of \Cref{eq:coh-info-lower-bound} is generally infeasible due to the exponential growth in $n$ of the number of parameters of $\psi_{\s{R}\s{A}^n}$.
Instead, one may try to restrict the optimization to a smaller class of multipartite quantum states $\ket{\psi}_{\s{R}\s{A}^n}$ that is numerically tractable.
For example, in our previous paper \cite{BL18} we optimized the coherent information in \Cref{eq:coh-info-lower-bound} over neural network states, a class of multipartite quantum states with polynomially many degrees of freedom \cite{CT17}.
For sufficiently low $n\in\bN$, this ansatz proved successful in efficiently approximating the channel coherent information for noise models such as the generalized amplitude damping channel \cite{myatt2000decoherence,turchette2000,chirolli2008,zou2017} or 
the dephrasure channel \cite{LLS18}.

In this work, we focus on the class of Pauli channels as defined in \Cref{eq:pauli-channel}.
The form of these channels naturally suggests a restriction in the maximization in \Cref{eq:coh-info-lower-bound} to the family of \emph{stabilizer states}, which are states defined in terms of their invariance properties under (an Abelian group of) many-qubit Pauli operators.
In addition, for a special type of stabilizer states called \emph{graph states} the coherent information \Cref{eq:coherent-information} turns into a ``classical'' quantity that can be computed more efficiently.
This is the starting point for our analysis of Pauli channels and explained in the next section, in which we first review stabilizer and graph states along with some of their useful properties.

\subsection{Graph States}\label{sec:graph-states}

Before we define and discuss graph states, the main object of study in this paper, we briefly introduce stabilizer groups and stabilizer states \cite{Got97}.

Let $\cP=\{ I,X,Y,Z\}$ be the set of Pauli matrices including the identity $I$ on $\bC^2$, and define the $n$-qubit Pauli group (or simply Pauli group)
\begin{align*}
\cP_n = \left\{ \ii^k \sigma_1\ox \dots \ox \sigma_n\colon \sigma_i\in\cP, k=0,1,2,3 \right\}.
\end{align*}
A \emph{stabilizer group} $\cS$ is an Abelian subgroup of $\cP_n$ that does not contain $-I^{\ox n}$.
Since operators in $\cS$ commute pairwise, they can be simultaneously diagonalized.
The \emph{stabilizer subspace} $\cH_\cS$ is defined as the joint eigenspace of the operators in $\cS$ corresponding to the eigenvalue 1: 
\begin{align*}
\cH_\cS = \left\{ \ket{\psi}\in (\bC^2)^{\ox n}\colon S \ket{\psi} = \ket{\psi} \text{ for all }S\in\cS \right\}.
\end{align*}
Let $S_1,\dots,S_r$ with $r\leq n$ be a (minimal) set of \emph{stabilizer generators} for the stabilizer group $\cS$.
Then it can be shown that $\dim \cH_\cS = 2^k$ where $k = n-r$. 
In the language of quantum error-correcting codes, the \emph{stabilizer code} $\cS$ (or $\cH_\cS$) encodes $k$ logical qubits into $n$ physical qubits, and is commonly denoted as $[n,k]$.\footnote{We will not be concerned with the \emph{distance} of a stabilizer code.}
Let $N_{\cP_n}(\cS) = \lbrace P\in\cP_n\colon P\cS P^\dagger \subset \cS\rbrace$ be the \emph{normalizer} of $\cS$ in $\cP_n$.
The stabilizer group $\cS$ is a normal subgroup of $N_{\cP_n}(\cS)$, and the factor group $\modfrac{ N_{\cP_n}(\cS) }{\cS}$ is isomorphic to $\cP_k$, i.e., one can choose $\bX_i,\barZ_i\in \modfrac{ N_{\cP_n}(\cS) }{\cS}$ for $i=1,\dots,k$ satisfying the usual relations for Pauli matrices.
The operators $\bX_i,\barZ_i$ are called \emph{logical operators}; they commute with all stabilizers in $\cS$, but act non-trivially on the code space $\cH_\cS$.
If $r = n$, then we may choose a pure state $\ket{\psi}$ such that $\cH_S=\operatorname{span}(\ket{\psi})$, and we refer to $\ket{\psi}$ as a \emph{stabilizer state}.

\subsubsection{Definition and Properties}\label{sec:graph-states-definition}
We now define graph states.
Let $\Gamma=(V,E)$ be a simple undirected graph with vertex set $V=\{ 1,\dots,n\}$ and edge set $E$, and identify each vertex with a qubit associated with a Hilbert space $\bC^2$.
The \emph{graph state} $\ket\Gamma \in (\bC^2)^{\ox n}$ is defined as
\begin{align*}
\ket\Gamma = \prod_{e\in E} CZ_e \ket+^{\ox n},
\end{align*}
where $\ket+=(\ket0+\ket1)/\sqrt{2}$ and $CZ_e$ denotes the controlled-$Z$ gate $\ketbra{0}_i\ox \one_j + \ketbra{1}_i\ox Z_j$ acting on the qubits $i$ and $j$ connected by the edge $e\in E$.
Equivalently, the graph state $\ket\Gamma$ can be regarded as a stabilizer state defined by the $n$ stabilizer generators 
\begin{align}
S_i = X^i Z^{N_i}\quad\text{for $i=1,\dots,n$},\label{eq:graph-state-stabilizers}
\end{align} 
where $N_i\coloneqq \{ k\in V\colon \{ k,i\} \in E\}$ is the neighborhood of the vertex $i\in V$ consisting of all vertices $k\in V$ adjacent to $i$.
In \Cref{eq:graph-state-stabilizers} and throughout the paper, the notation $O^U$ for a subset $U\subset V$ indicates that the single-qubit operator $O$ acts on the vertices contained in $U$, e.g., for $U=\{ 2,3\}$ and $|V|=5$ we have $O^U = \one_1\ox O_2\ox O_3\ox \one_4 \ox \one_5$.

By the above observation, every graph state is a stabilizer state.
Conversely, any stabilizer state is local-Clifford-equivalent to a graph state, which we show in detail in \Cref{app:LU-equivalence}.
Here, the set of \emph{local Clifford operations} is defined as those local unitaries $U_1\ox\dots\ox U_n$ mapping the Pauli group $\cP_n$ into itself.

The action of any local Clifford operation on a graph state can be decomposed into a series of so-called \emph{local complementations} (LC) \cite{vandennest04}.
The local complementation $\LC_i$ of a graph $\Gamma=(V,E)$ with respect to a vertex $i\in V$ is obtained by replacing the subgraph induced by the vertex neighborhood $N_i$ with its complement.
For example, the graphs \raisebox{-2pt}{\begin{tikzpicture}[scale=0.35]\draw[fill] (0,0) circle(0.5ex); \draw[fill] (1,0) circle(0.5ex); \draw[fill] (0,1) circle(0.5ex); \draw[fill] (1,1) circle(0.5ex); \draw (0,0) -- (1,0); \draw (0,0) -- (1,1); \draw (0,0) -- (0,1); \draw (0,1) -- (1,0); \draw (1,0) -- (1,1);\end{tikzpicture}}
and
\raisebox{-2pt}{\begin{tikzpicture}[scale=0.35]\draw[fill] (0,0) circle(0.5ex); \draw[fill] (1,0) circle(0.5ex); \draw[fill] (0,1) circle(0.5ex); \draw[fill] (1,1) circle(0.5ex); \draw (0,0) -- (1,0); \draw (0,0) -- (1,1); \draw (0,0) -- (0,1); \draw (0,1) -- (1,1);\end{tikzpicture}}
can be obtained from each other by a local complementation with respect to the bottom-left vertex.
If a quantum channel $\cN$ is \emph{covariant} with respect to local Clifford operations $C$, i.e.,
\begin{align*}
    \cN^{\ox n} (C \cdot C^\dagger) = C \cN^{\ox n} (\cdot) C^\dagger \qquad\text{ for all local Cliffords $C$},
\end{align*}
then two graph states with LC-equivalent graphs $\Gamma$ and $\Gamma' = \LC_i(\Gamma)$ have the same coherent information, $I_c(|\Gamma\rangle,\cN^{\ox n}) = I_c(|\Gamma'\rangle,\cN^{\ox n})$.
This holds in particular for the depolarizing channel $\cN_{\mathbf{p}}$ with $\mathbf{p}=(1-p,p/3,p/3,p/3)$, which is in fact covariant with respect to any local unitary operation.
Hence, a maximization of the coherent information of a depolarizing channel over all graph states corresponds to a maximization over all stabilizer states.

Graph states form a rich class of multipartite quantum states with many appealing properties; in the following we briefly collect the most important ones needed in this paper.
A comprehensive review of graph states and their properties can be found in \cite{graphstatesReview06}.

\paragraph{Graph State Basis.}
A helpful tool for working with graph states is the so-called \textit{graph state basis} for $(\bC^2)^{\ox n}$.
For a given graph $\Gamma=(V,E)$ with associated graph state $\ket{\Gamma}$, it is defined as $\mathcal{G}(\Gamma) = \{ \ket{U}\}_{U\subset V}$, where
\begin{align*}
\ket{U} &= Z^U \ket{\Gamma}.
\end{align*}
To see that $\mathcal{G}(\Gamma)$ is indeed a basis for $(\bC^2)^{\ox n}$, observe that for $U,U'\subset V$ with $U\neq U'$ we have $S=(U\cup U')\setminus (U\cap U')\neq \emptyset$.
Hence, the operator $O=Z^U Z^{U'}=Z^S$ has at least one $Z$ operator acting on some qubit, and consequently $\bra{\Gamma} O \ket{\Gamma} = 0$ \cite{Hel14}.
This can be checked by noting that $Z^i$ commutes with any controlled-$Z$ gate involving the qubit $i$, and $CZ^2=\one$. 
It then follows for $O=Z^S$ with $S\neq 0$ that 
\begin{align*}
\bra{\Gamma} O \ket{\Gamma} = \bra{+}^{\ox n} \left(\prod\nolimits_{e\in E} CZ_e\right) O \left(\prod\nolimits_{e\in E} CZ_e\right) \ket+^{\ox n}= \bra{+}^{\ox n} O \ket+^{\ox n}= 0.
\end{align*}

\paragraph{Partial Trace Formula.}
Consider a graph $\Gamma\tot$ on a vertex set $V=V\sys\cup V\env$ consisting of vertices (qubits) $V\sys$ and $V\env$ belonging to quantum systems $A$ and $R$, respectively.
We denote by $\Gamma\systosys$ the induced subgraph on $V\sys$ of $\Gamma\tot$.
There is a simple formula for the mixed marginal $\rho\sys = \tr\env \Gamma\tot$, which should not be confused with the pure subgraph state $\ket{\Gamma\systosys}$.
To state the formula, we consider a block form of the adjacency matrix of $\Gamma\tot$ (which we denote by the same symbol in slight abuse of notation):
\begin{align}
\Gamma\tot = \begin{pmatrix}
\Gamma\systosys & \Gamma\systoenv\\ \Gamma\systoenv^T & \Gamma\envtoenv
\end{pmatrix},
\label{eq:bipartite-adjacency-matrix}
\end{align}
where $\Gamma\systoenv$ is a $|V\sys|\times |V\env|$-matrix encoding the edges between the vertex sets $V\sys$ and $V\env$.
Throughout the paper, for an arbitrary set $X$ we identify subsets $Y\subset X$ with binary vectors $\tilde{Y}\in\bF_2^{|X|}$, where $\tilde{Y}_i = 1$ if $i\in Y$ and $0$ otherwise.
In slight abuse of notation, we henceforth use the same symbol for subsets and binary vector representations.
Abbreviating $\Gamma'=\Gamma\systoenv$, the following formula for the mixed marginal $\rho\sys$ on the vertices belonging to $A$ is proved in \cite{graphstatesReview06}:
\begin{align}
\rho\sys = \frac{1}{2^{|V\env|}} \sum_{R'\subset V\env} \ketbra{\Gamma' R'}\sys,\label{eq:graph-state-marginal}
\end{align}
where $\Gamma'R'\subset V\sys$ is understood as matrix multiplication modulo 2, and the state $\ket{\Gamma' R'}\sys =Z^{\Gamma'R'}\ket{\Gamma\systosys}$ is an element of $\mathcal{G}(\Gamma\systosys)$.
Equation \eqref{eq:graph-state-marginal} shows that the marginal $\rho\sys$ of a graph state $\ket{\Gamma\tot}$ is diagonal in a preferred basis, the graph state basis $\mathcal{G}(\Gamma\systosys)$ of the subgraph $\Gamma\systosys$.

\subsubsection{Decoherence under Pauli Noise}
We now review the decoherence properties of graph states under Pauli noise discussed in \cite{HDB05}.	
The crucial observation is that graph states translate any Pauli error into $Z$-type errors: If qubit $i$ is affected by an $X$ error, then
\begin{align*}
X^i \ket\Gamma = X^i S_i \ket\Gamma = X^i X^i Z^{N_i} \ket\Gamma = Z^{N_i} \ket\Gamma = \ket{N_i},
\end{align*}
where we used that $S_i=X^i Z^{N_i}$ stabilizes $\ket\Gamma$.
Furthermore,
\begin{align*}
Y^i \ket\Gamma = -\ii Z^i X^i \ket\Gamma = -\ii Z^i \ket{N_i} = -\ii \ket{N_i \oplus i},
\end{align*}
where $\oplus$ is addition modulo 2 and we identify a subset $U\subset V$ with its binary representation.

The above observation shows that under conjugation with a Pauli error $E\in\cP$ the graph state $\ketbra\Gamma$ is mapped onto some $\ketbra U$, where $\ket U$ is the element of the graph state basis $\mathcal{G}(\Gamma)$ corresponding to the subset $U\subset V$.
This holds more generally for any $n$-qubit Pauli group element $E_n\in\cP_n$.
Consider now a Pauli channel 
\begin{align*}
\cN_{\mathbf{p}}(\rho) = p_0 \rho + p_1 X\rho X + p_2 Y\rho Y + p_3 Z\rho Z,
\end{align*}
defined in terms of a probability distribution $\mathbf{p}=(p_0,p_1,p_2,p_3)$.
Then it follows from above that under the i.i.d.~action of a Pauli channel $\cN_{\mathbf{p}}$ the channel output $\sigma = \cN_{\mathbf{p}}^{\ox n}(\Gamma)$ is diagonal with respect to $\mathcal{G}(\Gamma)$, i.e., $\sigma = \sum_{U\subset V} \lambda_U \ketbra U$ for some probability distribution $\{ \lambda_U\}_{U\subset V}$.
A compact formula for the coefficients $\lambda_U$ is derived in \cite{HDB05}.

In this paper, we use a simple generalization of this formula to the situation where the (i.i.d.~action of the) Pauli noise only affects a subset of the qubits.
To this end, consider a graph $\Gamma\tot=(V\sys\cup V\env,E)$ whose vertices are partitioned into $k\equiv k\sys$ system vertices (qubits) $V\sys$ and $k\env$ environment vertices (qubits) $V\env$.
By the above discussion, the state $\sigma_\s{BR} = (\cN_{\mathbf{p}}^{\ox k} \ox \id\env)(\Gamma\tot)$ is again diagonal in the graph state basis $\mathcal{G}(\Gamma\tot)$,
\begin{align}\label{eq:lambda-state-sum}
\sigma_\s{BR} = \sum_{U\subset V\sys\cup V\env} \lambda_U \ketbra{U}.
\end{align}
Setting $q_i = p_i/p_0$ for $i=1,2,3$, we have the following formula for $\lambda_U$ with $U\subset V$ as a straight-forward generalization of the one reported in \cite{HDB05}:
\begin{align}\label{eq:lambda-vector}
\lambda_U = p_0^{k}\sum_{(U_1,U_2,U_3)\in \cM(U)} q_1^{|U_1|} q_2^{|U_2|} q_3^{|U_3|},
\end{align}
where $\cM(U)$ is the set of all triples $(U_1,U_2,U_3)$ satisfying
\begin{flalign}
    \begin{aligned}
    U_i &\subset V\sys &\text{for $i=1,2,3$,}\\
    U_i \cap U_j &= \emptyset &\text{for $1\leq i<j\leq 3$,}\\
    U &= \Gamma\tot (U_1 \oplus U_2) \oplus U_2 \oplus U_3.
    \end{aligned}
    \label{eq:lambdaU-U-condition}
\end{flalign}
Here and throughout the paper, $\oplus$ denotes addition mod 2.
Note that in \Cref{eq:lambdaU-U-condition} we interpret the $U_i$ as subsets of $V=V\sys\cup V\env$, which amounts to padding the binary vectors of length $k\sys$ with $k\env$ $0$'s.
The resulting subset $U$ has support in $V\env$ whenever there is an edge in $\Gamma\tot$ from a vertex in $U_1\oplus U_2$ to a vertex in $V\env$ (that is, the off-diagonal block $\Gamma\systoenv$ in \Cref{eq:bipartite-adjacency-matrix} has an entry $1$ in a row corresponding to a vertex in $U_1\oplus U_2$).

\subsubsection{The Partial Trace of a Decohered Graph State}\label{sec:lambda-state-tr-sum}

We now derive a formula for taking the partial trace of a graph state after the i.i.d.~action of a Pauli channel.
More precisely, let $\Gamma\tot$ be a graph on a vertex set $V=V\sys\cup V\env$ with $k\equiv k\sys=|V\sys|$ and $k\env=|V\env|$, and let $\cN_{\mathbf{p}}$ be a Pauli channel.
Then by the above discussion, we have 
\begin{align}
    \sigma_\s{BR} = \cN_{\mathbf{p}}^{\ox k}(\Gamma\tot) = \sum_{U\subset V} \lambda_U \ketbra U,\label{eq:sigmaRB}
\end{align}
where the coefficients $\lambda_U$ are given by \Cref{eq:lambda-vector}.

We now apply the partial trace formula \Cref{eq:graph-state-marginal} to \Cref{eq:sigmaRB}.
Abbreviating $\Gamma'=\Gamma\systoenv$, and for a subset $U\subset V=V\sys\cup V\env$ denoting by $U\sys$ its restriction to $V\sys$, we have:
\begin{align}
    \sigma_{\s{B}} &= \tr\env \sigma_\s{BR}\notag\\
    &= \sum_{U\subset V} \lambda_U \tr\env \ketbra{U}\notag\\
    &= \sum_{U\subset V} \lambda_U \tr\env Z^U\ketbra{\Gamma\tot}Z^U\notag\\
    &= \sum_{U\subset V} \lambda_U Z^{U\sys} \tr\env (\ketbra{\Gamma\tot}) Z^{U\sys}\notag\\
    &= \frac{1}{2^{k\env}} \sum_{\substack{U\subset V\\R'\subset V\env}} \lambda_U Z^{U\sys\oplus\Gamma'R'} \ketbra{\Gamma\systosys} Z^{U\sys\oplus\Gamma'R'},
    \label{eq:intermediate-partial-trace}
\end{align}
where we used \Cref{eq:graph-state-marginal} in the last step.
This can be rewritten as 
\begin{align}\label{eq:lambda-state-tr-sum}
\sigma_{\s{B}} = \sum_{W\subset V\sys} \lambda_W \ketbra{W},
\end{align}
where
\begin{align}\label{eq:lambda-vec-tr}
    \lambda_W &= \frac{1}{2^{k\env}} \sum_{(U,R')\in \cM(W)} \lambda_U,\\
    \cM(W) &= \{ (U,R')\colon U\subset V, R'\subset V\env, U\sys \oplus \Gamma'R' = W\}.
\end{align}
Similar to $\sigma_\s{BR}$ above, the operator $\sigma_{\s{B}}$ is again diagonal in a preferred basis, the graph state basis $\cG(\Gamma\systosys)$ of the subgraph $\Gamma\systosys$.
    
\subsection{Coherent Information of Decohered Graph States}\label{sec:decohered-graph-states-ci}

Recall from \Cref{sec:quantum-capacity} that the coherent information of a quantum code $\ket{\psi}\tot$ and a channel $\cN\colon A\to B$ is given by
\begin{align}
    I_c(\psi,\cN) = S(\cN(\psi\sys)) - S(\cN(\psi\tot)).\label{eq:coherent-information-alg}
\end{align}
Consider now a graph state $\ket{\Gamma\tot}$ defined in terms of a graph $\Gamma\tot = (V\sys\cup V\env,E)$ with $k\equiv k\sys = |V\sys|$ and $k\env = |V\env|$, and a Pauli channel $\cN_{\mathbf{p}}$.
The formulas \eqref{eq:lambda-state-sum} for $\sigma_\s{BR} = \cN_{\mathbf{p}}^{\ox k}(\Gamma\tot)$ and \eqref{eq:lambda-vec-tr} for $\sigma_{\s{B}} = \tr\env\sigma_\s{BR}$ show that these operators are diagonal with respect to the graph state bases $\cG(\Gamma\tot)$ and $\cG(\Gamma\systosys)$, respectively.
These observations yield a ``classical'' algorithm for computing the coherent information \Cref{eq:coherent-information-alg}, which is described in \Cref{alg:coherent-information-full-algorithm} and explained in detail in the following paragraph.
A MATLAB implementation of \Cref{alg:coherent-information-full-algorithm} can be found at \cite{CImatlab}.
A similar algorithm to efficiently compute the coherent information of Pauli channels acting on quantum codes diagonal in the graph state basis of a given graph was derived in \cite{graphstatebasis}.

To start \Cref{alg:coherent-information-full-algorithm}, we obtain in \Cref{line:getUsubsets} the subsets $U_i\subset V$ satisfying the conditions in \Cref{eq:lambdaU-U-condition} as matrices whose columns are the binary vector representations of $U_i$.
We store their cardinality $|U_i|$ in vectors $u_i$ in \Cref{line:getUsubsetsCard}, and determine the matrix $U=\Gamma(U_1\oplus U_2) \oplus U_2 \oplus U_3$ in \Cref{line:Usubsets}, assigning the first $k\sys$ rows of $U$ to a matrix $U_a$ in \Cref{line:UsubsetsA}.
The columns of $U$ and $U_a$ are binary vectors indexing subsets of $V$, and in \Cref{line:binary-reps} we convert them to decimal numbers indicating their corresponding position in coefficient vectors $\lambda$ and $\mu$ for the states $\sigma_\s{BR}=(\cN_{\mathbf{p}}^{\ox k}\ox \id\env)(\Gamma\tot)$ and $\omega_{\s{B}} = \cN_{\mathbf{p}}^{\ox k}(\Gamma\systosys)$, respectively.\footnote{Note that the state $\omega_{\s{B}}$, obtained from applying the channel to all qubits of the \emph{pure} graph state $\ket{\Gamma\systosys}$ on the vertices $V\sys$ corresponding to the subgraph $\Gamma\systosys$, is an auxiliary state that does not feature itself in the formula for the coherent information.}
In \Cref{line:coefficients} we compute a vector $c$ whose entries are the coefficients $p_0^k q_1^{|U_1|} q_2^{|U_2|} q_3^{|U_3|}$ appearing in \Cref{eq:lambda-vector}.
Note that $q_1\cdot q_2$ denotes entry-wise multiplication.
If one of the $q_i$ is zero, we eliminate it from the equation altogether in order to avoid computing $0^0$.
For each element in $c$, we then add up the contributions to the coefficient $\lambda_U$ indexed by the subset $U\subset V$ in \Cref{line:lambda-loop}.
In the same loop, we also add up the contributions to the coefficient of $U\sys$ in the marginal $\sigma_{\s{B}}$ in \Cref{line:lambdaA-loop}.
This procedure computes the eigenvalues $\lambda$ and $\mu$ of the states $\sigma_\s{BR}=(\cN_{\mathbf{p}}^{\ox k}\ox \id\env)(\Gamma\tot)$ and $\omega_{\s{B}} = \cN_{\mathbf{p}}^{\ox k}(\Gamma\systosys)$, which are diagonal in their respective graph state bases $\cG(\Gamma\tot)$ and $\cG(\Gamma\systosys)$.

In the second part of \Cref{alg:coherent-information-full-algorithm} we compute the marginal $\sigma_{\s{B}}$ from the intermediate state $\omega_{\s{B}}$ based on the following observation: using \Cref{eq:intermediate-partial-trace}, we can write $\sigma_{\s{B}}$ as
\begin{align*}
    \sigma_{\s{B}} &= \frac{1}{2^{k\env}} \sum_{R'\subset V\env} Z^{\Gamma' R'} \left( \sum_{U\subset V} \lambda_U Z^{U\sys} \ketbra{\Gamma\sys} Z^{U\sys} \right) Z^{\Gamma' R'}\\
    &= \frac{1}{2^{k\env}} \sum_{R'\subset V\env} Z^{\Gamma' R'} \left( \sum_{U\subset V} \lambda_U \ketbra{U\sys}\right)Z^{\Gamma' R'}\\
    &= \frac{1}{2^{k\env}} \sum_{R'\subset V\env} Z^{\Gamma' R'} \omega_{\s{B}} Z^{\Gamma' R'}.
\end{align*}
We have
\begin{align*}
\ketbra{U\sys}\mapsto Z^{\Gamma'R'}\ketbra{U\sys}Z^{\Gamma'R'} = \ketbra{U\sys \oplus \Gamma'R'},
\end{align*}
which amounts to a permutation $A'\mapsto A' \oplus \Gamma'R'$ of the subsets $A'\subset V\sys$ and hence a permutation of the coefficients $\lambda\sys$ of $\sigma_{\s{B}}$ indexed by $A'\subset V\sys$.
To use this observation, in \Cref{line:A-subsets,line:R-subsets} we define $A'$ and $R'$ to be matrices whose columns are the subsets of $V\sys$ and $V\env$ in binary vector representation, respectively.
Furthermore, in \Cref{line:Delta} we define $\Delta$ to be the subsets of $V\sys$ obtained from moving the $V\env$-subsets $R'$ into $V\sys$ via $\Gamma'=\Gamma\systoenv$, i.e., along the edges connecting vertices in $V\sys$ to vertices in $V\env$.

For a fixed column $\Delta_j$ (corresponding to a particular subset of $V\env$ moved into $V\sys$), we then compute the map $A' \mapsto \Delta_j \oplus A'$ for all $A'\subset V\sys$ in \Cref{line:R-permutation}, and convert the resulting matrix to a decimal representation of the corresponding subsets in $V\sys$ in \Cref{line:convert-to-permutation}.
The result of this operation is a vector $V_{\mathrm{ind}}$ which can be interpreted as a permutation in $S_k$ with $k= k\sys=|V\sys|$.
The corresponding permutations of the eigenvalue vector $\mu$ of $\omega_{\s{B}}$ are then added to $\lambda\sys$ (the eigenvalue vector of $\sigma_{\s{B}}$) in \Cref{line:add-permuted-vectors}.
Finally, in \Cref{line:compute-entropies} we compute the Shannon entropies of the resulting eigenvalue vectors $\lambda\sys$ and $\lambda$ to obtain the coherent information of $\sigma_\s{BR}$.

While \Cref{alg:coherent-information-full-algorithm} is a ``classical'' algorithm computing the probability distributions corresponding to the diagonal states $\sigma_\s{BR}$ and $\sigma_{\s{B}}$, it nevertheless has exponential scaling due to the loops over subsets of $V\sys$ and $V\env$.
In the next section, we deal with this problem by exploiting the symmetries in the underlying graph $\Gamma\tot$ with the goal of speeding up the computation of the coherent information.

\begin{algorithm}[t]
\begin{algorithmic}[1]
\Function{$I_c$}{$\ket{\Gamma},\cN_{\mathbf{p}}^{\ox k}$}
\State $q_i \gets p_i/p_0$ for $i=1,2,3$ \label{line:q_i}
\State $(U_1,U_2,U_3) \gets \Call{GetUSubsets}{k,k\env}$ \label{line:getUsubsets}
\State $(u_1,u_2,u_3) \gets \Call{GetUSubsetsCard}{k}$ \label{line:getUsubsetsCard}
\State $U \gets \Gamma(U_1 \oplus U_2) \oplus U_2 \oplus U_3$ \label{line:Usubsets}
\State $U_a \gets U_{1:k,:}$ \Comment{$U_{1:k,:}$ denotes the submatrix of the first $k$ rows and all columns.} \label{line:UsubsetsA}
\State $v \gets \Call{BinaryToDecimal}{U}$, $w \gets \Call{BinaryToDecimal}{U_a}$ \label{line:binary-reps}
\State $c \gets p_0^k \left( q_1^{u_1} \cdot q_2^{u_2} \cdot q_3^{u_3}\right)$ \label{line:coefficients}
\State $\lambda \gets 0^{2^{k+k\env}}$, $\mu \gets 0^{2^k}$ \Comment{$0^n$ denotes a zero-vector of length $n$.}
\State $\lambda\sys \gets 0^{2^k}$

\For{$j \in \{ 1, \ldots, |c| \}$}
    \State $\lambda_{v_j} \gets \lambda_{v_j} + c_j$ \label{line:lambda-loop}
    \State $\mu_{w_j} \gets \mu_{w_j} + c_j$ \label{line:lambdaA-loop}
\EndFor

\State $A' \gets \Call{Subsets}{k}$ \label{line:A-subsets}
\State $R' \gets \Call{Subsets}{k\env}$ \label{line:R-subsets}
\State $\Delta \gets \Gamma' R' \mod 2$ \label{line:Delta}

\For{$j \in \{ 1, \ldots, 2^{k\env} \}$}
    \State $V_{\mathrm{str}} \gets \Delta_{:,j} \oplus A'$ \label{line:R-permutation}
    \State $V_{\mathrm{ind}} \gets \Call{BinaryToDecimal}{V_{\mathrm{str}}}$ \label{line:convert-to-permutation}
    \State $\lambda\sys \gets \lambda\sys + 2^{-k\env}\Call{PermVector}{\mu,V_{\mathrm{ind}}}$ \label{line:add-permuted-vectors}
\EndFor

\State\Return $\frac{1}{k}(\Call{ShannonEntropy}{\lambda\sys} - \Call{ShannonEntropy}{\lambda})$ \label{line:compute-entropies}
\EndFunction
\end{algorithmic}
\caption{Algorithm to compute the coherent information $I_c(\ket{\Gamma},\cN_{\mathbf{p}}^{\ox k})$ of a graph state $\ket{\Gamma}\tot$ defined in terms of a graph $\Gamma=(V\sys\cup V\env,E)$ with $k\equiv k\sys=|V\sys|$, $k\env = |V\env|$, and a Pauli channel $\cN_{\mathbf{p}}$ with $\mathbf{p}=(p_0,p_1,p_2,p_3)$.
We denote the adjacency matrix of the graph $\Gamma$ by the same symbol, and set $\Gamma'=\Gamma\systoenv$ (see \Cref{eq:bipartite-adjacency-matrix}).
Unless stated otherwise, binary operations on vectors and matrices are carried out entry-wise.
The subroutines \textproc{GetUSubsets}, \textproc{GetUSubsetsCard}, \textproc{BinaryToDecimal}, \textproc{Subsets}, \textproc{PermVector}, and \textproc{ShannonEntropy} are defined in \Cref{app:subroutines-for-full-alg}.
Due to the concurrent write operations in \Cref{line:lambda-loop,line:lambdaA-loop}, parallelising the trivial inner loop does not necessarily lead to a runtime speedup.
}
\label{alg:coherent-information-full-algorithm}
\end{algorithm}

\section{Exploiting Graph Symmetries}\label{sec:symmetries}
One crucial observation for quantum codes based on graph states is that we can exploit the symmetries of the underlying graph state in the calculation of the coherent information.
Consider e.g.\ a graph corresponding to a repetition code with $k\sys=5$ system vertices and $k\env=1$ environment vertex colored in white:\footnote{For brevity and consistency we will call this a \${1-in-5} code; see \Cref{sec:rep-codes}.}
\begin{align}\label{eq:1-in-5}
    \Gamma = \graph*{1-in-5-small}[7cm]
\end{align}
This graph is evidently invariant under any permutation of the vertices $\{2,3,4,5\}$.
How is this reflected in the calculation of the coherent information?
We will see that we can indeed exploit the graph's symmetries to significantly speed up the performance of calculating the coherent information for a graph state.

In the following sections, we will summarize the graph-theoretic background necessary, and refer the reader to any standard textbook (e.g.\ \cite{Trudeau1993,Diestel2010}) on graphs wherever we do not explicitly cite any sources.

\subsection{Automorphisms of Colored Graphs}\label{sec:automorphisms}
In this section we collect some fundamental facts of graph symmetries and introduce the notation used throughout the paper.
To start, we analyse the \${1-in-5}-graph in \Cref{eq:1-in-5}.
As mentioned, any interchange of the vertices $2$ to $5$ yields precisely the same graph;
if we were to ignore the colors this symmetry would also include vertex $6$.
However, we use the color to signify the difference between a system vertex---which is passed through the channel---and an environment vertex, which is not; we do not want to be able to interchange it with the system.

The \emph{colored graph's automorphism group} is the right concept to capture this type of symmetry.
In the example in \Cref{eq:1-in-5}, it is generated by a list of transpositions,
\begin{align}\label{eq:1-in-5-autG}
    \Aut(\Gamma) = \langle (23),(34),(45) \rangle \subset S_6,
\end{align}
where $S_6$ is the finite symmetric group of $6$ elements.
More generally, for a graph state we define it to be as follows.
\begin{definition}[Automorphism Group for Graph State]\label{def:automorphism-group}
Let $\ket{\Gamma}$ be a graph state with underlying graph $\Gamma=(V,E)$ and $|V|=k\tot$, such that the vertex set $V=(v_1,\ldots,v_{k\tot})$ has a specific ordering.
Let $1 \le k\sys < k\tot$. Then
\[
    \Aut(\ket{\Gamma},k\sys) := \Aut(\Gamma(k\sys)),
\]
where $\Gamma(k\sys)$ is the vertex-colored graph of $\Gamma$ such that
\[
    \mathrm{color}(v_i) = \begin{cases}
    \text{black} & i \le k\sys \\
    \text{white} & \text{otherwise.}
    \end{cases}
\]
\end{definition}

We note that it is generally \emph{not} the case that, just by coloring the environment vertices in disctinct colors, the automorphism group factors across this bipartition; a counterexample is a graph state such as
\[
    \graph*{3-2-line}[7cm],
\]
with automorphism group $\langle(45)(23)\rangle$.

\subsection{Cosets of Graph Colorings}\label{sec:cosets}
Let $\Gamma=(V,E)$ be a graph.
To visualize a vertex subset $U \subset V$ we can give them a distinct color, just as we marked the subset of environment vertices in white in the last section.
We give an example; if $\Gamma$ is the six-vertex graph in \Cref{eq:1-in-5} with the environment vertex marked in white, we want to select the vertex subset $U=\{2,3\}$, which yields the following colored graph:
\begin{align}
    \phantom{''}\Gamma' = \graph*{col-ex1}[7cm]
    \label{eq:gamma-prime}
\end{align}
To simplify notation, we write $\Gamma' = \Gamma(k\sys,U)$, indicating that all vertices with index beyond $k\sys$ are white, and further all vertices in $U$ have an extra set color.

However, under the graph's automorphism group $\Aut(\Gamma(k\sys))=\langle(23),(24),(25)\rangle$ which allows us to treat the originally-black leaves $2,\dots,5$ interchangeably, this subset coloring is equivalent to marking a different set of vertices such as
\[
    \phantom{'}\Gamma'' = \graph*{col-ex2}[7cm].
\]
Here, we call two colorings equivalent if there exists a permutation $\tau \in \Aut(\Gamma(k\sys))$ such that $\tau(\Gamma')=\Gamma''$, where we define $\tau(\Gamma') \coloneqq (V',E')$ with $V' = \lbrace v_{\tau(i)} \colon v_i \in V \rbrace$, and $E'$ is defined analogously.

In this case we only wanted to color a subset in one color; 
in general, we can obviously allow multiple colors, in which case the graph coloring is simply a list of associations such as $\{ 2 \rightarrow 1, 3 \rightarrow 1, 4 \rightarrow 2 \}\subset V\times\field N$.
The second element is an arbitrary indexing of colors (e.g.\ 1 is green, 2 is blue etc.).
As mentioned, applying a coloring to a graph $\Gamma(k\sys,C)$ is then done in such a way that the environment vertices are marked in a different color.
Continuing with our example, this means that if $6\in U$, then
\[
    \Gamma''' = \graph*{col-ex3}[7cm]
\]
where the environment vertex is filled out in white.
As vertex $6$ cannot be interchanged with any of the system vertices, $\Gamma'''$ is not equivalent to $\Gamma'$.
For consistency  with the graph state terminology and clarity of presentation we will generally call a graph $\Gamma(k\sys)$ \emph{uncolored}---despite the fact that the environment vertices are singled out; in contrast, $\Gamma(k\sys,C)$ is \emph{colored}.
Likewise, we refer to the (strictly speaking 2-colored) automorphism group $\Aut(\Gamma(k\sys))$ simply as the automorphism group of a graph state.

Coming back to our example, how many different equivalent graph colorings for $V$ are there?
To answer this question rigorosly, we define the following quantity.
\begin{definition}[Number of Equivalent Graph Colorings]\label{def:coloring-multiplicity}
For a graph state $\ket{\Gamma}$ with underlying graph $\Gamma=(V,E)$ and $k\sys < k\tot = |V|$, and a coloring $C \subset V \times \field N$, the number of equivalent graph colorings for $C$ is defined as
\[
\#(\Gamma,k\sys,C) \coloneqq \left| \{ \Gamma' \colon \exists\, \tau \in \Aut(\Gamma(k\sys)) \text{ s.t.~} \tau(\Gamma') = \Gamma(k\sys,C) \} \right|.
\]
\end{definition}
In our example graph $\Gamma'$ in \Cref{eq:gamma-prime}, it is easy to see that $\#(\Gamma,5,U) = 6$.
But how many is it in general?
This question is answered most concisely using the theory of group actions.
A (left) action of a group $G$ on a set $X$ is a map $G \times X \ni (g,x)\mapsto g \cdot x$ satisfying $g\cdot (h\cdot x) = (gh)\cdot x$ for all $g,h\in G$, $x\in X$, and $e\cdot x = x$ for all $x\in X$, where $e\in G$ denotes the identity element in $G$.
Whenever there is no source of confusion we simply write $gx\equiv g\cdot x$ for $g\in G$ and $x\in X$.
We use the notation $_GX$ to indicate an action of $G$ on $X$.
For each $x\in X$ we define the \emph{orbit} $Gx=\lbrace g x\colon g\in G\rbrace$ and the \emph{stabilizer group} $G_x = \lbrace g\in G\colon g x = x\rbrace\leq G$.
We will make heavy use of the \emph{orbit-stabilizer theorem} relating the two objects:
\begin{align}
    |Gx| = [G:G_x] = \frac{|G|}{|G_x|},
    \label{eq:orbit-stabilizer-theorem}
\end{align}
and the elements of the orbit $Gx$ are in bijective correspondence with a transversal of $G_x$ in $G$.
The set of orbits of points $x\in X$ under the action of $G$ forms a partition of $X$ which we denote by $\modfrac{X}{G}$.

Returning to graph colorings, we now consider the action of the group $G=\Aut(\Gamma(k\sys))$ on the set of graph colorings, and the stabilizer group $\Stab(G,C)=G_C$ of a coloring $C$, i.e., the subgroup of $G$ such that for all elements $\tau \in \Stab(G,C)$ we have $\tau(C) = C$.
The orbit-stabilizer theorem \eqref{eq:orbit-stabilizer-theorem} now gives the following result.
\begin{lemma}\label{lem:coloring-multiplicity}
    Let $G=\Aut(\Gamma(k\sys))$ be the graph state's automorphism group.
    Then $\#(\Gamma,k\sys,C) = |G| / |\Stab(G,C)|$.
\end{lemma}

Now that we know which graph colorings are equivalent, we are interested in the question whether we can enumerate all non-equivalent graph colorings.
To formalize this, let $[c]^k$ be the set of all possible colorings of a graph $\Gamma$, where either $k=k\sys$ (padded with no colors for the environment vertices), or $k=k\tot$.
In view of the discussion of group actions above, we observe the following:
\begin{lemma}\label{lem:all-graph-colorings}
The set $\Xi(\Gamma,c,k)$ of non-equivalent colorings of $c$ colors of a graph $\Gamma$ is equal to the set of orbits of all colorings under the action of the graph state's automorphism group.
Formally,
\[
\Xi(\Gamma,c,k) = \modfrac{[c]^k}{\Aut(\Gamma(k\sys)).}
\]
\end{lemma}

We will often drop the third argument in the coloring set $\Xi(\Gamma,c)=\Xi(\Gamma,c,k\sys)$ for brevity.
Can we explicitly list all colorings for a graph state?
It turns out that this question is closely linked to the active field of research of finding canonical representatives within each coset $\omega\in\Xi(\Gamma,c,k)$.
This in turn is closely related to the field of graph isomorphisms, i.e.\ deciding whether two graphs $\Gamma_1$ and $\Gamma_2$ are equivalent under vertex permutations.

\begin{definition}[Canonical Map and Coloring]
A canonical map for a graph $\Gamma$ and $c$ colors is a map
\[
    \Can\colon\Xi(\Gamma,c,k) \longrightarrow [c]^k,
\]
such that $\Can(\omega) = C_\omega$ for some coloring $C_\omega \in \omega$; we also call $\Can(\omega)$ a canonical representative of the coloring $\omega$.
A canonical coloring is then the set of canonical representatives $\Can(\Xi(\Gamma,c,k))$.
\end{definition}

Given the last definition, an immediate question is whether we can efficiently list such a canonical coloring.
Assuming a canonical map $\Can$ is available, it is obvious that for small graphs a filtering approach---iterating over all colorings $[c]^k$ and collecting the results in a hash set---is a viable option.
However, this method scales exponentially in $k\sys$: for e.g.\ $c=4$ and $k=20$, we have to filter $2^{40}$ different colorings.
Another caveat is that this only works as long as the hash set of previously-determined cosets fits into the working memory, as we need to check whether a specific coloring has been found before.

A different problem is that $\Can$ might not, in fact, be available, or very expensive to compute (as in our case).
It turns out that a series of significantly faster algorithms are known to address this problem: one by \textcite{McKay2014} which we call \alg{VColG},\footnote{This is derived from \tool{nauty}'s toolkit \cite{McKay2014}, which provides a program for enumerating vertex colorings.} and another algorithm by \textcite{Borie2013} which we denote \alg{SGSColG}.
Both are based on checking whether a coloring is already canonical or not, which is much easier from a computational perspective.
For the sake of completeness we explain \alg{SGSColG} in detail in the following section.

\subsection{Canonical Colorings}\label{sec:canonical-colorings}
The algorithm \alg{SGSColG} described by \textcite{Borie2013} is based on introducing a total order on the colorings, which provides a way of testing whether a specific choice of order is already maximal.
In conjunction with a breadth-first search respecting said order, this allows for early pruning of what is essentially a backtracking algorithm.
Because it is significantly more memory-efficient to traverse the colorings depth-first, we alter the algorithm by exchanging the tree traversal method in the following;
the rest of the procedure closely follows \cite{Borie2013}.
As \citeauthor{Borie2013} mentions, this algorithm is an extension of McKay's canonical graph labeling algorithm described in \cite{Hartke2009}.

The automorphism group $G=\Aut(\Gamma)$ of a graph $\Gamma=(V,E)$ is a permutation group, i.e.\ a subgroup of the symmetric group $S_{|V|}$.
For a permutation group $G\leq S_n$ acting on $[n]$, a \emph{base} is a set of points $\lbrace \beta_1,\dots,\beta_m\rbrace$ with $\beta_j\in [n]$ such that the identity in $G$ is the only group element stabilizing all base points.
Given a base $\lbrace\beta_1,\ldots,\beta_m\rbrace$, we can define the \emph{group stabilizer chain}\footnote{We write $G_i \leq G_{i-1}$ since the subgroup $G_{i-1}$ can in fact equal $G_i$; however, one can discard equal groups as redundant in the following algorithms.} 
\begin{align}\label{eq:stabilizer-chain}
    \{ \id \} &= G_m \leq G_{m-1} \leq \ldots \leq G_1 \leq G_0 = G \\
\intertext{where}
    G_i &= \{ \tau \in G : \tau(\beta_j)=\beta_j\ \forall j \le i \}. \nonumber
\end{align}

In words, the stabilizer chain is a sequence of nested subgroups of $G$, such that each subgroup $G_i$ leaves all base points $\beta_1,\ldots,\beta_i$ invariant.
Since the action base $\{\beta_i\}$ is not unique the stabilizer chain is not unique either; 
unless stated otherwise, in the following we use the \emph{standard base} $\lbrace \beta_1,\dots,\beta_n\rbrace$ with $\beta_i=i$, where $i$ corresponds to the graph's $i$\textsuperscript{th} vertex, and such that the system vertices are sorted to the front.

As shown in the definition below, a stabilizer chain naturally yields a list of generators for the group which respect a specific ordering of the points called a Strong Generating Set.
It has many uses in computational group theory algorithms and can be efficiently calculated using the Schreier-Sims algorithm \cite{Seress2003}.
\begin{definition}[Strong Generating Set]\label{def:strong-generating-set}
Let $G$ be a finite permutation group with stabilizer chain corresponding to some action base $\{ \beta_i \}$ as in \Cref{eq:stabilizer-chain}.
Then a strong generating set $S\subseteq G$ is a set of generators for $G$ satisfying $\langle S \cap G_i \rangle = G_i$, i.e., the $i$\textsuperscript{th} stabilizer in the chain is generated by $S\cap G_i$ for all $i$.
The strong generating set inherits its base-dependence from the stabilizer chain.
\end{definition}
Using such a strong generating set, it is straightforward to calculate the orbit of a point $\beta_{i+1}$ under the action of the stabilizer group $G_i$.
In turn, this allows us to determine a transversal for $\modfrac{G_i}{G_{i+1}}$ in each link of the stabilizer chain in \Cref{eq:stabilizer-chain} relative to a base $\lbrace \beta_1,\dots,\beta_m\rbrace$.
This can be obtained via the bijection between the orbit $\lbrace \sigma \beta_{i+1}\colon \sigma\in G_i\rbrace = \lbrace o_1,\dots,o_l\rbrace$ and the cosets of $\modfrac{G_i}{G_{i+1}}$: a transversal of the latter is obtained from cycling through $o_1,\dots,o_l$ and for each $o_j$ selecting an element $\tau_{ij}\in G_i$ such that $\tau_{ij}\beta_{i+1} = o_j$ for all $1\leq j\leq l = [G_i : G_{i+1}]$.

For our applications, we specialize this to the standard base $\beta_i = i$: 
For all $j \geq i+1$ in the orbit $G_i(i+1)$, we can track which permutation $\tau_{ij} \in G_i$ maps $i+1\mapsto j$, i.e.\ $\tau_{ij}(i+1) = j$.
This yields a transversal for $\modfrac{G_i}{G_{i+1}}$ with respect to the stabilizer chain of $G$ relative to the standard base, which is an example of a so-called Strong Generating System.
\begin{definition}[Strong Generating System]\label{def:strong-generating-system}
A strong generating system for a finite permutation group with respect to some base $\{ \beta_i \}$ is a list of transversals for a strong generating set of the stabilizer chain $\{ G_i \}_i$ for base $\{ \beta_i \}$.
For the standard base $\beta_i = i$ we denote the strong generating system with $T(G):=\{ t_1, \ldots, t_k \}$,
where each $t_i$ is a list of permutations representing the transversal of\, $\modfrac{G_{i-1}}{G_i}$.
\end{definition}

\newcommand{\lex}{_\mathrm{lex}}
A strong generating system will be tremendously useful when deciding whether a graph coloring is already canonical.
In order to exploit it, we define a specific canonical function $\Can\lex$: a coloring is canonical under $\Can\lex$ if it is maximal within its coset, with respect to the natural lexicographical ordering.

\begin{algorithm}[t]
\begin{algorithmic}[1]
\Function{\alg{IsCanonical}}{$C\in \field N^k$, $T=\{t_1, \ldots, t_k \}:\text{transversal}$}
    \State $queue \gets \{ C \}$
    \For{$i \in \{1, \ldots, k\}$}
        \State $set \gets \{\}$
        \For{$w \in queue$}
            \For{$ child \in t_i \cdot w$}
            \Comment{elementwise image of $w$}
                \If{$c_{1:i} <\lex child_{1:i}$}
                \Comment{compare prefix of length $i$}
                    \State \Return \texttt{false}
                \EndIf
                \If{$c_{1:i} = child_{1:i}$}
                    \State $set = set \cup \{ child \}$
                \EndIf
            \EndFor
        \EndFor
        \State $queue \gets set$
    \EndFor
    
    \State \Return \texttt{true}
\EndFunction
\end{algorithmic}
\caption{Is given graph coloring canonical with respect to $\Can\lex$?}\label{alg:is-canonical}
\end{algorithm}
More specifically, for two strings of colors $v,w \in [c]^k$ we define $v <\lex w$ to be true if $v$ comes before $w$ with respect to a standard lexicographical ordering; $\le\lex$ is defined analogously.
This ordering allows us to define a subroutine \alg{IsCanonical} given in \Cref{alg:is-canonical}, and enables us to determine if a given coloring is maximal with respect to $\Can\lex$.

The intuition behind the algorithm is the following:
Instead of computing the entire orbit of a given coloring $C$, we explore the orbit of the prefix of $C$ up to some $i$, denoted $C_{1:i}$.
In order to do so, we use a transversal $t_i$ which applies permutations from the full group $G$ on positions $1,\ldots,i$ only, where we require the indices $C_i$ to be in little-Endian order.
The way a stabilizer chain acts thus respects lexicographic number ordering, as the high-Endian bits (i.e.\ those with higher indices) are acted on and compared first.
As an example, consider the number string $23070$, indexed in little-Endian order such that the lowest index $1$ points to the rightmost $0$, and the highest index $5$ points to the leftmost $2$.
Assuming $G=S_5$, we have that $t_4=\modfrac{G_4}{G_5} = G_4$, which in the standard basis simply contains the transposition $(45)$.
Applying it to $23070$ yields $32070$, which is larger in lexicographical order; thus we know $23070$ was not the canonical element in the orbit of $S_5$, despite not having seen the maximal element---i.e.\ in this case, it sufficed to only explore the prefix $23$ first.

For such a child element $child$, if $child_{1:i} >\lex C_{1:i}$, then already $child >\lex C$ and $C$ was not canonical.
If $child_{1:i} <\lex C_{1:i}$, we drop said child (because no permutation of the lower-Endian sites can ever make it larger than $C$) and proceed with the next element.
Finally, if the prefixes are equal, we cannot say anything: we have to check further bits and enqueue the child to be processed by $t_{i+1}$ etc.
The benefit is clear: instead of calculating the full group $G$ and checking all possible permutations of $C$ under it---where we note that $|G|$ can be enormous---we at most calculate the orbit of $C$ under $G$; the worst-case complexity of the algorithm thus has the promise to be significantly better than a na\"ive approach.
Finally, we only need to keep track of a much smaller list of elements within each loop since we compare the element's lexicographical order position-by-position; this means we try to only ever explore a small portion of the orbit of $C$.

The $\alg{IsCanonical}$ subroutine allows us to compute one representative of each canonical coloring using any backtracking technique.
\citeauthor{Borie2013} used a breadth-first approach; we found that the unpredictable memory consumption and the associated memory allocations made this very inefficient in practice, so we replace it with depth-first pruning, given in \Cref{alg:canonical-colorings}.
\begin{algorithm}[t]
\begin{algorithmic}[1]
\Function{ColoringChildren}{$C \in \field N^k$, $c \in \field N$}
    \For{$i\in \{ k, \ldots, 1 \}$}
        \If{$c_i \neq 0$}
            \If{$C_i < c$}
                \State \Yield $C + \ket{i}$
                \Comment{unit vector $\ket{i}$ of length $k\sys$}
            \EndIf
            \State \Break
        \EndIf
        \State \Yield $C + \ket{i}$
    \EndFor
\EndFunction
\Function{ExploreColoringsDF}{$C \in \field N^k$, $root \in \field N^{k}$, \alg{AcceptBranch}: function}
    \For{$child \in \Call{ColoringChildren}{C, root}$}
        \If{$ \Call{AcceptBranch}{child} $}
            \State \Yield $child$
            \State \Call{ExploreColoringsDF}{$C, child$, \alg{Prune}}
        \EndIf
    \EndFor
\EndFunction
\Function{CanonicalColorings}{$\Gamma\colon \text{graph}$, $k\sys\in\field N$, $c \in \field N$}
    \State $G\sys \gets \Aut(\Gamma(k\sys))|\sys$
    \Comment{action restricted to system vertices}
    \State $T \gets$ \Call{StrongGeneratingSystem}{$G\sys$}
    \Comment{from \Cref{def:strong-generating-system}}
    \State $root \gets (0, \ldots, 0)$
    \Comment{vector length $k\sys$}
    \State \alg{AcceptBranch} $\gets$ \Call{IsCanonical}{$\cdot$, $T$}
    \Comment{from \Cref{alg:is-canonical}}
    \State \Call{ExploreColoringsDF}{$root$, \alg{Prune}}
\EndFunction
\end{algorithmic}
\caption{All canonical colorings of $c$ colors for graph $\Gamma$.}\label{alg:canonical-colorings}
\end{algorithm}

\subsection{Canonical Images for Colored Graphs}\label{sec:canonical-images}
A closely-related question to the one answered in \alg{IsCanonical} in \Cref{alg:is-canonical}---namely whether a given graph's coloring $C$ is already canonically-ordered---is whether we can map $C$ to its canonical image, i.e.\ whether we can compute $\Can(C)$ efficiently.
At first glance this problem might look like it is independent of the graph $\Gamma$ that gives rise to the permutation group $G$ based on which we want to decide canonicity.
In fact, it turns out that there is a package for \tool{GAP} that does just that, based on a given permutation group \cite{Jefferson2019}.
The fundamental issue with this approach is that it is slower than making use of the extra information available---the graph---and isomorphism-based algorithms developed for the latter.\footnote{%
Of course another factor is that \tool{GAP} is not as fast as \tool{C}.
However, even a re-implementation will be unlikely to perform as well as the graph-based approach we are about to introduce.
}

The way we address this particular question is by drawing on graph isomorphism research, in particular the question of finding canonical images of colored graphs.
Indeed, the way many graph isomorphism tests work in practice is by defining a canonical function on graphs, just like for colorings; we then have that two graphs $\Gamma_1 \simeq \Gamma_2$ if and only if their canonical images $\Can(\Gamma_1) = \Can(\Gamma_2)$, where for simplicity of notation we use $\Can$ as for colorings (which will never be ambiguous).
For colored graphs, this of course implies both that the uncolored graph is canonically mapped, and that the coloring itself is in a specific order dictated by $\Can$.

There is a series of tools that can be used to calculate such canonical graph images, foremost \tool{nauty} and \tool{Traces} by \textcite{McKay2014} and \tool{bliss} by \textcite{Junttila2007}, to which e.g.\ \tool{IGraph} is linked.\footnote{%
Other tools are \tool{saucy} and \tool{conauto}; we used neither of them.}
All of them use rigorously-proven underlying algorithms for finding canonical graph images, and we refer the reader to the respective papers for more details.
While they solve a seemingly larger problem---graph isomorphisms instead of canonical colorings---we found this approach to be extremely fast even for moderately-sized graphs of up to 50 vertices.
For a discussion on the tested graphs vs.\ size of their respective automorphism groups we refer the reader to \Cref{sec:results}.

For a given colored graph $\Gamma$ with $c$ colors, the canonical graph image approach gives us a map
\begin{align}
    \Gamma(k\sys,C) \longmapsto \Can(\Gamma)(k\sys,\Can(C)) 
    \quad\text{where}\quad
    \Can(C) \in \Xi(\Gamma,c,k).
    \label{eq:-can-map}
\end{align}
Here we assumed that the graph canonical map $\Can(\Gamma)$ respects the partition between system and environment vertices (which one can usually achieve in practice).
\Cref{eq:-can-map} thus induces a function
\begin{align}\label{eq:canonical-image}
    \alg{CanonicalImage}_{\Gamma(k\sys)} \colon \field [c]^k \longrightarrow \Xi(\Gamma,k\sys,c).
\end{align}
In case the graph in the context is clear we leave out the subscript, and
we emphasize that, for our purposes, it will not be important that \alg{CanonicalImage}$(C)=\Can\lex(C)$.

\subsection{Homomorphic Group Actions}\label{sec:homomorphic-groups}

The final ingredient for the main result of this section, an algorithm computing the coherent information of a graph state that exploits the graph's symmetries, is the concept of homomorphic group actions.
Following \cite{Kerber1999}, two group actions ${}_GX$ and ${}_HY$ are called \emph{homomorphic} if there exists an epimorphism $\eta\colon G \longrightarrow H$ (i.e., a surjective group homomorphism) and a surjective map $\theta\colon  X \longrightarrow Y$, such that $\theta(\tau x) = \eta(\tau) \theta(x)$ for all $\tau\in G$ and $x \in X$.
In this situation we have the following result.
\begin{lemma}[\cite{Kerber1999}]\label{lem:homomorphic-group-action}
Let ${}_GX$ and ${}_HY$ be homomorphic group actions under $\theta$, and let $T$ be a transversal of $\modfrac{Y}{H}$. Then for all $\omega \in \modfrac{X}{G}$ there exists a unique $y\in T$ such that $\omega \cap \theta^{-1}(y) \neq \emptyset$.
\end{lemma}

In other words, \Cref{lem:homomorphic-group-action} says that if a map $\theta$ induces a group action homomorphism, it is never the case that elements from an orbit $\omega\in\modfrac{X}{G}$ are mapped to different orbits in $\modfrac{Y}{H}$.
Of course it is still possible that two distinct $\omega,\omega'\in \modfrac{X}{G}$ are mapped to the same orbit $\tau \in \modfrac{Y}{H}$.

\subsection{Coherent Information of Decohered Graph States with Symmetries}\label{sec:symmetric-lambda}

As we have seen in the last section, for each of the $2^{k\sys}$ coefficients $\lambda_U$ in \Cref{eq:lambda-vector}, we need to evaluate a sum over $2^{k\sys}$ terms; and another sum over $2^{k\sys}$ elements in \Cref{eq:lambda-state-tr-sum} to calculate the marginal $\sigma_\mathrm{B}$; the algorithm runtime thus scales as $\Omega(4^{k\sys})$.
However, if the graph $\Gamma=(V,E)$ has symmetries, many of the $\lambda_U$ will be identical.
Furthermore, many of the triples $(U_1,U_2,U_3)\in\mathcal M(U)$ defined in terms of the vertex subsets $U_i\subset V$ will have the same cardinalities.
The expression of the full state $\sigma_\mathrm{BR}$ in \Cref{eq:lambda-state-sum} is thus highly redundant.
We will exploit this redundancy induced by the graph's symmetries to significantly speed up the calculation of the $\lambda$-vectors in \Cref{eq:lambda-vector,eq:lambda-state-tr-sum}.

Take a graph $\Gamma=(V,E)$ where the vertex set is partitioned into system and environment vertices $V = V\sys \cup V\env$ with $|V\sys|=k\sys$, $|V\env|=k\env$, and for a Pauli channel with noise parameters $(p_0, p_1, p_2, p_3)$ for which $p_0 = 1 - \sum_i p_i$ let $q_i := p_i / p_0$.
For some $U\subset V$, we then have
\[
    \lambda_U = p_0^{k\sys} \sum_{(U_1,U_2,U_3)\in\cM(U)} q_1^{|U_1|} q_2^{|U_2|} q_3^{|U_3|},
\]
where $\cM(U)$ is defined in \eqref{eq:lambdaU-U-condition}.
Each triple $(U_1,U_2,U_3)$ of disjoint subsets of $V$ can be identified with a base $4$ string via the map
\begin{align}\label{eq:U-color-identification}
    (U_1,U_2,U_3) \overset\sim\longmapsto 1\cdot U_1 + 2 \cdot U_2 + 3 \cdot U_3 \in [4]^{k\sys}.
\end{align}
Let $G = \Aut(\Gamma(k\sys))$, which acts on $[c]^k$ by permuting indices.
Using the identification in \Cref{eq:U-color-identification}, we can define a map
\begin{align}\label{eq:group-homom-theta}
    \theta' \colon [4]^{k\sys} \longrightarrow [2]^{k\tot},\quad
    \theta'(s) =  \Gamma(U_1 \oplus U_2) \oplus U_2 \oplus U_3
\end{align}
where we implicitly assume the $U_i$ to be padded with zeros within the brackets.
This leads to the following technical lemma.

\begin{lemma}\label{lem:theta-bijection}
The map  $\theta \coloneqq\theta'(\cdot)|\sys$ with $\theta'$ as in \Cref{eq:group-homom-theta} induces a surjective map between $X\coloneqq\modfrac{[4]^{k\sys}}{G}$ and $Y\coloneqq\modfrac{[2]^{k\sys}}{G}$.
Moreover, letting $n(U)\coloneqq | \theta^{-1}(U)|$, we have for all $\tau\in Y$ that there exists an $n_0\in \bN$ such that $n(U) = n_0$ for all $U\in\tau$.
\end{lemma}
\begin{proof}
We first note that $\theta$ is trivially surjective.
Now let $\sigma\in G$. Then for some $s\simeq (U_1,U_2,U_3)\in[4]^{k\sys}$, we have
\begin{align*}
    \theta'(\sigma s) & =\Gamma(\sigma U_1 \oplus \sigma U_2) \oplus \sigma U_2 \oplus \sigma U_3 \\
    &= \sigma \Gamma \sigma^{-1}(\sigma U_1 \oplus \sigma U_2) \oplus \sigma U_2 \oplus \sigma U_3 \\
    &= \sigma \left[ \Gamma(U_1 \oplus U_2) \oplus U_2 \oplus U_3 \right] \\
    &= \eta(\sigma) \theta'(s)
\end{align*}
for $\eta\equiv\id$, which is a trivial epimorphism $\eta\colon G \longmapsto G$.
Hence, $_G[4]^{k\sys}$ and $_G[2]^{k\sys}$ are homomorphic group actions via $\theta$, and the first claim follows from \Cref{lem:homomorphic-group-action} and the fact that restriction $\cdot|\sys$ and $G$ commute.

We now prove the second claim of the lemma.
By \Cref{lem:homomorphic-group-action}, for a fixed $\omega\in X$ we have $\theta(\omega)\subset\tau$ for a particular $\tau\in Y$.
In fact, $\omega$ covers $\tau$ uniformly via $\theta$.
To see this, let $x$ be a representative of $\omega$, i.e., $\omega = \{ \sigma x\colon \sigma\in G\}$, and let $\sigma\in G_x$, the stabilizer group of $x$.
Then, $\theta(x) = \theta(\sigma x) = \sigma \theta(x)$, and hence also $\sigma\in G_{\theta(x)}$.
This means that we have a chain of subgroups $G_x\leq G_{\theta(x)}\leq G$.
Define now $p = [G_{\theta(x)}:G_x]$ and $q=[G:G_{\theta(x)}]$, and choose left transversals $\{  a_1 G_x,\dots, a_p G_x\}$ of $G_x$ in $G_{\theta(x)}$ and $\{  b_1 G_{\theta(x)}, \dots,  b_q G_{\theta(x)}\}$ of $G_{\theta(x)}$ in $G$.
Then $\{  b_i a_j G_x\colon i\in [q], j\in [p]\}$ is a left transversal of $G_x$ in $G$, and we have
\begin{align}
    \theta(b_i a_j x) = b_i a_j \theta(x) = b_i \theta(x)\quad \text{for all $i\in[q]$, $j\in[p]$},
    \label{eq:cosets-theta}
\end{align}
since $a_j\in G_{\theta(x)}$.
The left transversal $\{  b_i a_j G_x\colon i\in [q], j\in [p]\}$ is in bijective correspondence with the elements of $\omega$ by the orbit-stabilizer theorem, and so is $\{  b_1 G_{\theta(x)}, \dots,  b_q G_{\theta(x)}\}$ with the elements of $\tau$.
\Cref{eq:cosets-theta} shows that each element $b_i\theta(x)$ in $\tau$ is hit by exactly $p$ elements $b_i a_j x$ in $\omega$, and the orbit-stabilizer theorem implies that
\begin{align*}
    p = [G_{\theta(x)}:G_x] = \frac{|G_{\theta(x)}|}{|G_x|} = \frac{\modfrac{|G|}{|G_x|}}{\modfrac{|G|}{|G_{\theta(x)}|}} = \frac{|\omega|}{|\tau|}.
\end{align*}
The proof is concluded by noting that for any $\tau\in Y$ and $U\in \tau$ we have
\begin{align*}
    n(U) = \frac{1}{|\tau|}\sum_{\substack{\omega\in X \\ \omega\cap\theta^{-1}(U)\neq \emptyset}} |\omega|,
\end{align*}
which again follows from \Cref{lem:homomorphic-group-action}.
\end{proof}

This finally leads us to the following main result, which allows us to calculate $\lambda_U$ in a symmetry-aware fashion.
\begin{theorem}\label{th:symm-lambda-state-sum}
Let $X:=\modfrac{[4]^{k\sys}}{G}$, and $ Y':=\modfrac{[2]^{k\tot}}{G}$, where $G=\Aut(\Gamma(k\sys))$.
Then \Cref{eq:lambda-state-sum} becomes
\begin{align*}
    \sigma_\s{BR} =& \sum_{\tau \in Y'} \lambda_\tau \Pi_\tau \\
\intertext{where}
    \lambda_\tau :=&\ p_0^{k\sys} \sum_{\substack{\omega \in X\ \\ \theta'(\omega)\subset \tau}}
    \frac{|\omega|}{|\tau|} q_1^{|U_1|} q_2^{|U_2|} q_3^{|U_3|},\\
    \Pi_\tau :=& \sum_{U \in \tau} \ketbra{U}.
\end{align*}
\end{theorem}
\begin{proof}
Since $|\sigma U_i|=|U_i|$ for $\sigma\in G$, we have that $\lambda_{\sigma U}=\lambda_U$ for any $\sigma\in G$.
We can therefore partition $Y'$ into the orbits $\tau\in Y'$ induced by the graph's automorphism group $G$, i.e.
\[
    \sum_U \lambda_U \ketbra{U} = \sum_{\tau\in Y'} \sum_{U\in\tau} \lambda_U \ketbra{U} = \sum_{\tau\in Y'} \lambda_\tau \Pi_\tau.
\]
By \Cref{lem:theta-bijection} we further know that, for each $\omega\in X$ mapped into $\tau\in Y'$ via $\theta'$, there exist precisely $|\omega|/|\tau|$ elements in $\omega$ that map to the same $U\in\tau$.
The claim follows.
\end{proof}
We emphasize that the sum in the expression for $\lambda_\tau$ in \Cref{th:symm-lambda-state-sum} can be calculated over a traversal of $X$, and it is irrelevant which one we choose.
Indeed, as mentioned at the end of \Cref{sec:canonical-images}, this precisely corresponds to the choice of coloring iterator which can follow a different ordering than the canonical map used to collect the nonzero terms within $\lambda_\tau$.

Each $\lambda_U$ in \Cref{eq:lambda-state-sum} thus occurs with a certain multiplicity $|\tau|$ and weight $|\omega|/|\tau|$; calculating the polynomial only once and keeping track of either quantity thus has the promise of saving a tremendous amount of computational time and memory.
Effectively, we have constructed a sparse representation of $\sigma_\s{BR}$, which compresses identical expressions within the sum that defines $\lambda_U$ into one term, \emph{and} compresses repeated $\lambda_U$.

Yet in order to make full use of the symmetries, we also need to be able to efficiently calculate the partial trace of $\sigma_\s{BR}$.
Showing how this can be achieved given the sparse representation of $\sigma_\s{BR}$ is our next main result.
\newcommand\lpre{\lambda^\mathrm{pre}}
\newcommand\mpre{\mu^\mathrm{pre}}
\begin{theorem}\label{th:symm-lambda-state-tr-sum}
We use the same notation for $X$ and $Y'$ as in \Cref{th:symm-lambda-state-sum}, and we let $Y = \modfrac{[2]^{k\sys}}{G}$.
Then \Cref{eq:lambda-state-tr-sum} becomes
\begin{align*}
    \sigma_{\s{B}} =& \sum_{\tau \in Y} \sum_{b\in [2]^{k\env}} \lpre_{\tau,b} \,\Pi_\tau,
\intertext{where}
    \lpre_{\tau,b} :=& \sum_{\substack{\tau' \in Y' \\ \tau'|\sys = h(\tau,b)}}\!\!\!\! \frac{|\tau'|}{|\tau'|\sys|} \lambda_{\tau'}, \\
    h(a,b) :=&\ \Gamma\systoenv\cdot b \oplus a
    \quad\text{for}\quad
    \Gamma = \begin{pmatrix}
    \Gamma\systosys & \Gamma\systoenv \\
    \Gamma\envtosys & \Gamma\envtoenv
    \end{pmatrix}.
\end{align*}
\end{theorem}
\begin{proof}
    As a first step, we show that $h(a,b)$ transforms covariantly under the group's action on the system vertices.
    For some permutation $\sigma\in G$ and denoting with $[a,0]$ and $[0,b]$ the padded vectors of $a$ and $b$ in $[2]^{k\tot}$, restriction to the system vertices $\cdot|\sys$ and the group action of $G=\Aut(\Gamma(k\sys))$ commute, since system and environment vertices are never interchanged by construction of $G$.
    We thus have
    \begin{align*}
        h(\sigma(a),b) &= \Gamma\systoenv \cdot b \oplus \sigma(a) \\
        &= \left(\Gamma \cdot [0,b] \oplus \sigma[a,0]\right)|\sys \\
        &= \sigma(\sigma^{-1} \Gamma \cdot[0,b] \oplus [a,0] )|\sys \\
        &= \sigma(h(a,\sigma^{-1}(b))),
    \end{align*}
    where we used $\sigma^{-1}\Gamma \sigma = \Gamma$ in the last step.
    This shows that demanding $\tau'|\sys = h(\tau,b)$ is a well-defined expression, since 
    \begin{align*}
    \sum_{b\in[2]^{k\env}} \lpre_{\tau,\sigma^{-1}(b)}\Pi_\tau=\sum_{b\in[2]^{k\env}} \lpre_{\tau,b}\Pi_\tau.
    \end{align*}
    
    Our starting point is the partial trace algorithm for \Cref{eq:lambda-vec-tr,eq:lambda-state-tr-sum}, described in \Cref{sec:decohered-graph-states-ci}:
    \begin{align}\label{eq:--symm-1}
        \sigma_{\s{B}} = \sum_{a \in [2]^{k\sys}} \sum_{b \in [2]^{k\env}} \mpre_{a,b} \ketbra{a}
    \end{align}
    where
    \begin{align*}
        \mpre_{a,b} &= \sum_{ \substack{U \in [2]^{k\tot}\colon \\ U|\sys = h(a,b) }} \lambda_U \\
        &= \sum_{\tau' \in Y'} \sum_{\substack{U \in \tau' \\ U|\sys = h(a,b)}} \lambda_U  \\
        &= \sum_{\substack{\tau' \in Y'\colon \\ \tau'|\sys \ni h(a,b)}} \frac{|\tau'|\sys|}{|\tau'|\sys|}  \sum_{\substack{U \in \tau'\colon \\ U|\sys = h(a,b)}} \lambda_U \\
        &= \sum_{\substack{\tau' \in Y'\colon \\ \tau'|\sys \ni h(a,b)}} \frac{|\tau'|}{|\tau'|\sys|}
        \underbrace{\frac{1}{|\tau'|} \sum_{U\in \tau'} \lambda_U.}_{\equiv \lambda_{\tau'}}
    \end{align*}
    \Cref{eq:--symm-1} therefore becomes
    \begin{align*}
        \sigma_{\s{B}} &= \sum_{\tau \in Y} \sum_{a \in \tau} \sum_{b\in [2]^{k\env}}
        \sum_{\substack{\tau' \in Y' \\ \tau'|\sys \ni h(a,b)}} \frac{|\tau'|}{|\tau'|\sys|} \lambda_{\tau'}
    \end{align*}
    and the claim then follows from covariance of $h$ under $G$.
\end{proof}

Finally, we translate \Cref{th:symm-lambda-state-sum,th:symm-lambda-state-tr-sum} into an algorithm, using the coloring and canonical map sub-procedures developed in the last few sections.
The pseudocode listing can be found in \Cref{alg:symm-lambda}.
\begin{algorithm}[t]
\begin{algorithmic}[1]
\Function{SymmetricLambda}{$\Gamma=(V,E)$ : graph, $k\sys\in\field N$}
    \State $\Can \gets [(\cdot) \longmapsto \alg{CanonicalImage}_{\Gamma(k\sys)}(\cdot)]$
    \Comment{see \Cref{eq:canonical-image}}
    \State $\alg{Multiplicity} \gets [(\cdot) \longmapsto \#(\Gamma,k\sys,\cdot)]$
    \Comment{see \Cref{def:coloring-multiplicity} and \Cref{lem:coloring-multiplicity}}
    \State $\lambda, \lambda_a, \lpre \gets \{ \}$
    \Comment{empty hashtables of tuples}
    \For{$(U_1,U_2,U_3) \in \Call{CanonicalColorings}{\Gamma,k\sys,4}$}
    \Comment{see \Cref{alg:canonical-colorings}}
        \State $C = U_1 + 2U_2 + 3U_3$
        \Comment{coloring $\in[4]^{k\sys}$}
        \State $m_4 \gets$ \Call{Multiplicity}{$C$}
        \State $U \gets \Gamma(U_1 \oplus U_2) \oplus U_2 \oplus U_3$
        \State $m_2 \gets$ \Call{Multiplicity}{$U$}
        \State $m_2' \gets$ \Call{Multiplicity}{$U|\sys$}
        \State $idx \gets \Can(U)$
        \State $idx' \gets \Can(U|\sys)$
        \State $p \gets q_1^{|U_1|} q_2^{|U_2|} q_3^{|U_3|}$
        
        \State $(poly,m) \gets \&\lambda_{idx}$
        \Comment{reference to hashtable entry}
        \State $poly \gets poly + \frac{m_4}{m_2} p$
        \State $m \gets m_2$
        \Comment{$m=0$, or $m=m_2$ already}
        
        \State $(poly,m) \gets \&\lpre_{idx'}$
        \State $poly \gets poly + \frac{m_4}{m_2'} p$
        \State $m \gets m_2'$
        \Comment{$m=0$, or $m=m_2'$ already}
    \EndFor
    
    \For{$a \in \Call{CanonicalColorings}{\Gamma,k\sys,2}$}
        \For{$b \in [2]^{k\env}$}
            \State $b2a \gets \Gamma\envtosys\cdot b \oplus a$
            \State $idx \gets \Can(a)$
            \State $idx' \gets \Can(b2a)$
            \State $(poly,m) \gets \&\lambda_{a,idx}$
            \State $(poly',m') \gets \& \lpre_{idx'}$
            \State $poly \gets poly + poly'$
            \State $m \gets m'$
            \Comment{$m=0$, or $m=\alg{Multiplicity}(b2a)=m'$ already}
        \EndFor
    \EndFor
    
    \State \Return $(\lambda,\lambda_a)$
\EndFunction
\end{algorithmic}
\caption{Symmetry-aware $\sigma_{\s{BR}}$ and $\sigma_{\s{B}}$ from \Cref{th:symm-lambda-state-tr-sum,th:symm-lambda-state-sum}.
The $\lambda$ hashtables have a tuple value type $P \times \field N$, where the first component stores a polynomial in the $q_i$, and the second entry is its multiplicity.
Parallelising the inner loops leads to a speedup since the concurrent operations are much less costly than the rest of the loop body, in contrast to the non-symmetric algorithm in \Cref{alg:coherent-information-full-algorithm}.
The group orders necessary for $\#(\Gamma,k\sys,\cdot)$ in \textsc{Multiplicity} are calculated in \tool{nauty}.}\label{alg:symm-lambda}
\end{algorithm}

\section{Numerical Methods}\label{sec:numerical-methods}
As part of this project we developed a suite called \tool{CoffeeCode} (\cite{CoffeeCode19}) that implements \Cref{alg:symm-lambda,alg:canonical-colorings,alg:is-canonical} described in \Cref{sec:symmetries}.
For special cases we replace the coloring iterator \alg{SGSColG} with \alg{VColG}; we found that depending on the automorphism group of the graph one or the other performs faster.
A MATLAB implementation of \Cref{alg:coherent-information-full-algorithm} in \Cref{sec:graph-states} (which is not exploiting graph symmetries) is available at \cite{CImatlab}.

To obtain our main results in \Cref{sec:results} we used the suite \tool{CoffeeCode} in the following way.
Starting from a sparse expression for the $\lambda_U$ as explained in \Cref{sec:symmetric-lambda}, we can substitute the parameters $q_i$ with the channel's noise tuple $\mathbf p$ to obtain a numerical expression for the coherent information.
Of particular interest to us in this context is the question for which noise parameters of the associated Pauli channels the coherent information for a given graph state code is positive---i.e., determining the noise threshold above which the coherent information is zero.

As explained in \Cref{sec:quantum-capacity}, it is known that for specific Pauli channels various codes exhibit superadditivity.
For 1-parameter families $x\mapsto \cN_x$ of channels such as the depolarizing, BB84 or 2-Pauli channel, a special case of superadditivity of coherent information occurs when the \emph{noise threshold}, the largest $x$ such that $Q(\cN_x)>0$, is higher than the \emph{single-letter threshold}, the largest $x$ such that $I_c(\cN_x)>0$.
Naturally, what we mean by ``higher'' in this context depends on the ordering of a family of channels.
For instance, unbiased noise, modeled by the qubit depolarizing channel $\cD_p$ with parameters as given in \Cref{eq:depolarizing-channel}, $\mathbf p=(1-p,p/3,p/3,p/3)$ for some $p\in[0,1]$, exhibits a natural order: a larger $p$ implies a more likely $X, Y$ or $Z$ flip.
Similarly the 2-Pauli channel $\mathbf p=(1-p,p/2,0,p/2)$ or BB84 channel $\mathbf p = ((1-p)^2,p-p^2,p^2,p-p^2)$ are 1-parameter families with a natural ordering for which the definition of threshold makes intuitive sense.

The notion of threshold becomes more involved when looking at the general case of Pauli channels, where we have three independent noise parameters $p_1,p_2$ and $p_3$.
Can we define a meaningful ordering $\mathbf p>\mathbf p'$ for arbitrary Pauli channels?
The most natural choice appears to be an unbiased one which allows a comparison only if the \emph{ratio} of noise types is equivalent: if
\begin{align}\label{eq:x-parametrization}
    \mathbf p_x := \left(1-x, x p_1, x p_2, x p_3\right)
    \quad\text{for}\quad
    \sum\nolimits_i p_i = 1 \text{ and } x \in [0,1],
\end{align}
then we define $\mathbf p_x \geq\mathbf p_y$ if and only if $x \geq y$, and similarly for strict inequality; if the $p_i$ comprising $\mathbf p_x$ and $\mathbf p_y$ are not identical, the channels are incomparable.
This introduces a partial order on the set of all Pauli channels, with the intuitive understanding that a larger $x$ in $\mathbf p_x$ simply means ``more of the same noise''.
Increased thresholds for this parametrization thus only occur within families of Pauli channels $\cN_{\mathbf{p}_x}$ parametrized by $x$, and there is one family for each probability distribution $(p_1,p_2,p_3)$.
Note that the above partial ordering includes the depolarizing channel $x\mapsto (1-x,x/3,x/3,x/3)$ and the 2-Pauli channel $x\mapsto(1-x,x/2,0,x/2)$, but \emph{not} the BB84-channel $x\mapsto ((1-x)^2,x-x^2,x^2,x-x^2)$.

\subsection{Visualizing Thresholds}
\begin{figure}[tb]
    \centering
    \begin{tikzpicture}[scale=0.8]
    \node at (0,0) {\includegraphics[height=8cm]{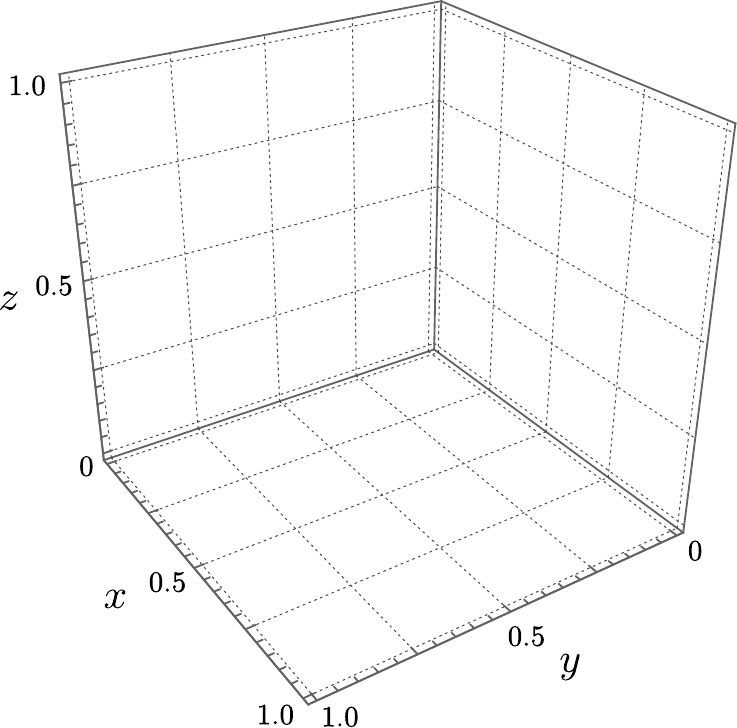}};
    \node[scale=.84,xshift=-43.2,yshift=-8.8] at (0,0) {\includegraphics[height=8cm]{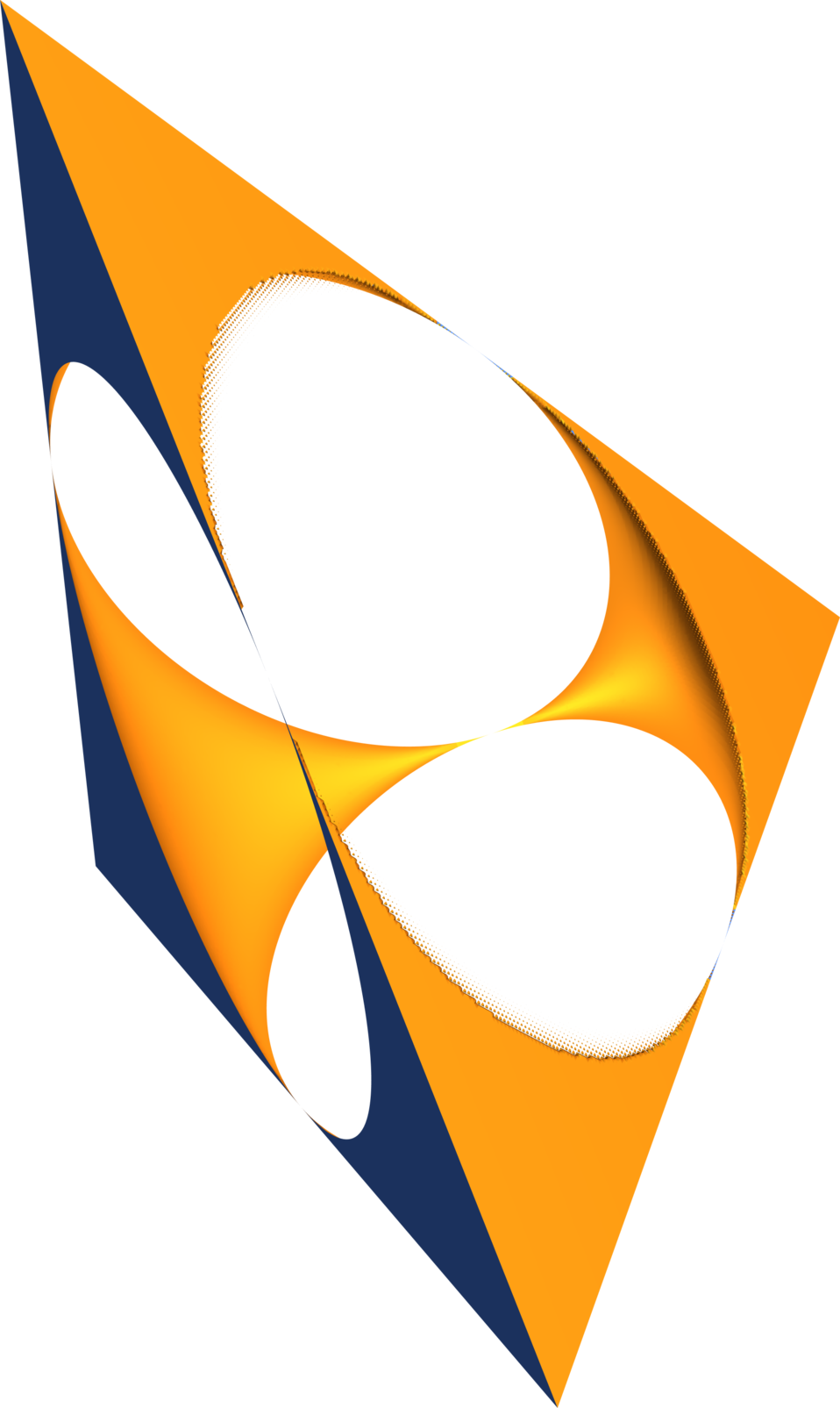}};
    \end{tikzpicture}
    \caption{Set of all non-antidegradable Pauli channels, as derived from \Cref{eq:pauli-antidegradable}.
    The sector close to $p_i=0$ in the lower left corner is equivalent to the three others: If e.g.\ $z$-type noise has $p_3 > 0.5$, i.e.\ $z$ flips occur with more than $50\%$ probability, one can first apply a $z$ flip, mapping that sector back to $p_3 < 0.5$.
    For this reason we restrict the analysis of this paper to the sector around $p_i=0$.
    }
    \label{fig:antidegradability}
\end{figure}

In order to visualize the coherent information threshold for a graph state code $\ket\Gamma$ for all possible Pauli channels, we chose the parametrization given in \Cref{eq:x-parametrization}.
The advantage of this parametrization is that for fixed $p_1,p_2,p_3$ the 1-parameter family $x\mapsto \cN_{\mathbf{p}_x}$ defined in terms of $\mathbf p_x = (1-x, x p_1, x p_2, x p_3)$ has a unique noise level $x_0$ such that the quantum capacity of $\cN_{\mathbf{p}_x}$ vanishes for all $x\geq x_0$, i.e., $Q(\cN_{\mathbf{p}_x}) = 0$ for all $x\geq x_0$.
This is proved in \Cref{lem:x-param-monotonicity} below.
The value $x_0$ thus serves as an \emph{upper bound} on the true noise threshold of the 1-parameter family $x\mapsto \cN_{\mathbf{p}_x}$.
In \Cref{sec:results} we determine \emph{lower bounds} on these noise thresholds.

To state the advertised lemma on the noise threshold upper bound $x_0$, we first recall that every quantum channel $\cN\colon A\to B$ can be written as $\cN(\rho_A) = \tr_E V\rho_{A} V^\dagger$ in terms of an isometry $V\colon \cH_A\to \cH_B\otimes \cH_E$ and an auxiliary Hilbert space $\cH_E$ usually referred to as the environment.
A complementary channel $\cN^c\colon A\to E$ modeling the leakage of information to the environment is defined as $\cN^c(\rho_A) = \tr_B V\rho_{A} V^\dagger$.
A channel $\cN\colon A\to B$ is called \emph{antidegradable}, if there exists another quantum channel $\cA\colon E\to B$ such that $\cN = \cA\circ\cN^c$.
For antidegradable channels $\cN$, we have $I_c(\sigma,\cN)\leq 0$ for any quantum state $\sigma$, and also $Q(\cN)=0$, which we show in \Cref{sec:proof-monotonicity}.

For qubit-qubit quantum channels $\cN$ there is a necessary and sufficient condition for antidegradability \cite{Myhr2009,symmetric2qubit}.
Specialized to Pauli channels $\cN_{\mathbf{p}}$ with $\mathbf{p}=(p_0,p_1,p_2,p_3)$, this condition reads as follows (see e.g.~\cite{PP16}):
\begin{align}
    1 \geq 2\left( p_0^2 + p_1^2 + p_2^2 + p_3^2 \right) - 8\sqrt{p_0p_1p_2p_3}
    \label{eq:pauli-antidegradable}
\end{align}

Checking the inequality \eqref{eq:pauli-antidegradable} in the whole probability simplex $\lbrace (p_1,p_2,p_3)\colon p_i\geq 0, \sum\nolimits_i p_i = 1\rbrace$ yields the region of antidegradable Pauli channels.
Since the quantum capacity is zero for those channels, we focus our analysis on the region of non-antidegradable Pauli channels, depicted in \Cref{fig:antidegradability}.
Evidently, this region consists of four symmetric `sectors', and we focus in the following on the sector closest to the origin.
This sector includes the set of channels $\cN_{\mathbf{p}_x}$ for $x\in[0,1/2]$, where $\mathbf{p}_x = \left(1-x, x p_1, x p_2, x p_3\right)$ is the parametrization of Pauli channels introduced in \cref{eq:x-parametrization}.
This follows since for any non-deterministic probability distribution $(p_1,p_2,p_3)$ the channel $\cN_{\mathbf{p}_{x}}$ becomes antidegradable for some $x<1/2$, which is the content of the following lemma.

\begin{lemma}\label{lem:x-param-monotonicity}
Let $(p_1,p_2,p_3)$ be a fixed probability distribution, and consider the 1-parameter family of channels $x\mapsto \cN_{\mathbf{p}_x}$ where $\mathbf{p}_x = \left(1-x, x p_1, x p_2, x p_3\right)$.
If $(p_1,p_2,p_3)$ is non-deterministic, i.e., $p_i=0$ for at most one $i\in \lbrace 1,2,3\rbrace$, then there exists $0<x_0<1/2$ such that $\cN_{\mathbf{p}_x}$ is antidegradable for all $x\in [x_0,1/2]$, and hence $Q(\cN_{\mathbf{p}_x}) = 0$ for all $x\in [x_0,1/2]$.
If $(p_1,p_2,p_3)$ is deterministic, i.e., $p_i=1$ for exactly one $i\in \lbrace 1,2,3\rbrace$, then $\cN_{\mathbf{p}_x}$ is antidegradable if and only if $x=1/2$.
\end{lemma}

The proof is given in \Cref{sec:proof-monotonicity}. 
The goal of our numerical analysis is to find the largest threshold $x_1$ of a collection of graph states $\ket{\Gamma\tot}$ for a fixed 1-parameter family $\cN_{\mathbf{p}_x}$ such that $x_1$ is as close as possible to the antidegradability point $x_0$ found in \Cref{lem:x-param-monotonicity}.
We first note that the function $x\mapsto I_c(\ket{\Gamma\tot}, \cN_{\mathbf{p}_x}^{\ox k\sys})$ is continuous on $[0,1]$ \cite{leung2009continuity}.
Moreover, for the quantum codes based on graph states studied in this paper, this function is typically monotonically decreasing for $x\in(0,1/2)$.\footnote{All quantum codes in our numerical analysis exhibited this strictly decreasing behavior; however, even in the absence of strict monotonicity of $x\mapsto I_c(\ket{\Gamma\tot}, \cN_{\mathbf{p}_x}^{\ox k\sys})$ our numerical method always finds \emph{some} non-trivial noise threshold.}
We can hence borrow a technique from raymarching computer graphic engines, which essentially performs binary search to zone in on a root of a function within the interval parametrizing the ray.
Instead of a cost $\propto 1/\epsilon$ in the precision $\epsilon$ of the threshold this approach reduces the number of required steps to $\propto \log(1/\epsilon)$.
For our purposes we chose $\epsilon=2^{-20}$, which yields approximately $7.5$ decimal digits of precision;
higher accuracy is readily reached.

To cover all Pauli channels, we choose a high resolution covering of their parameter space derived from spherical coordinates,
\[
    (p_1,p_2,p_3) \in \{ (\sin\theta \cos\phi, \sin\theta \sin\phi, \cos\theta)\colon \theta,\phi \in 0,\delta,2\delta,\ldots,\pi/2 \}.
\]
For $\delta=2^{-10}\pi$ this yields a net of $512\times512$ rays of length $1/2$ starting at the origin, each for a separate channel family.\footnote{We remark that this parametrization distributes overproportionally many points towards the north pole of the parameter space where $\cos\theta$ is small, i.e.\ small $Z$ error.}
The zero threshold of the coherent information thus yields a surface in the Pauli channel simplex, as e.g.\ depicted in \Cref{fig:surface-plot-expl}.
In order to visualize superadditivity, we additionally color each position on the surface depending on how far away the point lies from the single-letter threshold along the same ray.
For a detailed explanation of the surface plot types see \Cref{fig:surface-plot-expl,fig:all-rep-codes-swatch}.

The single-letter threshold for a general Pauli channel $\cN_{\mathbf p}$ with $\mathbf{p}= (p_0,p_1,p_2,p_3)$ is determined by the so-called \emph{hashing bound} \cite{BDSW96}, which gives the optimal single-letter coherent information $I_c(\cN_{\mathbf p})$ of a Pauli channel:
\begin{align}
    I_c(\cN_{\mathbf p}) = 1 - H(\mathbf{p}),\label{eq:hashing-bound}
\end{align}
where $H(\mathbf p) = -\sum_i p_i \log p_i$ is the Shannon entropy of $\mathbf p$.
The coherent information in \eqref{eq:hashing-bound} is achieved by the graph state $\ket\Gamma$ where $\Gamma= \raisebox{0pt}{\begin{tikzpicture}[scale=0.5]\draw[fill] (1,0) circle(1ex); \draw (0,0) -- (1,0); \draw[fill=white] (0,0) circle(1ex); \end{tikzpicture}}$ and $\cN_{\mathbf p}$ acts on the solid vertex.
For the $x$-parametrization $x\mapsto \mathbf{p}_x = \left(1-x, x p_1, x p_2, x p_3\right)$ with a fixed probability distribution $(p_1,p_2,p_3)$ introduced in \Cref{eq:x-parametrization}, the single-letter threshold of $\cN_{\mathbf{p}_x}$ is determined by the (unique) root of $$I_c(\cN_{\mathbf p_x})=1-H(\mathbf{p}_x) = 1-h(x)-x H(( p_1,p_2,p_3)),$$ where $h(x) = -x\log x - (1-x)\log (1-x)$ denotes the binary entropy.

\begin{figure}[t]
    \centering
    \surfacePlot{10cm}{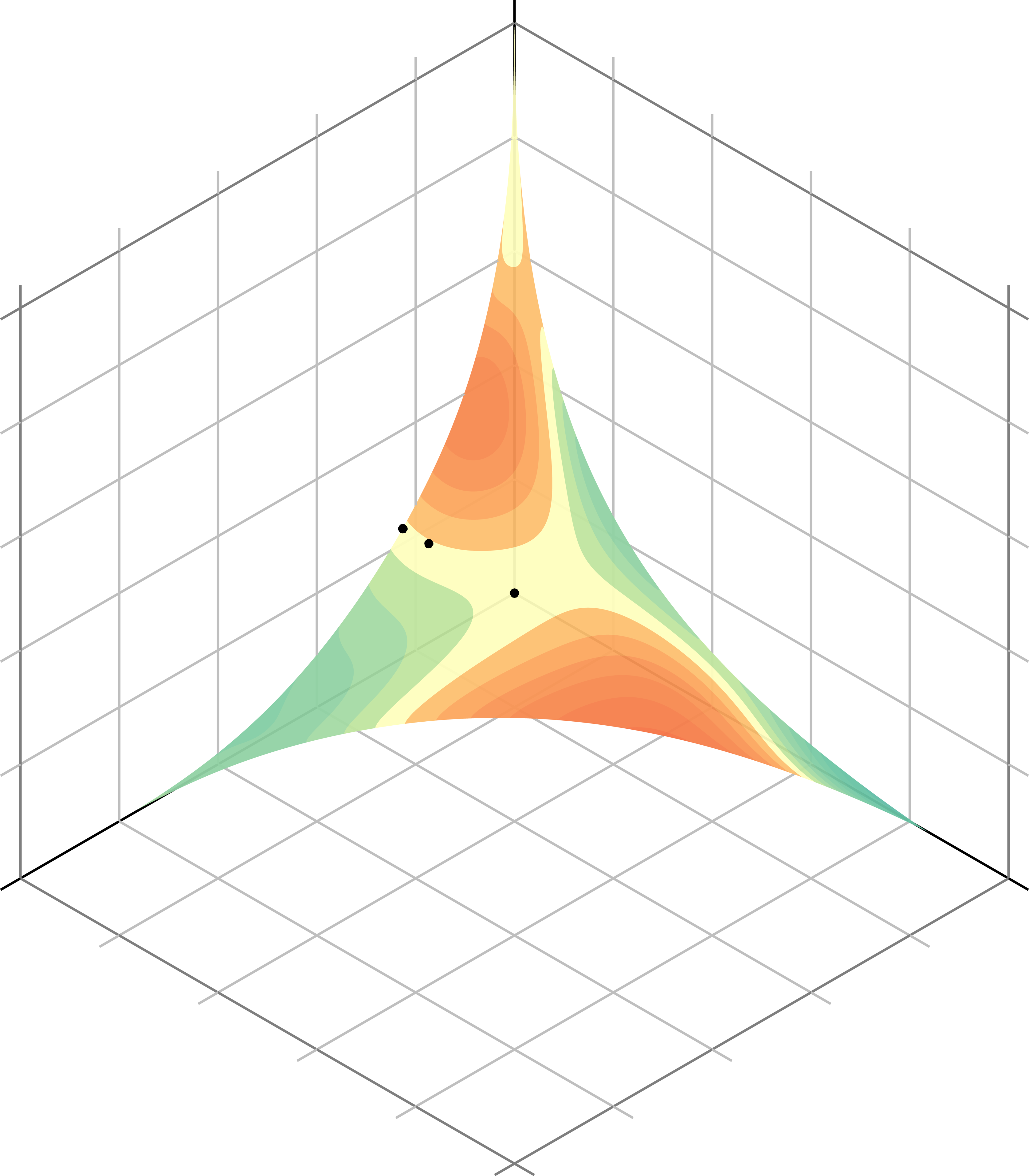}[1][%
        \node (dep) at (0,-.01) {};
        \node (bb) at (-.083,.0405) {};
        \node (pp) at (-.1085,.063) {};
        \node[scale=.5, fill=black, text=white, rectangle callout, callout absolute pointer={(dep)}, below of=dep] {dep};
        \node[scale=.5, fill=black, text=white, rectangle callout, callout absolute pointer={(bb)}, below of=bb] {BB84};
        \node[scale=.5, fill=black, text=white, rectangle callout, callout absolute pointer={(pp)}, above of=pp] {2P};
        \node[] at (-.6,0) {\includegraphics[width=.52cm]{graphs/colorbar.png}};
        \foreach \i/\l in {-4/-0.1,-3/-0.075,-2/-0.05,-1/-0.25,0/0,1/0.025,2/0.05,3/0.075,4/0.1}{
            \draw (-0.62,\i*.552/4) -- (-0.65,\i*.552/4) node[left] {$\l$};
        }
        \node[right] at (.5,-.3) {$0$};
        \node[right] at (.5,.26) {$0.5$};
        \node[left] at (-.5,-.3) {$0$};
        \node[left] at (-.02,-.58) {$0.5$};
        \node[right] at (.02,-.58) {$0.5$};
    ][]
    \caption{Surface plot of the tree graph code \${T$_{15}$} vs.\ the \${1-in-5} repetition code; the \${T$_{15}$}-code is shown in \Cref{fig:all-lvl-legend}, and the \${1-in-5}-code is shown in \Cref{eq:1-in-5}.
    Labeled are the points in the space of all Pauli channels where the families of depolarizing (dep), BB84 and 2-Pauli (2P) channels cut the threshold surface of the \${T$_{15}$} code.
    The color indicates the superadditivity of the code family along the $x$-parametrized ray starting at the origin according to \Cref{eq:x-parametrization}.
    The color scale for all plots of this type is identical throughout the paper, as well as the axes' extent; as such, to simplify the plots, we will generally leave out all but the axes labels.}
    \label{fig:surface-plot-expl}
\end{figure}

\section{Results}\label{sec:results}

In this section we present our main results of the paper.
Using the parametrization $x\mapsto (1-x,xp_1,xp_2,xp_3)$ for some probability distribution $(p_1,p_2,p_3)$ detailed in \Cref{sec:numerical-methods}, we investigate error thresholds for the quantum capacity of Pauli channels in the entire Pauli channel simplex.
Following the discussion in \Cref{sec:quantum-capacity}, this is achieved by choosing (families of) test input states $\psi_{RA}$ on $k\sys$ input qubits and determining the supremum over all $x$ for which $I_c(\psi,\cN_{\mathbf p_x}^{\ox k\sys}) > 0$.

In the following subsections we investigate three families of code states defined in terms of graph states.
The first one is the family of repetition codes (or GHZ states) $|0\rangle^{\ox n} + |1\rangle^{\ox n}$.
Despite being rather simplistic error-correcting codes, repetition codes have long been known to exhibit superadditivity of coherent information; in particular, they increase the error threshold for Pauli channels such as the depolarizing channel \cite{SS96,DSS98,SS07,FW08,BL18} or the BB84 (or independent bit and phase flip) channel \cite{SS07,FW08}, and non-Pauli channels such as
the dephrasure channel \cite{LLS18}.
The second family of codes we investigate comprises concatenated repetition codes, or \emph{cat codes}, which are obtained from concatenating a $Z$-type repetition code with an $X$-type repetition code.
With the right choice of repetition code lengths, these codes have been shown to increase the error thresholds of depolarizing and BB84 channel even further \cite{DSS98,SS07,FW08}.
Both repetition codes and cat codes have large symmetry groups, enabling us to obtain substantial speed-ups in computing their coherent information using \Cref{alg:is-canonical,alg:canonical-colorings,alg:symm-lambda} described in \Cref{sec:symmetries}.
Moreover, our analysis of these code families covers the \emph{entire} Pauli channel simplex, in contrast to previous works that only dealt with particular channels.
Finally, we identify a new family of codes based on tree graphs.
Their large symmetry groups again provide a speed-up via our algorithm, and the resulting error thresholds are better than the ones obtained from repetition and cat codes in large regions of the Pauli channel simplex.
In addition to determining the error thresholds of the above code families, we analyze their rate superadditivity below the threshold as well.

In \Cref{sec:exhaustive} we also carry out an exhaustive search on all graphs with up to 6 system vertices and 4 environment vertices.
However, due to the number of possible graph configurations this exhaustive search quickly becomes infeasible both in terms of runtime and memory.

\subsection{Repetition Codes}\label{sec:rep-codes}
\begin{figure}[t!]
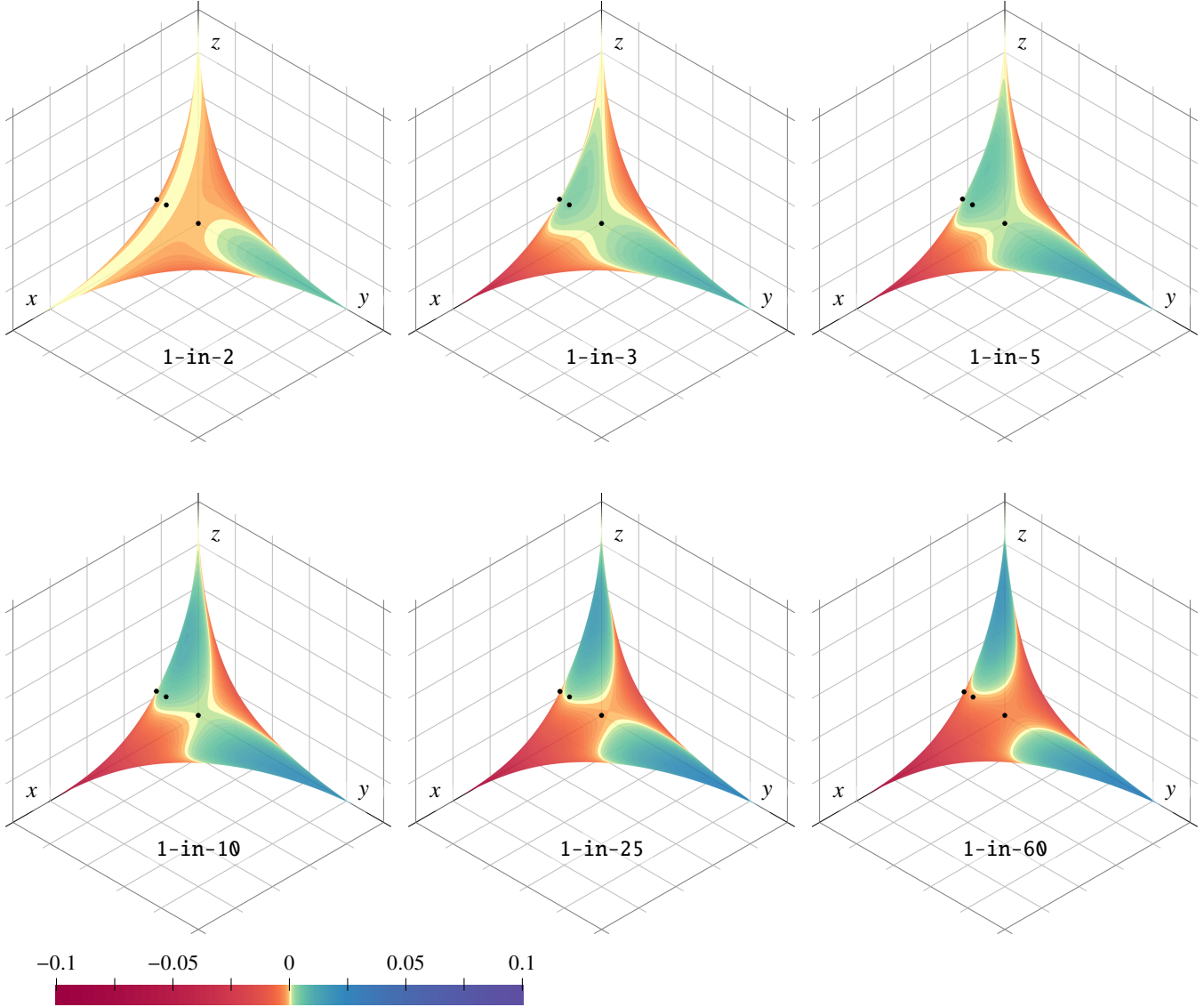

    \hspace*{-0.9cm}
    \begin{minipage}{18cm}
    \centering
    \subplotA{1-in-2}
    \subplotA{1-in-3}
    \subplotA{1-in-5}
    
    \subplotA{1-in-10}
    \subplotA{1-in-25}
    \subplotA{1-in-60}
    \end{minipage}
    \smallColorScale
    \caption{Thresholds for repetition codes; plotted are the codes $1$-in-$k\sys$ for $k\sys\in\{2, 3, 5, 10, 25, 60\}$ from top left to bottom right.
    The axes' extent is the interval $[0,1/2]$, and the three black dots (from center to left) are the loci of depolarizing, BB84, and 2-Pauli channel families, respectively (for details see \Cref{fig:surface-plot-expl}).
    The color scale indicates the distance between noise threshold and single-letter threshold.}
    \label{fig:thresholds-rep}
\end{figure}

\begin{figure}
    \subplotCC{all-rep.vs.single}[all \${1-in-$k$} vs.\ single-letter]
    \hfill
    \subplotCC{all-rep.vs.5-in-5}[all \${1-in-$k$} vs.\ \${5-in-5}]
    \smallColorScale
    \caption{All repetition codes up to $k\sys \le 60$.
    The axes' extent is the interval $[0,1/2]$, and the three black dots (from center to left) are the loci of depolarizing, BB84, and 2-Pauli channel families, respectively---for details see \Cref{fig:surface-plot-expl}.
    Shown is the relative threshold of the best code in the code family vs.\ the single-letter threshold on the left, and vs.\ the \${5-in-5} cat code (defined in \Cref{eq:5-in-5}) on the right.
    In both cases, the color scale indicates the distance between the respective thresholds.}
    \label{fig:all-rep-codes}
\end{figure}

\begin{figure}
    \surfacePlot{10cm}{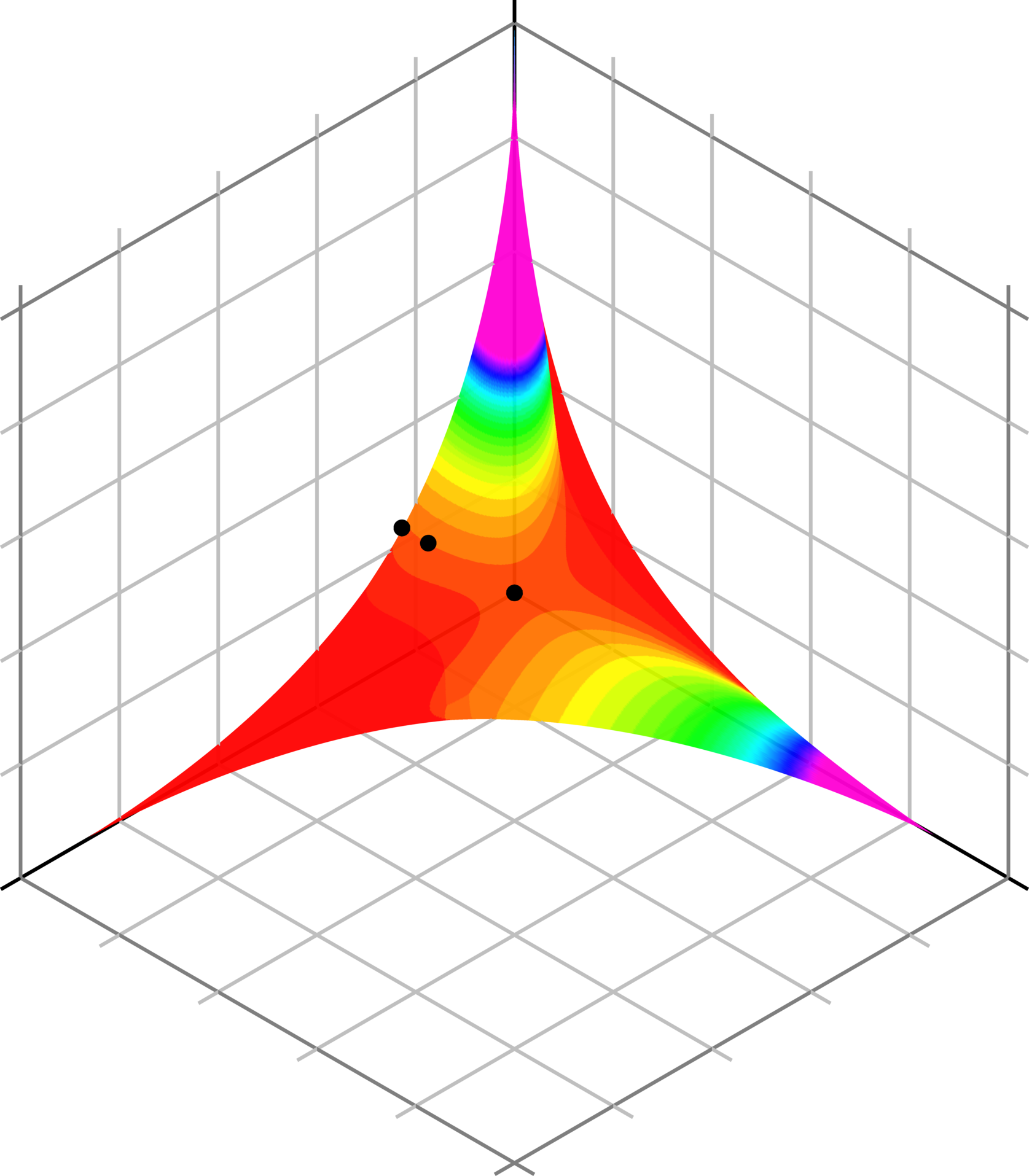}[1][%
        \node (dep) at (0,-.01) {};
        \node (bb) at (-.083,.0405) {};
        \node (pp) at (-.1085,.063) {};
        \node[scale=.5, fill=black, text=white, rectangle callout, callout absolute pointer={(dep)}, below of=dep] {dep};
        \node[scale=.5, fill=black, text=white, rectangle callout, callout absolute pointer={(bb)}, below of=bb] {BB84};
        \node[scale=.5, fill=black, text=white, rectangle callout, callout absolute pointer={(pp)}, above of=pp] {2P};
        \begin{scope}[scale=0.95]
        \foreach \y/\r/\g/\b\l in {2/255/0/3/1-in-2,3/255/18/0/,4/255/45/0/,5/255/66/0/1-in-5,6/255/88/0/,7/255/110/0/,8/255/137/0/,9/255/158/0/,10/255/180/0/1-in-10,11/255/201/0/,12/255/228/0/,13/255/250/0/,14/237/255/0/,15/214/255/0/,16/189/255/0/,17/167/255/0/,18/145/255/0/,19/123/255/0/,20/97/255/0/1-in-20,21/75/255/0/,22/54/255/0/,23/27/255/0/,24/5/255/0/,25/0/255/16/,26/0/255/37/,27/0/255/64/,28/0/255/86/,29/0/255/107/,30/0/255/129/1-in-30,31/0/255/161/,32/0/255/182/,33/0/255/204/,34/0/255/231/,35/0/255/252/,36/0/235/255/,37/0/213/255/,38/0/186/255/,39/0/164/255/,40/0/143/255/1-in-40,41/0/121/255/,42/0/94/255/,43/0/72/255/,44/0/50/255/,45/0/29/255/,46/0/1/255/,47/19/0/255/,48/41/0/255/,49/66/0/255/,50/90/0/255/1-in-50,51/112/0/255/,52/133/0/255/,53/161/0/255/,54/182/0/255/,55/204/0/255/,56/226/0/255/,57/253/0/255/,58/255/0/234/,59/255/0/212/,60/255/0/191/1-in-60}{%
            \definecolor{swatchcolor}{RGB}{\r,\g,\b} 
            \draw[fill=swatchcolor,draw=white] (-.65,\y*.02-.65) rectangle (-.72,\y*.02+.02-.65);
            \ifx\l\empty
            \else
            \draw[black] (-.7,\y*0.02+0.01-.65) -- (-.76,\y*0.02+0.01-.65) node[left] {\${\l}};
            \fi
        }
        \end{scope}
        \node[right] at (.5,-.3) {$0$};
        \node[right] at (.5,.26) {$0.5$};
        \node[left] at (-.5,-.3) {$0$};
        \node[left] at (-.02,-.58) {$0.5$};
        \node[right] at (.02,-.58) {$0.5$};
    ]
    \caption{Continuation of \Cref{fig:all-rep-codes}.
    The plot shows the same surface as before, only this time indexed by colors referring to the highest CI threshold repetition code, starting at red for \${1-in-2} up to magenta for \${1-in-60}.}
    \label{fig:all-rep-codes-swatch}
\end{figure}

Repetition codes are the most simple type of error correcting code.
Represented as a graph state, they correspond to star graphs with $k\tot$ vertices; one ray of the star is the environment vertex.
The \${1-in-5} code in \Cref{eq:1-in-5} is an example of a repetition code.

It is immediately obvious that star graphs have a large symmetry group.
If $\Gamma_n$ denotes a \${1-in-$n$} repetition code, and using the same vertex enumeration as in \Cref{eq:1-in-5}---i.e.\ such that the root has index $1$, and the environment vertex has index $n$---its automorphism group is given by $\Aut(\Gamma_n) = S_{\{2,\ldots,n-1\}}$, and hence $|\Aut(\Gamma_n)| = (n-2)!$.
It is this fact together with the very simple orderings of the canonical image and coloring algorithms in \Cref{sec:symmetric-lambda} that make repetition codes particularly amenable to our numerical methods, allowing us to evaluate codes on $\approx100$ system qubits in total without much difficulty.

In \Cref{fig:thresholds-rep}, we plot the coherent information threshold for various repetition codes.
Visible is how for growing $k\sys$ the code shifts from superadditivity in the $Y$ and $Z$-directions to superadditivity in the $X$ and $Z$-directions.
However, it is known that for unbiased noise---i.e.\ a depolarizing channel---the \${1-in-5} code has the optimal threshold under all repetition codes.
We also point out that while the smaller codes (apart from the $1$-in-$2$ code) are not symmetric, for growing $k\sys$ the code family has an emerging approximate mirror symmetry across the plane dividing the $X$ and $Z$ axes.

Treating all repetition codes together as one ``family'' of codes, we compare the largest thresholds of this code family with the single-letter thresholds and a \${5-in-5} code (see \Cref{sec:cat-codes}) in \Cref{fig:all-rep-codes}.
The left-hand plot of \Cref{fig:all-rep-codes} shows that repetition codes have better thresholds than the single-letter threshold for $Z$-and $Y$-biased noise and unbiased noise, but not for $X$-biased noise.
Moreover, the right-hand plot shows that repetition codes provide advantage over the \${5-in-5} code for all biased noise types, but not for the region around the depolarizing channel (unbiased noise), as observed in \cite{DSS98}.
It is also interesting to see which repetition code is optimal for which type of noise shifts when the noise becomes more and more biased to one type, see \Cref{fig:all-rep-codes-swatch}.

We note here that the apparent difference between the three types of biased noise in \Cref{fig:thresholds-rep,fig:all-rep-codes,fig:all-rep-codes-swatch} (as well as all plots in the subsequent sections) is merely due to the particular basis choice in the definition of graph states (see \Cref{sec:graph-states-definition}).
To see this, note that for a given graph $\Gamma=(V,E)$ the stabilizer generators $S_i$ of the associated graph state can be chosen as
\begin{align*}
    S_i = R^i T^{N_i}
\end{align*}
for \emph{any} choice of $R,T\in\lbrace X,Y,Z\rbrace$ with $R\neq T$; different choices for $R$ and $T$ yield different relabelings of the axes in e.g.~\Cref{fig:thresholds-rep}.
In this paper, we use the canonical choices $R=X$ and $T=Z$.

\subsection{Concatenated Codes}\label{sec:cat-codes}
\begin{figure}
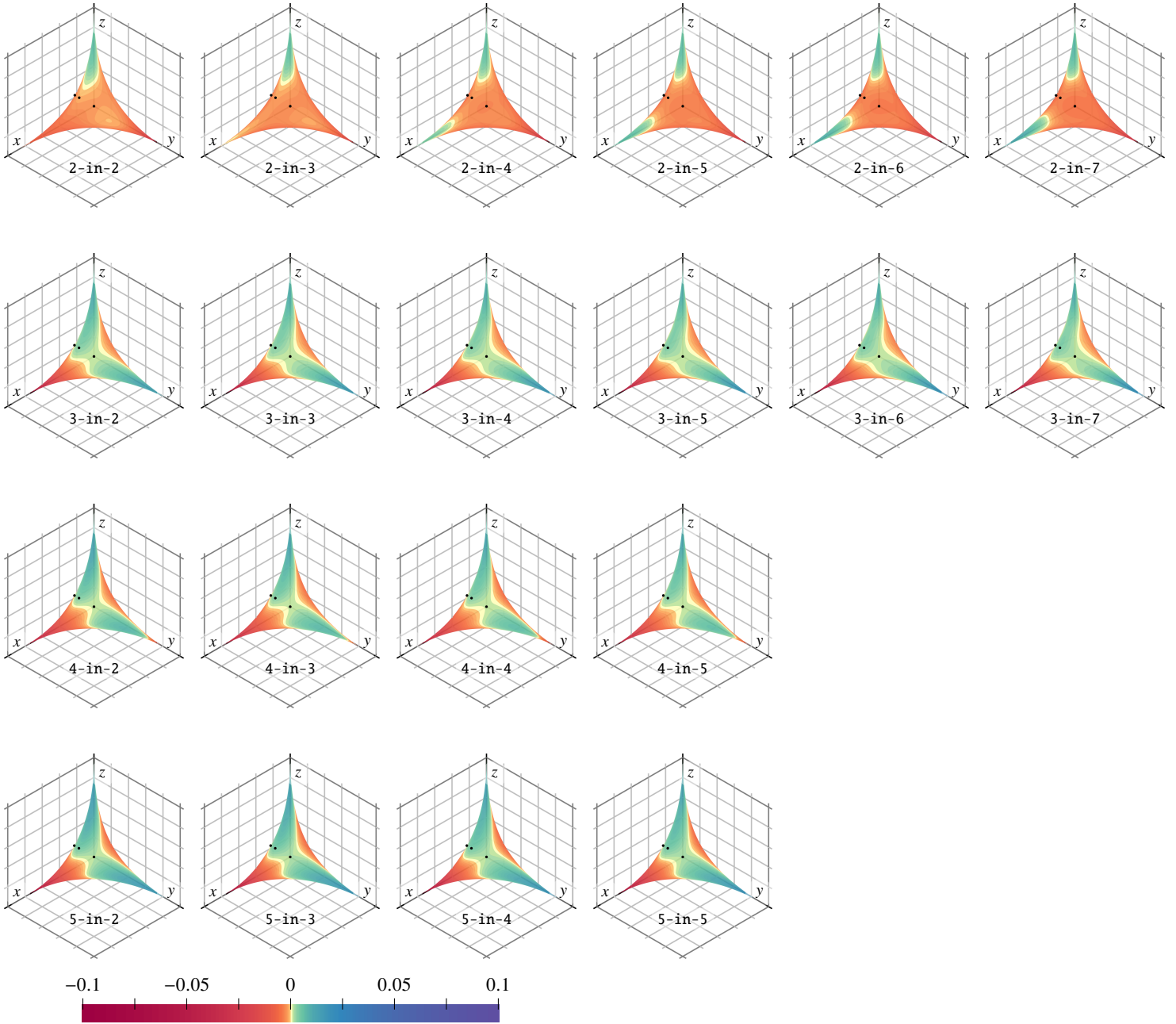

    \hspace{-1.45cm}
    \begin{minipage}{20cm}
    \subplotB{2-in-2}
    \subplotB{2-in-3}
    \subplotB{2-in-4}
    \subplotB{2-in-5}
    \subplotB{2-in-6}
    \subplotB{2-in-7}
    
    \subplotB{3-in-2}
    \subplotB{3-in-3}
    \subplotB{3-in-4}
    \subplotB{3-in-5}
    \subplotB{3-in-6}
    \subplotB{3-in-7}
    
    \subplotB{4-in-2}
    \subplotB{4-in-3}
    \subplotB{4-in-4}
    \subplotB{4-in-5}
    \\
    \subplotB{5-in-2}
    \subplotB{5-in-3}
    \subplotB{5-in-4}
    \subplotB{5-in-5}
    \end{minipage}
    \smallColorScale
    \caption{Cat codes $4\times 7$ table.
    The axes' extent is the interval $[0,1/2]$, and the three black dots (from center to left) are the loci of depolarizing, BB84, and 2-Pauli channel families, respectively---for more details see \Cref{fig:surface-plot-expl}.
    The color scale indicates the distance between noise threshold and single-letter threshold.}
    \label{fig:cat codes}
\end{figure}

Concatenated codes, or \emph{cat codes}, arise from concatenating a $Z$-type repetition code $\ket 0^{\ox n_1} + \ket 1^{\ox n_1}$ with an $X$-type repetition code $\ket +^{\ox n_2} + \ket -^{\ox n_2}$, where $\ket \pm = (\ket 0 \pm \ket 1)/\sqrt{2}$.
We call such a code an \${$n_1$-in-$n_2$} code.
The well-known 9-qubit Shor code \cite{Shor95} is a \${3-in-3} code.
Cat codes are stabilizer codes and hence local-Clifford-equivalent to graph states; in our paper, we simply refer to any state of this equivalence class as cat code.

Before giving the construction for the underlying graph of a cat code, we briefly review the procedure of concatenation of stabilizer codes.
Let $[n_1,k]$ and $[n_2,1]$ be qubit stabilizer codes with generators $S=\lbrace s_1,\dots,s_{n_1-k}\rbrace$ and $T=\lbrace t_1,\dots,t_{n_2-1}\rbrace$, respectively.
We call $S$ and $T$ the \emph{outer} and \emph{inner} code, respectively.
Furthermore, denote by $\bX$ and $\barZ$ the logical operators of the inner code $T$, and denote by $\lbrace |\bz\rangle , |\bo\rangle\rbrace$ the computational basis of the encoded qubit, i.e., the eigenstates of $\barZ$ corresponding to $+1$ and $-1$, respectively.
Then the concatenated stabilizer code $[n_1n_2,k]$ is a stabilizer code whose generators are $n_1$ copies of the $n_2-1$ generators of the inner code acting on the physical qubits within blocks of size $n_2$, plus the $n_1-k$ generators of the outer code acting on the \emph{logical} qubit of the inner code, for a total of $n_1(n_2-1) + n_1 - k = n_1n_2 - k$ generators.
The latter $n_1-k$ stabilizers are obtained by the replacements $X\leftarrow\bX$ and $Z\leftarrow\barZ$ in every $s\in S$.

As an example, consider Shor's 9-qubit code, which is a $[3,1]$ $Z$-repetition code with stabilizers $Z^1Z^2, Z^1Z^3$ and logical operators $\bX=X^1X^2X^3, \barZ = Z^1Z^2Z^3$, concatenated with a $[3,1]$ $X$-repetition code with stabilizers $X^1X^2, X^1X^3$ and logical operators $\bX=X^1X^2X^3, \barZ = Z^1Z^2Z^3$.
This results in a $[9,1]$ code by the above discussion, with inner and outer code stabilizers 
\begin{align}
    \begin{aligned}
    Z^1Z^2,\quad Z^1Z^3,\quad Z^4Z^5,\quad Z^4 Z^6,\quad Z^7 Z^8,\quad Z^7 Z^9\\
    X^1X^2X^3X^4X^5X^6,\quad X^1X^2X^3 X^7 X^8 X^9
    \end{aligned}
    \label{eq:shor-code-stabilizers}
\end{align} 
and logical operators
\begin{align}
    \bX &= X^1X^2X^3X^4X^5X^6 X^7 X^8 X^9 & \barZ &= Z^1 Z^2 Z^3 Z^4 Z^5 Z^6 Z^7 Z^8 Z^9.
    \label{eq:shor-code-logical-ops}
\end{align}
To use the Shor code as an input state to $k\sys=9$ copies of a Pauli channel, we encode it in half of a maximally entangled state $\ket\phi = \frac{1}{\sqrt{2}}(\ket 0_R \ket \bz_A + \ket 1_R \ket \bo_A$, where $\ket\bz_A$ and $\ket\bo_A$ are the $\pm 1$ eigenstates of the logical $\barZ$ in \eqref{eq:shor-code-logical-ops}, respectively.
The state $\ket\phi$ is again a stabilizer state with stabilizer generators $X^1 \bX^l$ and $Z^1 \barZ^l$ (where we denote the logical qubit encoded in the 9-qubit code by $l$).
Hence, we add the all $X$- and all $Z$-stabilizers (defined on $10$ qubits) to the list in \eqref{eq:shor-code-stabilizers} and turn this stabilizer state into an LU-equivalent graph state using the procedure outlined in \Cref{app:LU-equivalence}.
This yields the graph $\Gamma_{\mathrm{Shor}}$ together with its adjacency matrix (where we replaced $0$ with $\cdot$ for readability) for the graph state LU-equivalent to the Shor code:
\begin{align}
\raisebox{-1.61cm}{\includegraphics[height=3.5cm]{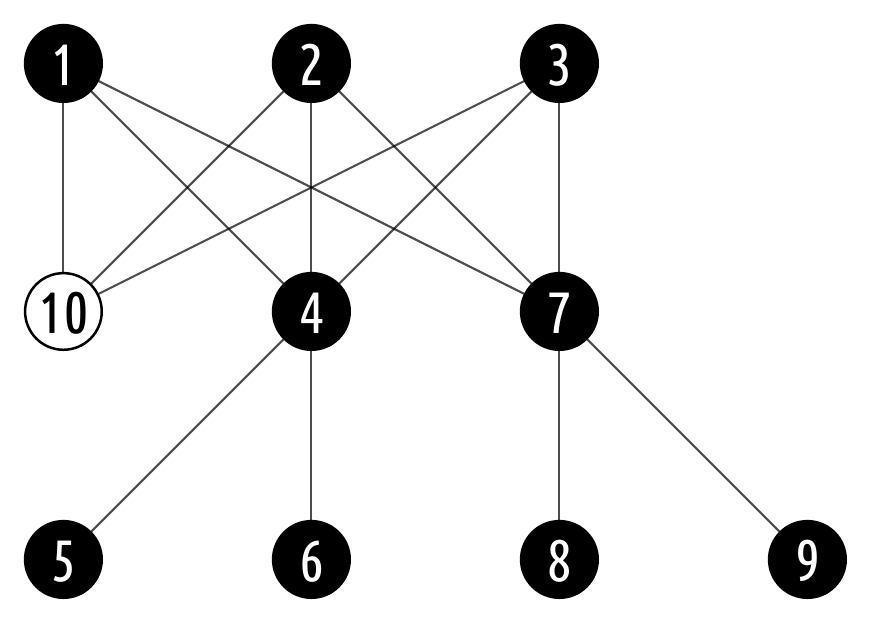}} \!\!\!\!\!\!\!\! = 
    \Gamma_{\mathrm{Shor}} = \begin{pmatrix}
 \cdot & \cdot & \cdot & 1 & \cdot & \cdot & 1 & \cdot & \cdot & 1 \\
 \cdot & \cdot & \cdot & 1 & \cdot & \cdot & 1 & \cdot & \cdot & 1 \\
 \cdot & \cdot & \cdot & 1 & \cdot & \cdot & 1 & \cdot & \cdot & 1 \\
 1 & 1 & 1 & \cdot & 1 & 1 & \cdot & \cdot & \cdot & \cdot \\
 \cdot & \cdot & \cdot & 1 & \cdot & \cdot & \cdot & \cdot & \cdot & \cdot \\
 \cdot & \cdot & \cdot & 1 & \cdot & \cdot & \cdot & \cdot & \cdot & \cdot \\
 1 & 1 & 1 & \cdot & \cdot & \cdot & \cdot & 1 & 1 & \cdot \\
 \cdot & \cdot & \cdot & \cdot & \cdot & \cdot & 1 & \cdot & \cdot & \cdot \\
 \cdot & \cdot & \cdot & \cdot & \cdot & \cdot & 1 & \cdot & \cdot & \cdot \\
 1 & 1 & 1 & \cdot & \cdot & \cdot & \cdot & \cdot & \cdot & \cdot
    \end{pmatrix}
    \label{eq:shor-code}
\end{align}

The procedure above can be generalized to arbitrary \${$n_1$-in-$n_2$} concatenated (``cat'') codes as follows.
We take a single environment vertex $v\env$ and disjoint sets of system vertices $A$ with $|A|=n_1$, $B$ with $|B|=n_2-1$, and $C^i$ for $i=1,\ldots,n_2-1$ with $|C^i|=n_1-1$.
We fully connect $\{v\env\}$ with $A$, and $A$ with $B$, such that the vertices $\{v\env\} \cup A$ form a complete bipartite graph with $B$. Finally, we fully connect the $i$\textsuperscript{th} vertex $b^i$ in $B$ with $C^i$.
Overall the graph has $n_1+(n_2-1)+(n_2-1)(n_1-1)=n_1n_2$ system vertices, and a single environment vertex, $v\env$.
For $n_1=1$ this construction yields a star graph on $n_2+1$ vertices, and hence a \${1-in-$n_2$}-cat code is a repetition code as defined in \Cref{sec:rep-codes}.

For the Shor code in \Cref{eq:shor-code} above, a \${3-in-3}-cat code, we have $A = \lbrace 1,2,3\rbrace$, $B = \lbrace 4,7\rbrace$, and $C^1 = \lbrace 5,6\rbrace$, $C^2= \lbrace 8,9\rbrace$.
We give a series of other examples: the \${2-in-2}, \${2-in-5}, \${5-in-2} and \${3-in-4}-cat codes have the graphs
\[
    \includegraphics[height=1.3cm]{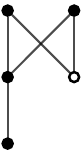}
    \quad
    \includegraphics[height=1.3cm]{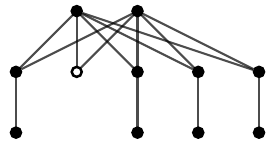}
    \quad
    \includegraphics[height=1.3cm]{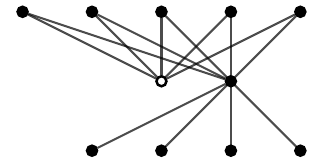}
    \quad
    \includegraphics[height=1.3cm]{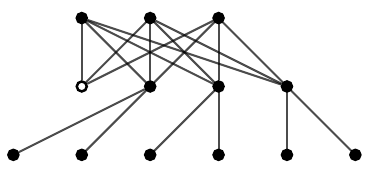}.
\]
Larger codes follow the same pattern; the \${5-in-5} code, for instance, has $25$ system vertices, and its graph is
\begin{align}
    \includegraphics[height=1.5cm]{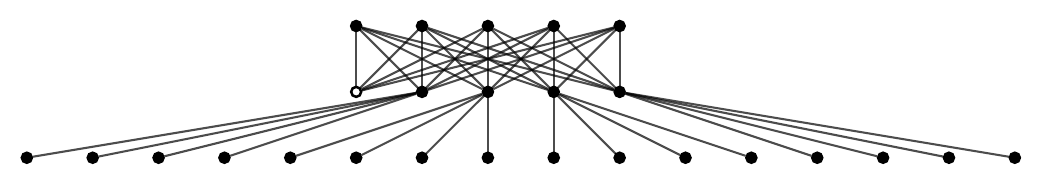}.
    \label{eq:5-in-5}
\end{align}

It is evident that the symmetry group of cat codes factors in a clean fashion with respect to these vertex subsets: if $\Gamma_{n_1,n_2}$ denotes the graph, we have
\[
\Aut(\Gamma_{n_1,n_2}) = S_{A} \times \left(S_{M^i} \cup \prod_{i=1}^{n_2-1} S_{C^i} \right),
\]
where $S_{M^i}$ denotes the symmetric group on the vertex set of the branches between $B$ and the $C^i$, i.e.\ $M^i := \{ \{ b^i \} \cup C^i : i = 1, \ldots, n_2-1 \}$.
A straightforward calculation thus gives the size of the automorphism group as $|\Aut(\Gamma_{n_1,n_2})| = n_1! (n_2-1)! (n_1-1)!^{(n_2-1)}$.

\Cref{fig:cat codes} shows a $4\times 7$-matrix plotting the threshold surfaces of \${$n_1$-in-$n_2$} cat codes for $2\leq n_1\leq 5$ (along the rows) and $2\leq n_2\leq 7$ (along the columns).
Note that \${1-in-$n_2$} cat codes are in fact repetition codes, which we discussed in \Cref{sec:rep-codes}.
It is evident from \Cref{fig:cat codes} that \${2-in-$n_2$} cat codes provide no advantage over the single-letter optimal code except for heavily $Z$-biased noise and $X$-biased noise (for $n_2\geq 4$).
Interestingly, the advantage for $X$-biased noise disappears for $n_1\geq 3$, and the increased threshold region extends from $Z$-biased noise over the unbiased center (corresponding to the depolarizing channel) to the $Y$-biased noise corner.
In particular, the cat codes with $n_1=3,5$ provide substantial threshold increases.
For unbiased noise (and the region around it), the increase peaks (in our analysis) at the \${5-in-5} cat code, which provides the optimal threshold for depolarizing noise up to $k\sys = 25$ system vertices (\cite{DSS98}; see also \Cref{tab:all-lvl-dep}). 

\subsection{Tree Codes}\label{sec:tree-codes}
\begin{table}[t]
    \centering
    \begin{tabular}{r l lll}
        \toprule
        \multirow{2}{*}{$k\sys$} & \multirow{2}{*}{code} & \multicolumn{3}{c}{threshold $p$ for} \\
        & & $\mathbf{p}_{\mathrm{dep}}$ & $\mathbf{p}_{\mathrm{two}}$ & $\mathbf{p}_{\mathrm{BB84}}$ \\
         \midrule
         1 & \${single-letter} & \textbf{0.189\,289\,69} & \textbf{0.227\,092\,195} & \textbf{0.110\,027\,864} \\
         5 & \${1-in-5} & \textbf{0.190\,356\,09} & 0.226\,678\,536\,079 & \textbf{0.112\,104\,217\,521}\\
         13 & \${T$_{13}$}  & \textbf{0.190\,402\,60} & 0.226\,568\,351\,666 & 0.112\,019\,072\,317 \\
         15 & \${T$_{15}$}  &  \textbf{0.190\,433\,58} & 0.226\,519\,095\,093 & 0.111\,945\,765\,086 \\
         16 & \${T$_{16}$}  & \textbf{0.190\,455\,95} & 0.226\,649\,487\,254 & 0.112\,065\,508\,789 \\
          & \${T$_{17}$}  & 0.190\,421\,11 & 0.226\,387\,237\,914 & 0.111\,813\,860\,915 \\
         18 & \${T$_{18}$}  & \textbf{0.190\,499\,74} & 0.226\,623\,046\,487 & 0.112\,011\,792\,424 \\
         & \${T$_{19}$}  & 0.190\,494\,51 & 0.226\,710\,682\,919 & 0.112\,101\,135\,434 \\
         21 & \${T$_{21}$}  & \textbf{0.190\,549\,41} & 0.226\,705\,344\,923 & 0.112\,063\,333\,384 \\
         25 & \${5-in-5}  & \textbf{0.190\,561\,422\,871} & 0.226\,836\,275\,260 & \textbf{0.112\,202\,204\,461} \\
         30 & \${5-in-6}  & \textbf{0.190\,597\,183\,933} & & \\
         \bottomrule
    \end{tabular}
    \caption{A selection of cat and $2$-level codes and their thresholds for depolarizing noise (third column, $\mathbf{p}_{\mathrm{dep}} = (1-p,p/3,p/3,p/3)$), 2-Pauli noise (fourth column, $\mathbf{p}_{\mathrm{two}} = (1-p,p/2,0,p/2)$), and the BB84 channel (fifth column, $\mathbf{p}_{\mathrm{BB84}} = ((1-p)^2,p-p^2, p^2, p-p^2)$). The naming for the $2$-level codes is given in \Cref{fig:all-lvl-legend}. In bold are the best known thresholds for any code with $\le k\sys$ system vertices.
    The single-letter, \${1-in-5}, and \${5-in-5} thresholds are confirmed in other works (see, e.g., \cite{DSS98,FW08}).}
    \label{tab:all-lvl-dep}
\end{table}

\begin{figure}
    \hspace*{-0.9cm}
    \begin{minipage}{18cm}
    \centering
    \subplotD[T$_{13}$]{T13}
    \subplotD[T$_{15}$]{T15}
    \subplotD[T$_{16}$]{T16}
    \subplotD[T$_{17}$]{T17}
    \subplotD[T$_{18}$]{T18a}
    \subplotD[T$_{19}$]{T19a}
    \subplotD[T$_{21}$]{T21}
    
    \subplotAA{all-lvl.vs.single}[all 2-level vs.\ single-letter]
    \subplotAA{all-lvl.vs.5-in-5}[all 2-level vs.\ \${5-in-5}]
    \subplotAA{all-lvl.swatch}
    \end{minipage}
    \smallColorScale
    \caption{2-level codes.
    The axes' extent is the interval $[0,1/2]$, and the three black dots (from center to left) are the loci of depolarizing, BB84, and 2-Pauli channel families, respectively---for more details see \Cref{fig:surface-plot-expl}.
    The first two rows show the noise thresholds vs.\ the single-letter threshold for the codes defined in \Cref{fig:all-lvl-legend}.
    In the third row, the first two plots are colored according to the relative threshold of the best code in the entire tree graph code family vs.\ single-letter (left) and vs.\ the \${5-in-5} cat code (middle).
    In all these plots, the color scale indicates the distance between the respective thresholds.
    The rightmost plot in the third row shows the same surface indexed by colors referring to the highest CI threshold graph state codes; for the corresponding color legend see \Cref{fig:all-lvl-legend}.
    }
    \label{fig:all-lvl}
\end{figure}

\begin{figure}
    \centering
    \begin{tikzpicture}[]
        \node at (0, 0) {\includegraphics[width=5cm, trim=0 2cm 0 2cm, clip]{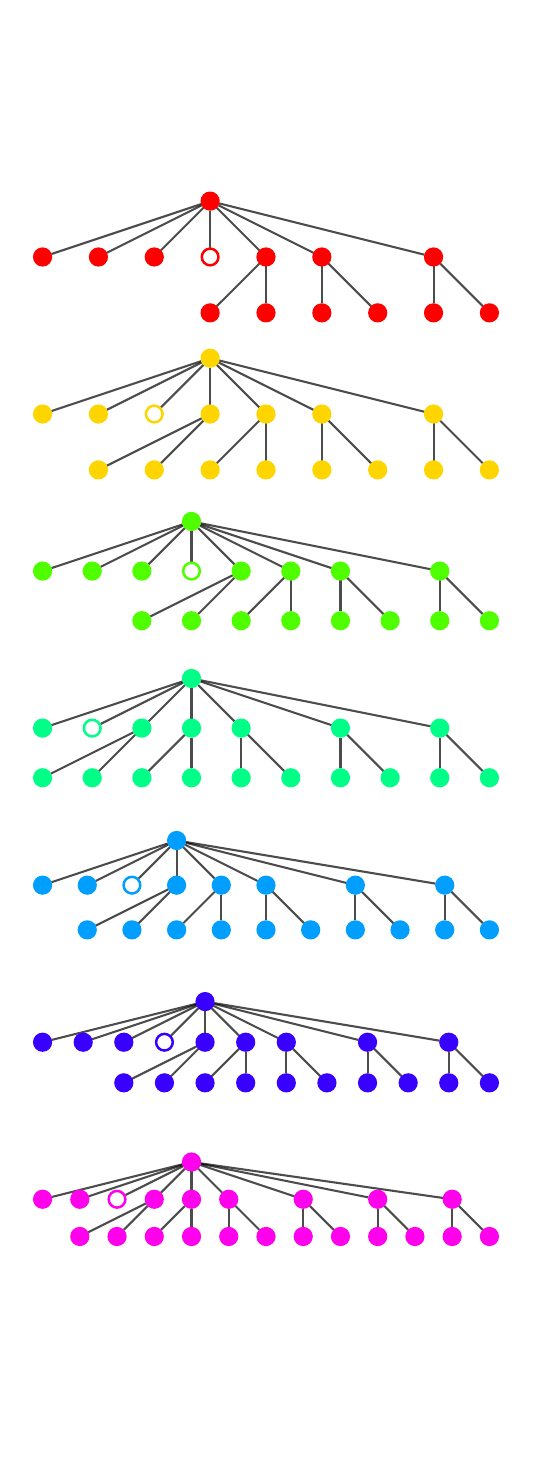}};
        \foreach \y/\l/\ll in {0/13/,1/15/,2/16/,3/17/,4/18/,5/19/,6/21/}{%
            \node[left] at (-3,4.2-1.47*\y) {\${T$_{\l\mathrm{\ll}}$}};
        }
    \end{tikzpicture}
    \caption{Color legend for 2-level codes in \Cref{fig:all-lvl}.
    The system vertices are filled with solid color, the (single) environment vertex is filled white.}
    \label{fig:all-lvl-legend}
\end{figure}

Another family of codes for which we have a sufficiently large symmetry group are graph states associated with tree graphs.
In particular, we consider shallow trees which due to their restricted variety of local neighborhoods exhibit large numbers of equivalent vertices.
While analyzing thresholds for the depolarizing channel alone, we found a particularly promising class of trees to be those of depth $2$, i.e.\ a star graph, where each of the spikes is branching at most once.
For all those tested $2$-level graphs, we found that a single environment vertex attached at the root of the tree gives the highest thresholds in the case of depolarizing noise:
for $5<k\sys<25$, those gave the highest thresholds of any codes we tested of up to the given system vertex count (see \Cref{tab:all-lvl-dep}).

The size of the automorphism group for a tree can be computed recursively, starting at the root node; instead of giving a general formula we point out that the best-performing $2$-level codes in \Cref{fig:all-lvl-legend} appear to have a particularly large symmetry group, as all the branches are equivalent.

The most promising $2$-level codes we found are given in \Cref{fig:all-lvl-legend} with a threshold comparison against the single-letter case and a \${5-in-5} code in \Cref{fig:all-lvl}.
Evidently, the $2$-level codes perform particularly well for $X$-biased noise, outperforming both the optimal single-letter code (for smaller bias towards $X$-noise) and the \${5-in-5} cat code (for larger bias).
A comparison of all code families in \Cref{fig:code-familiy-comparison} shows that tree graphs outperform both repetition and cat codes for a large region of $X$-biased noise in the Pauli channel simplex.
We also list the thresholds of the $2$-level codes for depolarizing noise in \Cref{tab:all-lvl-dep}, comparing it to the \${1-in-5} and \${5-in-5} codes.
Finally, \Cref{tab:summary-num-thresholds} contains an overview of the best and worst thresholds for the three code families in the whole Pauli channel simplex (see also \Cref{fig:summary-thresholds}), as well as the best thresholds for depolarizing noise, BB84 channel, and the 2-Pauli channel.

\begin{table}[t]
	\centering
	\begin{tabular}{r l l lll}
		\toprule
		\multirow{2}{*}{code family} & \multirow{2}{*}{minimum} & \multirow{2}{*}{maximum} & \multicolumn{3}{c}{$\Delta p$} \\
		& & & $\mathbf{p}_{\mathrm{dep}}$ & $\mathbf{p}_{\mathrm{two}}$ & $\mathbf{p}_{\mathrm{BB84}}$ \\
		\midrule
		\${1-in-$k$}, $k=2,\ldots,60$ & \makecell[l]{ $-0.006256$ \\ \${1-in-2} } & \makecell[l]{ $0.026737$ \\ \${1-in-60} } & \makecell[l]{ $1.066\times 10^{-3}$ \\ \${1-in-5} } & \makecell[l]{$-0.414 \times 10^{-3}$ \\ \${1-in-5}} & \makecell[l]{ $2.080 \times 10^{-3}$ \\ \${1-in-7} } \\[4mm]
		concatenated$^\dagger$ & \makecell[l]{ $-0.007456$ \\ \${2-in-7} } & \makecell[l]{ $0.018999$ \\ \${3-in-7} } & \makecell[l]{$1.308 \times 10^{-3}$ \\ \${5-in-6} } & \makecell[l]{$-0.219\times 10^{-3}$ \\ \${5-in-6}} & \makecell[l]{ $ 2.194 \times 10^{-3}$ \\ \${5-in-6} } \\[4mm]
		2-level$^\ddagger$ & \makecell[l]{ $-0.010638$ \\ \${T$_{17}$} } & \makecell[l]{ $0.019395$ \\ \${T$_{21}$} } & \makecell[l]{ $1.260 \times 10^{-3}$ \\ \${T$_{21}$} } & \makecell[l]{ $-0.382 \times 10^{-3}$ \\ \${T$_{19}$} } & \makecell[l]{ $2.087 \times 10^{-3}$ \\ \${T$_{19}$} } \\
		\bottomrule
	\end{tabular}
	\caption{Error thresholds for select code families; minima and maxima are located at the blue and red dots in \Cref{fig:summary-thresholds}. The last three columns correspond to the thresholds for the depolarizing channel ($\mathbf{p}_{\mathrm{dep}}=(1-p,p/3,p/3,p/3)$), the 2-Pauli channel ($\mathbf{p}_{\mathrm{two}}=(1-p,p/2,0,p/2)$), and the BB84 channel ($\mathbf{p}_{\mathrm{BB84}}=((1-p)^2,p-p^2,p^2,p-p^2)$), respectively; they are located at the black dots from the center outwards in \Cref{fig:summary-thresholds} (see also \Cref{fig:surface-plot-expl}). Here, $\Delta p$ is the difference between the best threshold for the given code and the single-letter threshold given by the hashing bound \eqref{eq:hashing-bound}.\\
		${}^\dagger$The family includes all \${$n_1$-in-$n_2$}-codes, for $n_1=2,\ldots,5$  and $n_2=2,\ldots,7$ (except \${5-in-7}) (cf.\ \Cref{fig:cat codes}).
		${}^\ddagger$The family includes all 2-level codes given in \Cref{fig:all-lvl-legend}.
		For a list of select code thresholds see \Cref{tab:all-lvl-dep}.
	}
	\label{tab:summary-num-thresholds}
\end{table}

\subsection{Rate Comparisons}\label{sec:rates}

So far we have focused on a special form of superadditivity of coherent information where the error threshold is increased relative to the single-letter threshold.
This is a fundamental question in quantum information theory since positive quantum capacity of a channel guarantees the theoretical ability to correct errors (modeled by that channel), albeit at a potentially miniscule rate.
From a practical point of view, one is also interested in the \emph{magnitude} of the rate itself, which governs the overhead of an error correction code, i.e., how many physical qubits are needed to protect logical qubits from noise.
In this section we study how the rates of repetition codes, cat codes, and tree graph codes compare to the optimal single-letter rates given in \eqref{eq:hashing-bound} \emph{below} the threshold.

In order to visualize rates within the Pauli simplex displayed in \Cref{fig:antidegradability}, we evaluate the CI rate on layers aligned with the $y$ axis, but with a specific geared ratio of $x=fz$ for various chosen $z$ as shown in \Cref{fig:rate-layers}.

\Cref{fig:rates-A,fig:rates-B,fig:rates-C} show rate comparisons for the three code families---repetition, concatenated, and tree codes---vs.\ the single-letter CI.
We have seen in the sections above that repetition codes exhibit the largest threshold increases overall, especially for noise biased towards a high $X:Z$ ratio.
It is therefore not surprising that the rates of said codes---as compared to the single-letter rate---reach an advantages of $\approx 0.01$ in a significant volume below the noise threshold.
The volume of superadditivity decreases for the family of cat codes, and is smallest for tree graph codes.

The relatively small region of superadditivity in \Cref{fig:rates-A,fig:rates-B,fig:rates-C} is not surprising, as all code families studied in this paper (repetition, cat, and tree graph codes) are rank-2 codes: they have a single purifying qubit, and hence the marginal input state on the $k\sys$ channel qubits has rank at most 2 by Schmidt decomposition, and a maximal entropy of $1$.
These low-rank codes are well-suited for increasing noise thresholds, albeit at the expense of a poor rate of at most $1/k\sys$ for noiseless channels and decreasing further for noisy channels.
An interesting question is whether the region of superadditivity relative to the single-letter rate in \Cref{fig:rates-A,fig:rates-B,fig:rates-C} can be increased using higher-rank graph states, i.e., graphs with more than one purifying qubit.
This analysis would likely extend superadditivity effects to the mid-noise regime away from the threshold.
In the low-noise regime, where the channel is close to a noiseless channel, the results of \cite{LLS17} show that superadditivity can only be limited, and the quantum capacity is essentially given by the single-letter coherent information (up to small corrections in the noise parameter).
We leave an in-depth investigation of the mid-noise regime for future work.

\begin{figure}
	\hspace*{-0.9cm}
	\begin{minipage}{18cm}
		\includegraphics[width=\linewidth]{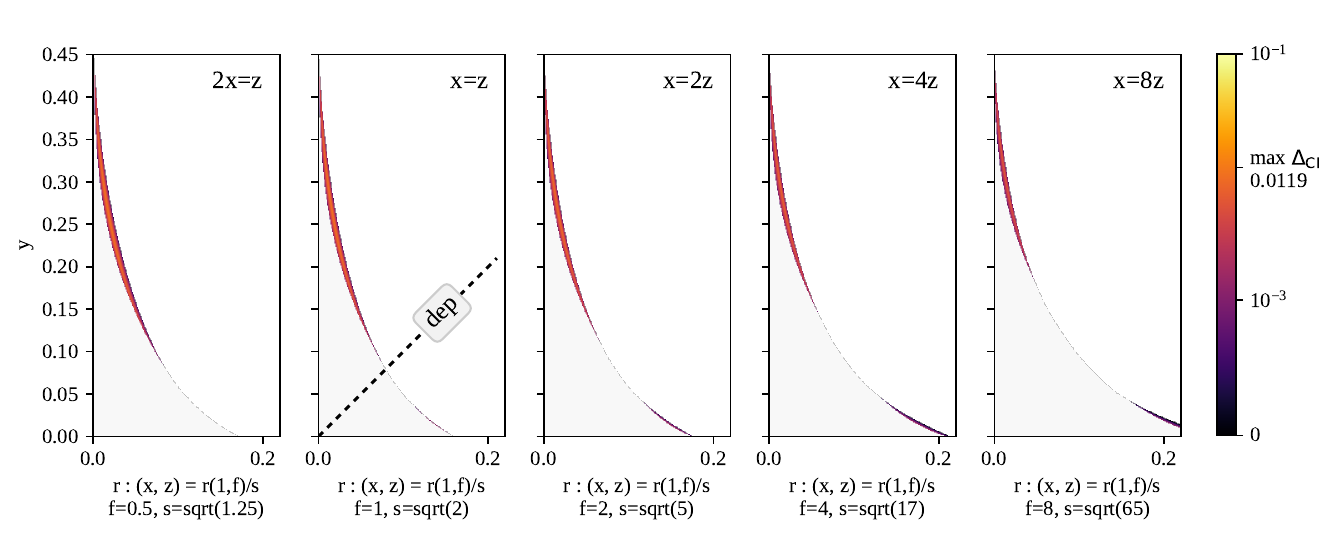}\\[.5cm]
	\end{minipage}
	\caption{
		Rate comparison of all repetition codes \${1-in-$k$} for $k=2,\ldots,60$ vs.\ the optimal single-letter rate.
		Shown is the CI difference of the two codes within planes that cut through the Pauli simplex (shown in \Cref{fig:rate-layers}), where the planes are parametrized by $(p_1,p_2,p_3) = (r/s,y,r f/s)$---i.e.\ they align with the $y$ axis, and are slanted with a specific ratio in the other two directions via $x=f z$ for specific choices of $f>0$ shown in the top right corner of each plot.
		We only draw the area where the rate difference is positive; the color scale and maximum rate difference are given on the right.
		Shaded in light grey is the area where the repetition code CI is non-negative but below the single-letter rate.
		For $x=z$, the dashed line indicates the family of depolarizing channels.}
	\label{fig:rates-A}
\end{figure}

\begin{figure}
	\hspace*{-0.9cm}
	\begin{minipage}{18cm}
		\includegraphics[width=\linewidth]{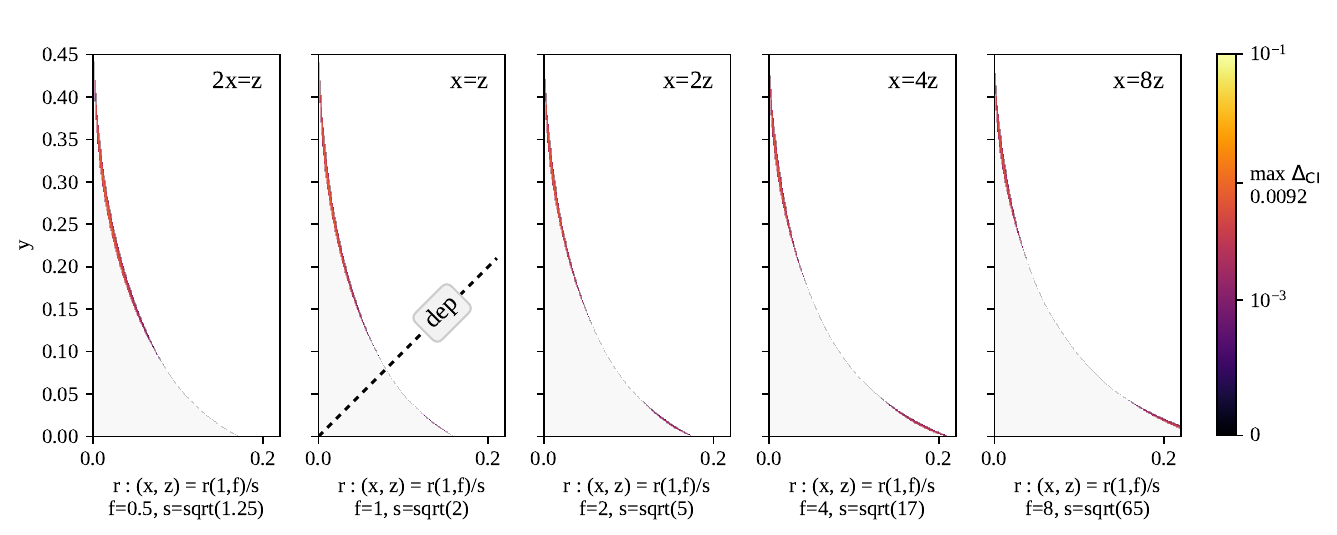}\\[.5cm]
	\end{minipage}
	\caption{All concatenated codes from \Cref{fig:cat codes} vs.\ single-letter rate.
		For a plot description see \Cref{fig:rates-A}.}
	\label{fig:rates-B}
\end{figure}

\begin{figure}
	\hspace*{-0.9cm}
	\begin{minipage}{18cm}
		\includegraphics[width=\linewidth]{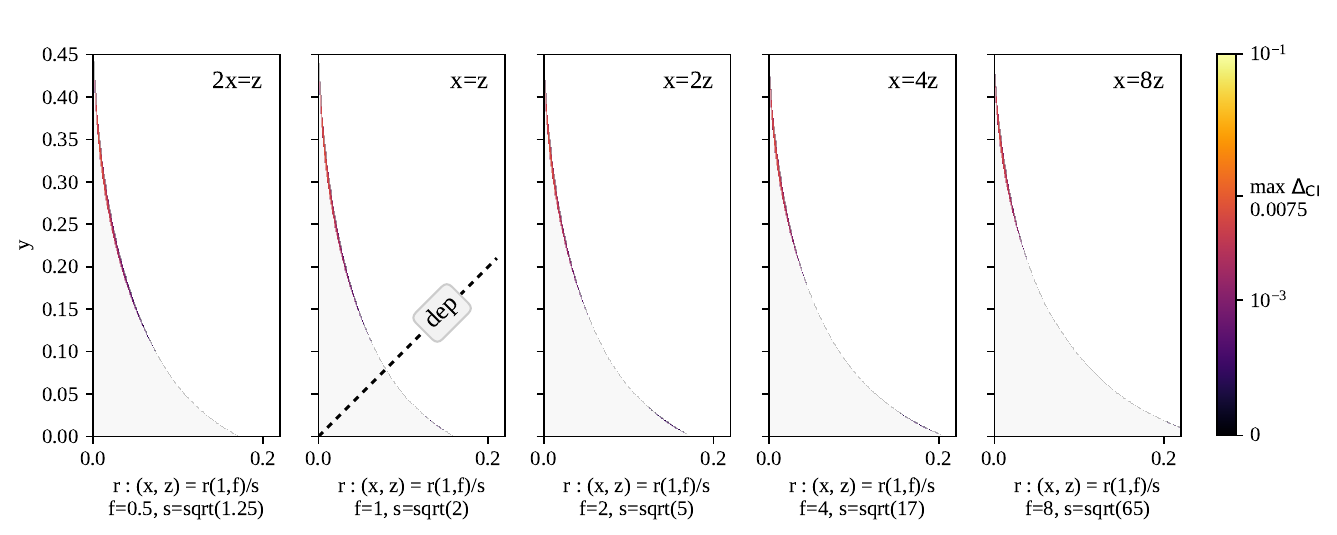}\\[.5cm]
	\end{minipage}
	\caption{All tree codes from \Cref{fig:all-lvl} vs.\ single-letter CI.
		For a plot description see \Cref{fig:rates-A}.}
	\label{fig:rates-C}
\end{figure}

\begin{figure}[ht]
	\centering
	\begin{tikzpicture}
  \begin{axis}[
    width=.8\linewidth,
    view={135}{30},
    xmin=0,
    xmax=0.5,
    ymin=0,
    ymax=0.5,
    zmin=0,
    zmax=0.5,
    axis equal image,
    scale=1,
    myplane/.style={
        surf,
        domain=0:0.5,
        samples=6,
        y domain=0:0.5,
        samples y=6,
        shader=faceted interp,
        z buffer=sort,
        opacity=1
    },
    colormap/viridis,
    point meta=y,
    xtick={0,.1,.2,.3,.4,.5},
    ytick={0,.1,.2,.3,.4},
    ztick={.1,.2,.3,.4,.5},
    xlabel={$p_1$},
    ylabel={$p_2$},
    zlabel={$p_3$},
    ]
    \foreach \f in {.125,.25,.5,8,4,2,1}{
        \addplot3[myplane] ({x / sqrt(1 + \f*\f)}, y, {\f*x / sqrt(1 + \f*\f)});
    };
  \end{axis}
\end{tikzpicture}
	\caption{Layer cuts through the simplex of Pauli channels; on the surfaces shown, we evaluate the CI for various graph state codes. The results are shown in \Cref{fig:rates-A,fig:rates-B,fig:rates-C}.
	}
	\label{fig:rate-layers}
\end{figure}
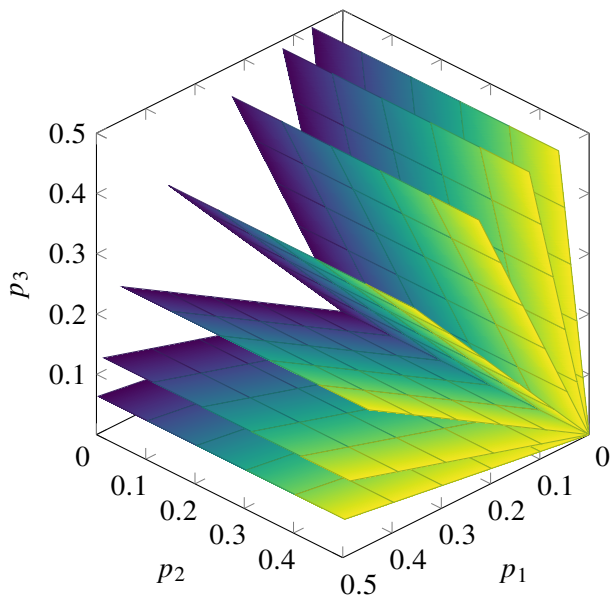

\section{Discussion}\label{sec:discussion}

\begin{figure}[t]
    \centering
	\subplotAA{all-rep.vs.single}[all rep vs.\\single-letter]
	\subplotAA{all-cat.vs.single}[all cat vs.\\single-letter]
	\subplotAA{all-lvl.vs.single}[all 2-level vs.\\single-letter]\\
	%
	\raisebox{5pt}{$\Delta x$} \smallColorScale\phantom{ $\Delta x$}
	\caption{Threshold plots for all repetition codes (up to \${1-in-60}), concatenated codes (all \${$n_1$-in-$n_2$}-codes, for $n_1=2,\ldots,5$  and $n_2=2,\ldots,7$ for $n_1\le 3$, otherwise $n_2=2,\ldots,5$), and 2-level codes (given in \Cref{fig:all-lvl-legend}), vs.\ the single-letter threshold.
		The black dots indicate the location of the depolarizing, BB84, and 2-Pauli channel families (from center to outside); the red and blue dots the respective global threshold maxima and minima (cf.\ \Cref{tab:summary-num-thresholds}).
	}
	\label{fig:summary-thresholds}
\end{figure}

\begin{figure}[t]
    \centering
    \surfacePlot{10cm}{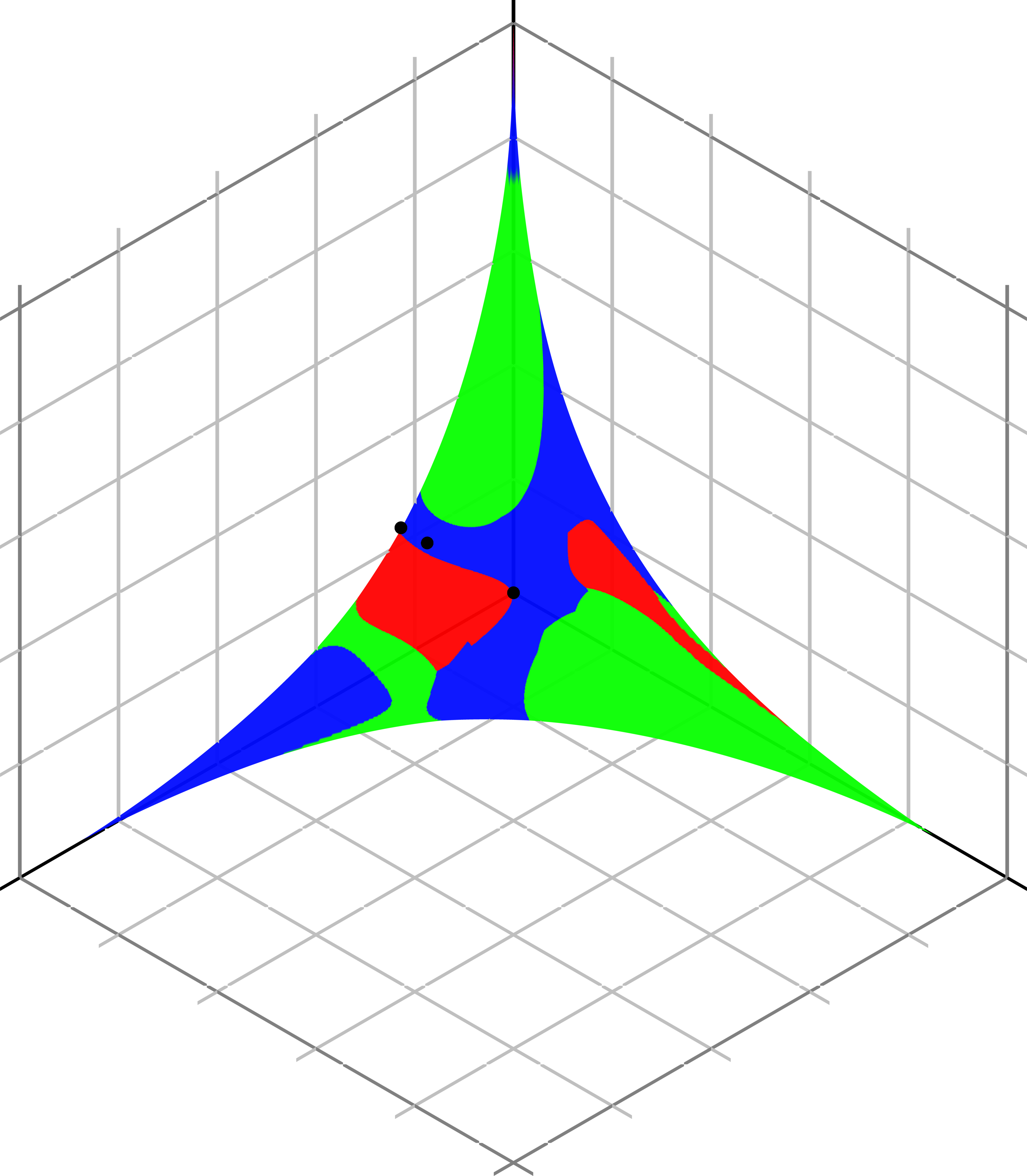}[1][%
        \node (dep) at (0,-.01) {};
        \node (bb) at (-.083,.0405) {};
        \node (pp) at (-.1085,.063) {};
        \node[scale=.5, fill=black, text=white, rectangle callout, callout absolute pointer={(dep)}, below of=dep] {dep};
        \node[scale=.5, fill=black, text=white, rectangle callout, callout absolute pointer={(bb)}, below of=bb] {BB84};
        \node[scale=.5, fill=black, text=white, rectangle callout, callout absolute pointer={(pp)}, above of=pp] {2P};
        \node[right] at (.5,-.3) {$0$};
        \node[right] at (.5,.26) {$0.5$};
        \node[left] at (-.5,-.3) {$0$};
        \node[left] at (-.02,-.58) {$0.5$};
        \node[right] at (.02,-.58) {$0.5$};
    ][]
    \caption{Repetition (green), concatenated (blue) and $2$-level (red) code family thresholds compared.
    Especially around the $2$-Pauli, BB84 and depolarizing channels biased towards $x$ noise our new code family outperforms the other two.}
    \label{fig:code-familiy-comparison}
\end{figure}

In this work, we studied various classes of graph state codes in the context of Pauli noise channels, and how their coherent information (CI) thresholds relate to the noise parameters.
For graphs with symmetries, we derived an algorithm to obtain an analytical expression for the CI, exploiting the graph state's automorphism group to significantly speed up the computation (\Cref{th:symm-lambda-state-sum,th:symm-lambda-state-tr-sum}); the speedup is most pronounced for graphs with large symmetry groups.

To quantitatively assess the codes' error correction potential in various noise regimes, we numerically evaluated the CI threshold over the entire simplex of Pauli channels for individual noise probabilities $\le 1/2$.
To this end, we developed a numerical bisection method that allows us to compute these thresholds in a quick fashion, based on the analytic CI expression obtained before.

The most well-studied Pauli channels are unbiased depolarizing noise, the BB84 channel (corresponding to independent bit and phase flips), and 2-Pauli noise (where only two out of the three possible Pauli errors occur).
For these channel models significant superadditivity of coherent information is achieved e.g.\ by nested concatenation of simpler codes \cite{DSS98,SS07,FW08}.
For more general biased noise, we find much more significant superadditivity; for example, repetition codes achieve a maximum possible superadditivity of $\delta p\approx 0.0267$ relative to the single-letter threshold for heavily biased noise.
This is almost two orders of magnitude beyond the superadditivity of $\delta p \approx 0.0007$ achieved by the \${5-in-5} code.
We find a maximum superadditivity of $\approx 0.0190$ for concatenated codes, and $\approx 0.0194$ for the codes based on $2$-level tree graphs introduced in this paper.
These findings are summarized in \Cref{fig:summary-thresholds}.
All maximum threshold increases were reached for $y$-biased noise, i.e., the region where $Y$ errors dominate.
Note however that a different assignment of the graph state stabilizers achieves arbitrary rotations and reflections in the $xyz$-hyperplane (as explained in \Cref{sec:rep-codes}), and thus our results apply to any type of biased noise.

For a given stabilizer state affected by i.i.d.\ Pauli noise, error syndrome measurements give way to a decomposition of the coherent information in terms of \emph{induced} single-qubit Pauli channels whose bias is determined by the structure of the stabilizer group \cite{DSS98,SS07}.
On the other hand, it is less obvious to determine which code performs best for a given Pauli channel.
Our results presented in \Cref{sec:results} allow us to answer this question for the code families (repetition, cat, tree graph codes) discussed in this paper, which we summarize in \Cref{fig:code-familiy-comparison}.
This plot shows that around depolarizing, BB84 and 2-Pauli noise the concatenated code family has the highest threshold; towards the $Y$ and $Z$ legs repetition codes are best.
Yet for large intermediate regions in the $ZX$ and $ZY$ plane, the $2$-level codes perform best.

\subsection{Noise Models in NISQ Devices}\label{sec:nisq}

The context of codes we discussed in this paper is that of communication, i.e.\ where entanglement is generated by sending part of a graph state through a noisy channel.
States that exhibit positive coherent information after having been affected by noise can e.g.\ be used to extract a secret key in quantum key distribution protocols \cite{BB84,Bruss6State,ShorPreskill,Lo6State,SRS08,KernRenes}.
More generally, a central goal in designing quantum storage and quantum computation devices is to maintain coherence of entangled states affected by environmental noise.
Usually, the noise is assumed to be uniform over the Pauli errors $X$, $Y$, $Z$, i.e., unbiased (depolarizing) noise.

A fundamental question is how reasonable this assumption of unbiased noise is.
For all but the most trivial tasks, data will have to be passed through a quantum device which, if not fully error-corrected, will be noisy.
This is not an issue if the device itself features full error correction: data is then guaranteed to stay in a coherent state throughout.
Yet the near-term prospects of quantum computing---dubbed the NISQ (noisy intermediate-scale quantum) era \cite{nisq}---comprise devices of $10$-$100$ qubits of various interaction topologies,
and with circuit depths limited to $\approx 100$ gates.
On such devices, achieving full-fledged error correction is unfeasible, and will remain a challenge for some time to come.

The way these limited near-term capabilities factor into the noise model is twofold.
First, because there is no error correction, individual qubits will experience different levels and different types of noise originating from cross-talk across the chip, e.g.\ depending on whether a qubit is located towards the center of the chip, or at an edge.
On average, it is thus conceivable that the noise channel modeling the device is not unbiased depolarizing noise, but dominated to some extent by a different Pauli error ratio---or even non-Pauli noise obtained from slightly perturbing Pauli channels; an analysis of the error thresholds of such close-to-Pauli channels was partly carried out in \cite{Jackson17}.

Secondly, the way physical qubits are implemented factors into the induced gate noise.
Any gate-level unitary on qubits is truly a hardware-level quantum channel that implements the coupling by some control pulse; depending on how the coupling is implemented, it is not obvious why the qubit-level noise should be unbiased.
In fact, many proposed architectures for quantum computation---such as superconducting qubits \cite{superconducting}, quantum dots \cite{quantumdots}, or trapped ions \cite{trappedions,Ospelkaus2008}---have noise biased towards one of the Pauli errors.

For such biased noise models, it is known that one can obtain significantly higher error correction thresholds than for unbiased noise, e.g.\ with a tensor network decoder for the surface code \cite{Tuckett2018}.
Yet even without error correction there exist circuit level procedures that render biased noise into weak but unbiased noise \cite{Aliferis2008,Wallman2016}, so it is arguable whether other non-depolarizing noise models
are worth the effort to construct tailored codes for them. 
However, the incurred overhead of such a transpiling procedure for error mitigation is prohibitive for any near-term quantum application with its limited circuit depth: any prolongation of the computation
would render the circuit much more noisy than initially, undermining any potential benefit from unbiasing the noise in first place.

So what type of block code is optimal for a specific device and communication task setup?
In particular if the overall noise level is close to the zero-CI threshold, this choice will begin to matter.
For instance, if the communication task is that of QKD and we choose to implement the BB84 protocol, the quantum devices used for encoding and decoding might shift the overall channel away from an exact BB84 channel,
and bias it e.g.\ towards a dominating $Z$ error.
In that case, concatenated codes would be a good choice (see \Cref{fig:code-familiy-comparison}).
On the other hand, if $X$ noise gets slightly more pronounced---e.g.\ a relative increase of $\approx 5\%$---then picking the \${T$_{21}$} code is a better option.

\subsection{Open Problems}
In light of the discussion above about more realistic noise models, it would be interesting to devise an algorithm computing the coherent information of graph states under non-i.i.d.\ channels, which potentially include memory effects between qubits such as nearest neighbor interactions.
Depending on the interaction topology, this type of noise is likely to lack the full permutation symmetry of i.i.d.\ channels, and hence poses a challenge to exploiting symmetries to speed up the computation of the coherent information as in \Cref{alg:symm-lambda} in \Cref{sec:symmetries}.

We conclude with a few open questions that will be the subject of further study.
As mentioned in \Cref{sec:rates}, we focused our analysis on low-rank graph state codes that increase error thresholds relative to the single-letter threshold, noting that they also boost coherent information rates for Pauli channels in a small region below the error threshold.
In order to study more general rate increases for Pauli channels in the mid-noise regime, it would be interesting to consider other families of graph states with a larger rank, i.e., for which $k\env \approx k\sys$.
Our analysis of error thresholds for Pauli channels can also readily be extended to other graph state families beyond cat codes and tree graph codes.
In both cases above, the challenge is to find graphs that have both large symmetry groups and interesting error-correcting properties. 
While the former is necessary for our algorithm to yield a computational speed-up, the latter guarantees that these graph states achieve superadditivity of coherent information.

\section*{Acknowledgements}
We would like to thank Graeme Smith for helpful feedback, and Brendan McKay for many enlightening conversations about computational graph theory and his excellent \tool{nauty} toolkit.
J.\,B.\ would like to thank the Draper's Research Fellowship at Pembroke College for their support.
F.\,L.\ acknowledges support from NSF Grant No.~PHY 1734006 and the hospitality of DAMTP at the University of Cambridge.
We are grateful for an \emph{AI Grant} that enabled us to perform the numerical studies within this paper.

\clearpage
\linespread{1}
\printbibliography

\linespread{1.13}
\appendix
\section{LU-Equivalence of an Arbitrary Stabilizer State to~a~Graph~State}\label{app:LU-equivalence}
In this appendix we explain in detail how to obtain the graph state that is LU-equivalent to a given stabilizer state. 
In the following, we adopt the notation of \cite{EEC08}.

To show the equivalence, it is convenient to use the binary representation of the Pauli group, in which the Pauli matrices are assigned vectors in $\bZ_2^2$ as follows:
\begin{align}
r(I) &\coloneqq (0,0) & r(X) &\coloneqq (1,0) & r(Y) &\coloneqq (1,1) & r(Z) & \coloneqq (0,1).\label{eq:binary-identification}
\end{align}
Ignoring phases, multiplication of Pauli matrices corresponds to addition in $\bZ_2$ in the binary representation.
Using this correspondence we assign to an element $P\in \cP_n$ in the $n$-qubit Pauli group a vector in $v\in\bZ_2^{2n}$ by setting $v_i = r(P_i)_1$ and $v_{i+n} = r(P_i)_2$.
As an example, the Pauli operator $X^1 Z^2 Y^4$ corresponds to the vector $v = (1 0 0 1 | 0 1 0 1)$.

For a stabilizer state on $n$ qubits with $n$ stabilizer generators, consider now the $n \times 2n$ \emph{generator matrix} whose rows are the binary representations of the stabilizer generators.
Multiplying a generator by another one does not change the stabilizer group; by the discussion above, this corresponds to adding rows modulo 2 in the generator matrix.
Moreover, exchanging the labels of the qubits $i$ and $j$ corresponds to simultaneously exchanging the columns $i$ and $j$ and $n+i$ and $n+j$.
Using these two operations, every generator matrix corresponding to a set of $n$ stabilizer generators describing a given stabilizer state on $n$ qubits can be brought into the \emph{canonical form} (or \emph{standard form}) \cite{NC00}
\begin{align}
\left(\begin{array}{cc|cc}
I_k & A & B & 0\\
0 & 0 & A^T & I_{n-k}
\end{array}\right),\label{eq:canonical-form}
\end{align}
where $B$ is a binary symmetric $k\times k$-matrix, $A$ is an arbitrary binary $k\times (n-k)$-matrix, and $I_j$ denotes the $j\times j$-identity matrix.
The canonical form of a graph state $\ket{\Gamma}$ with stabilizers \Cref{eq:graph-state-stabilizers} is given by
\begin{align}
\left( \begin{array}{c|c}
I_n & \Gamma
\end{array}\right),\label{eq:canonical-form-graph-state}
\end{align}
where by $\Gamma$ we also denote the \emph{adjacency matrix} of the graph $\Gamma$, i.e., $\Gamma_{ij} = 1$ if $\{ i,j\}\in E$ and $\Gamma_{ij}=0$ otherwise.

To convert a given generator matrix in canonical form \Cref{eq:canonical-form} to a graph state with generator matrix \Cref{eq:canonical-form-graph-state}, we recall the Clifford unitaries
\begin{align*}
H &\coloneqq \frac{1}{\sqrt{2}} \begin{pmatrix}
1 & 1\\1 & -1
\end{pmatrix} & 
S &\coloneqq \begin{pmatrix}
1 & 0\\ 0 & i
\end{pmatrix},
\end{align*}
which satisfy the relations $H X H = Z$ (and hence $X = HZH$) and $SYS^\dagger = -X$.
Applying the Hadamard matrix $H$ to qubit $j$ corresponds to exchanging the columns $j$ and $n+j$ in the generator matrix. 
Hence, applying $H$ to the last $n-k$ qubits achieves the transformation
\begin{align*}
\left(\begin{array}{cc|cc}
I_k & A & B & 0\\
0 & 0 & A^T & I_{n-k}
\end{array}\right)
\longrightarrow
\left(\begin{array}{cc|cc}
I_k & 0 & B & A\\
0 & I_{n-k} & A^T & 0
\end{array}\right).
\end{align*}
A valid adjacency matrix of a simple graph has $0$'s on the diagonal, and hence we are left with removing any $1$'s on the diagonal of $B$.
This can be achieved by applying $S$ to all qubits $j\leq k$ for which $B_{jj}=1$, and we are done.
Since the procedure outlined above only involves single-qubit Clifford unitaries, we can in fact assert that any stabilizer state is local-Clifford (LC) equivalent to a graph state.

\section{Subroutines for Full Coherent Information Algorithm}\label{app:subroutines-for-full-alg}

\Cref{alg:coherent-information-full-algorithm} makes use of the following subroutines:
\begin{itemize}
    \item \textproc{GetUSubsets}($k,r$)\\
    Input: $k,r\in\bN$\\
    Output: Three $(k+r)\times 4^k$-matrices, whose columns are the binary vector representations of all possible tuples $(U_1,U_2,U_3)$ satisfying $U_i\subset \lbrace 1,\dots,k\rbrace$ and $U_i\cap U_j=\emptyset$ for $i\neq j$. The $U_i$ are regarded as subsets in $\lbrace 1,\dots,k+r\rbrace$, and hence the $k$-binary vectors are padded with $r$ zeroes.
    \item \textproc{GetUSubsetsCard}($k$)\\
    Input: $k\in\bN$\\
    Output: Three row vectors of length $4^k$ storing the cardinalities of all possible tuples $(U_1,U_2,U_3)$ satisfying $U_i\subset \lbrace 1,\dots,k\rbrace$ and $U_i\cap U_j=\emptyset$ for $i\neq j$. 
    \item \textproc{BinaryToDecimal}($M$)\\
    Input: $M\in \cM(\bF_2)$\\
    Output: Row vector whose $i$-th entry is the decimal number corresponding to the binary string in the $i$-th column of $M$.
    \item \textproc{Subsets}($k$)\\
    Input: $k\in\bN$\\
    Output: $k\times 2^k$-matrix whose columns are the binary vectors representing the subsets of $\{ 1,\dots,k\}$.
    \item \textproc{PermVector}($v,\pi$)\\
    Input: $v\in\bR^n$, $\pi\in S_n$\\
    Output: $v' = (v_{\pi(1)},\dots,v_{\pi(n)})$
    \item \textproc{ShannonEntropy}($\mathbf{q}$)\\
    Input: Probability distribution $\mathbf{q} = (q_1,\dots,q_n)$\\
    Output: $-\sum_{i=1}^n q_i \log_2 q_i$.
\end{itemize}

\section{Antidegradability of quantum channels}\label{sec:proof-monotonicity}
In order to prove that antidegradable quantum channels have vanishing quantum capacity, we introduce the \emph{quantum relative entropy} $D(\rho\|\sigma)$, defined for a quantum state $\rho$ and a positive semidefinite operator $\sigma$ as
\begin{align*}
    D(\rho\|\sigma) = \begin{cases} \tr(\rho\log\rho - \rho\log\sigma) & \text{if }\supp\rho\subset\supp\sigma\\
    \infty & \text{otherwise},\end{cases}
\end{align*}
where $\supp\rho \coloneqq (\ker\rho)^\perp$.
The relative entropy satisfies the data-processing inequality \cite{lindblad1975completely}:
For any quantum channel $\cN$, 
\begin{align}
    D(\rho\|\sigma) \geq D(\cN(\rho)\|\cN(\sigma)).\label{eq:dpi}
\end{align}
The coherent information $I_c(\psi,\cN)$ of a pure state $\psi_{RA}$ and a quantum channel $\cN\colon A\to B$ can be expressed in terms of the quantum relative entropy as
\begin{align*}
    I_c(\sigma,\cN) = D(\cN(\psi_{RA}) \| \one_R\ox \cN(\psi_A) ).
\end{align*}

Recall that a channel $\cN\colon A\to B$ with complementary channel $\cN^c\colon A\to E$ is antidegradable if there exists another quantum channel $\cA\colon E\to B$ such that $\cN = \cA\circ \cN^c$.
Antidegradable channels always have non-positive coherent information on pure states $\psi$:
\begin{align}
    I_c(\psi,\cN) = I_c(\psi, \cA\circ\cN^c) \leq I_c(\psi, \cN^c) = -I_c(\psi,\cN),\label{eq:adg}
\end{align}
where the inequality follows from the data-processing inequality \eqref{eq:dpi}.
For the last equality, recall the definition of the complementary channel $\cN^c(\cdot) = \tr_B V\cdot V^\dagger$, where $V\colon \cH_A\to\cH_B\otimes \cH_E$ is a channel isometry satisfying $\cN(\cdot) = \tr_E V\cdot V^\dagger$.
The marginals of a pure bipartite quantum state have the same spectrum by Schmidt decomposition, and applying this argument to the pure state $(\one_R\otimes V)\ket{\psi}_{RA}$ on the systems $RBE$ yields the identity
\begin{align*} 
I_c(\psi,\cN) = S(\cN(\psi_A)) - S(\cN(\psi_{RA})) = S(\cN^c(\psi_{RA})) - S(\cN^c(\psi_A)) = -I_c(\psi, \cN^c).
\end{align*}
\Cref{eq:adg} shows that $I_c(\psi,\cN) \leq 0$ for all $\psi$.
Since $\cN^{\ox n}$ is antidegradable for all $n\in\bN$ if $\cN$ is, we also have $I_c(\psi_n,\cN^{\ox n}) \leq 0$ for all $n\in\bN$ and pure states $\psi_n$, and hence $Q(\cN)=0$ for antidegradable $\cN$ by the quantum capacity theorem in \Cref{eq:quantum-capacity}.

We now give the proof of \Cref{lem:x-param-monotonicity}:

\begin{proof}[Proof of \Cref{lem:x-param-monotonicity}]
    	The necessary and sufficient condition \eqref{eq:pauli-antidegradable} for antidegradability of the Pauli channel $\cN_{\mathbf{p}_x}$ with $\mathbf{p}_x = \left(1-x, x p_1, x p_2, x p_3\right)$ reads
    	\begin{align*}
    	1 \geq 2\left((1-x)^2 + x^2(p_1^2+p_2^2+p_3^2)\right) - 8\sqrt{(x^3-x^4) p_1p_2p_3} \eqqcolon f(x).
    	\end{align*}
    	That is, $\cN_{\mathbf{p}_x}$ is antidegradable if and only if $f(x) \leq 1$.
    	We have $f(0) = 2$ for any $(p_1,p_2,p_3)$ (note that $\cN_{\mathbf p_0}$ is the identity channel for any $(p_1,p_2,p_3)$).
    	First, let $(p_1,p_2,p_3)$ be deterministic, e.g., $(p_1,p_2,p_3) = (1,0,0)$.
    	Then $f(x) = 2((1-x)^2+x^2) = 1$ if and only if $x=1/2$.
    	
    	Now, let $(p_1,p_2,p_3)$ be non-deterministic, then $f(1/2) <1$:
    	\begin{align*}
    	f(1/2) &= 2\left( \frac{1}{4} + \frac{1}{4} (p_1^2+p_2^2+p_3^2)\right) - 2\sqrt{p_1p_2p_3}\\
    	&\leq \frac{1}{2}\left(1 + p_1^2+p_2^2+p_3^2\right)\\
    	&< \frac{1}{2}\left(1 + (p_1+p_2+p_3)^2\right) = 1.
    	\end{align*}
    	Moreover, a similar argument shows that 
    	\begin{align*}
    	f'(x) \leq 2\left(-2(1-x) + 2x(p_1^2 + p_2^2 + p_3^2)\right) < -4 + 8x \leq 0
    	\end{align*}
    	for $0\leq x\leq 1/2$.
    	Hence, $f(0)>1$, $f(1/2)<1$, and $f$ is strictly monotonically decreasing on $[0,1/2]$, which proves the lemma.
\end{proof}

\section{Exhaustive Search}\label{sec:exhaustive}
\begin{figure}
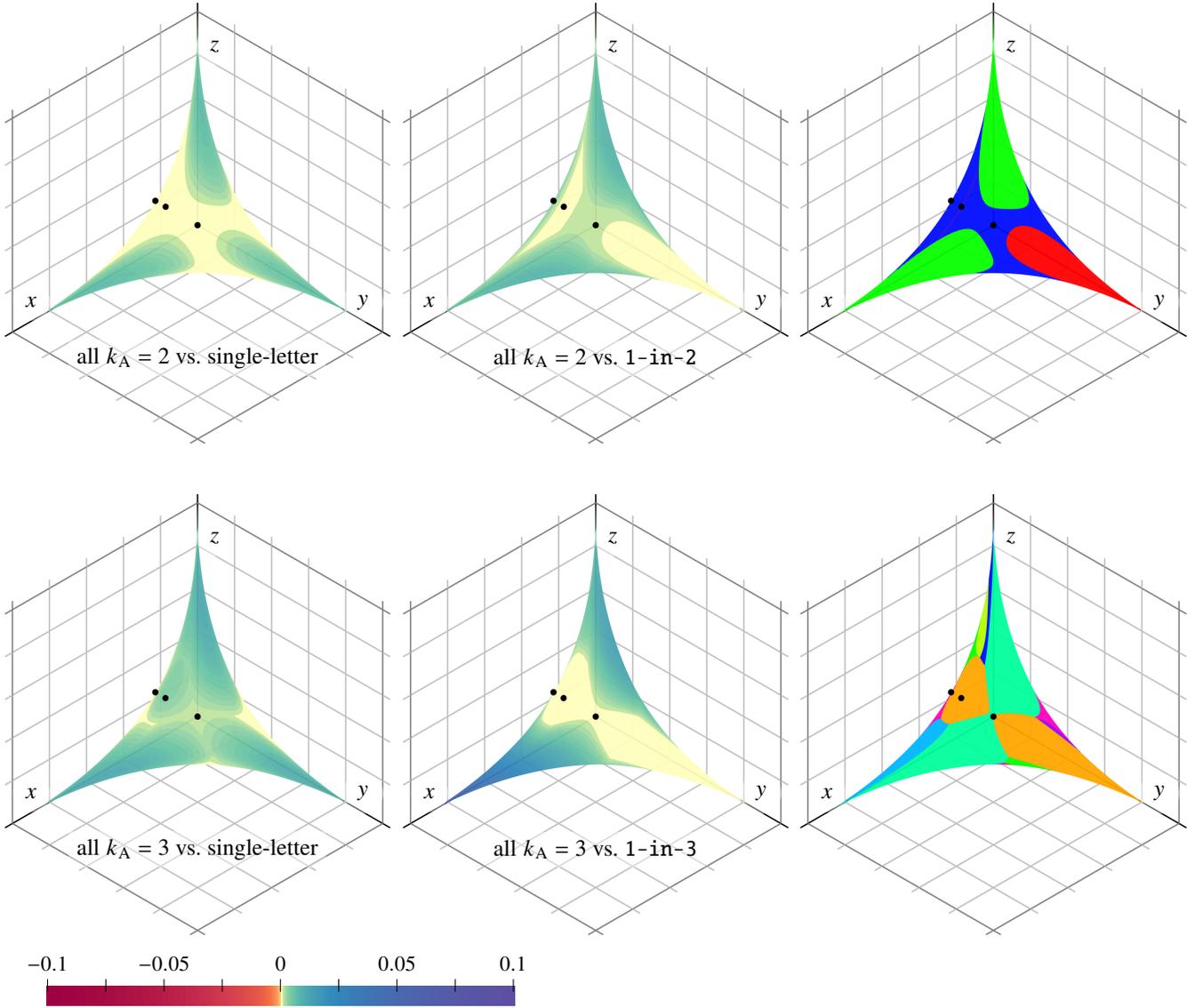

	\hspace*{-0.9cm}
	\begin{minipage}{18cm}
		\centering
		\subplotAA{all2.vs.single}[all $k\sys=2$ vs.\ single-letter]
		\subplotAA{all2.vs.1-in-2}[all $k\sys=2$ vs.\ \${1-in-2}]
		\subplotAA{all2.swatch}
		
		\subplotAA{all3.vs.single}[all $k\sys=3$ vs.\ single-letter]
		\subplotAA{all3.vs.1-in-3}[all $k\sys=3$ vs.\ \${1-in-3}]
		\subplotAA{all3.swatch}
	\end{minipage}
	\smallColorScale
	\caption{All possible graph codes with connected graphs for $k\sys \in \{ 2,3 \}$, continued in \Cref{fig:all-codes-B} for $k\sys \in \{4,5\}$---in particular, the plots for $k\sys=3$ do \emph{not} contain the smaller graph codes.
		The axes' extent is the interval $[0,1/2]$, and the three black dots (from center to left) are the loci of depolarizing, BB84, and 2-Pauli channel families, respectively---for more details see \Cref{fig:surface-plot-expl}.
		The first two columns are colored according to the relative threshold of the best code in the code family: the first column shows the threshold vs.\ the single-letter threshold, the second column vs.\ the threshold of a \${1-in-$k$} repetition code.
		The color scale indicates the distance between noise threshold and single-letter threshold.
		The rightmost column shows the same surfaces indexed by colors referring to the highest CI threshold graph state codes; for the corresponding color legend see \Cref{fig:all-codes-legend-2,fig:all-codes-legend-3}, respectively.
	}
	\label{fig:all-codes-A}
\end{figure}

\begin{figure}
	\hspace*{-0.9cm}
	\begin{minipage}{18cm}
		\centering
		\subplotAA{all4.vs.single}[all $k\sys=4$ vs.\ single-letter]
		\subplotAA{all4.vs.1-in-4}[all $k\sys=4$ vs.\ \${1-in-4}]
		\subplotAA{all4.swatch}
		
		\subplotAA{all5.vs.single}[all $k\sys=5$, $k\env\leq 3$ vs.\ single-letter]
		\subplotAA{all5.vs.1-in-5}[all $k\sys=5$, $k\env\leq 3$ vs.\ \${1-in-5}]
		\subplotAA{all5.swatch}
	\end{minipage}
	\caption{Continuation of \Cref{fig:all-codes-A}, for $k\sys\in\{4,5\}$ and $k\tot\le 8$ (in both cases).
		The color legends for the last column are in \Cref{fig:all-codes-legend-4,fig:all-codes-legend-5}.
	}
	\label{fig:all-codes-B}
\end{figure}

\begin{figure}
	\centering
	\includegraphics[width=.4\linewidth]{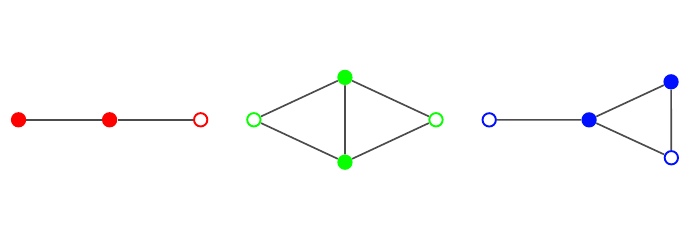}
	\caption{Color legend for all $k\sys=2$ codes in \Cref{fig:all-codes-A}.
		The system vertices are filled with solid color, the environment vertices in white.}
	\label{fig:all-codes-legend-2}
\end{figure}
\begin{figure}
	\centering
	\includegraphics[width=\linewidth]{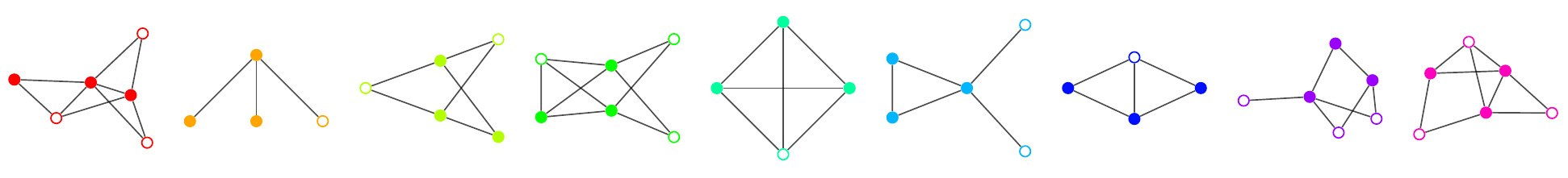}
	\caption{Color legend for all $k\sys=3$ codes in \Cref{fig:all-codes-A}.
		The system vertices are filled with solid color, the environment vertices in white.}
	\label{fig:all-codes-legend-3}
\end{figure}
\begin{figure}
	\centering
	\includegraphics[width=\linewidth]{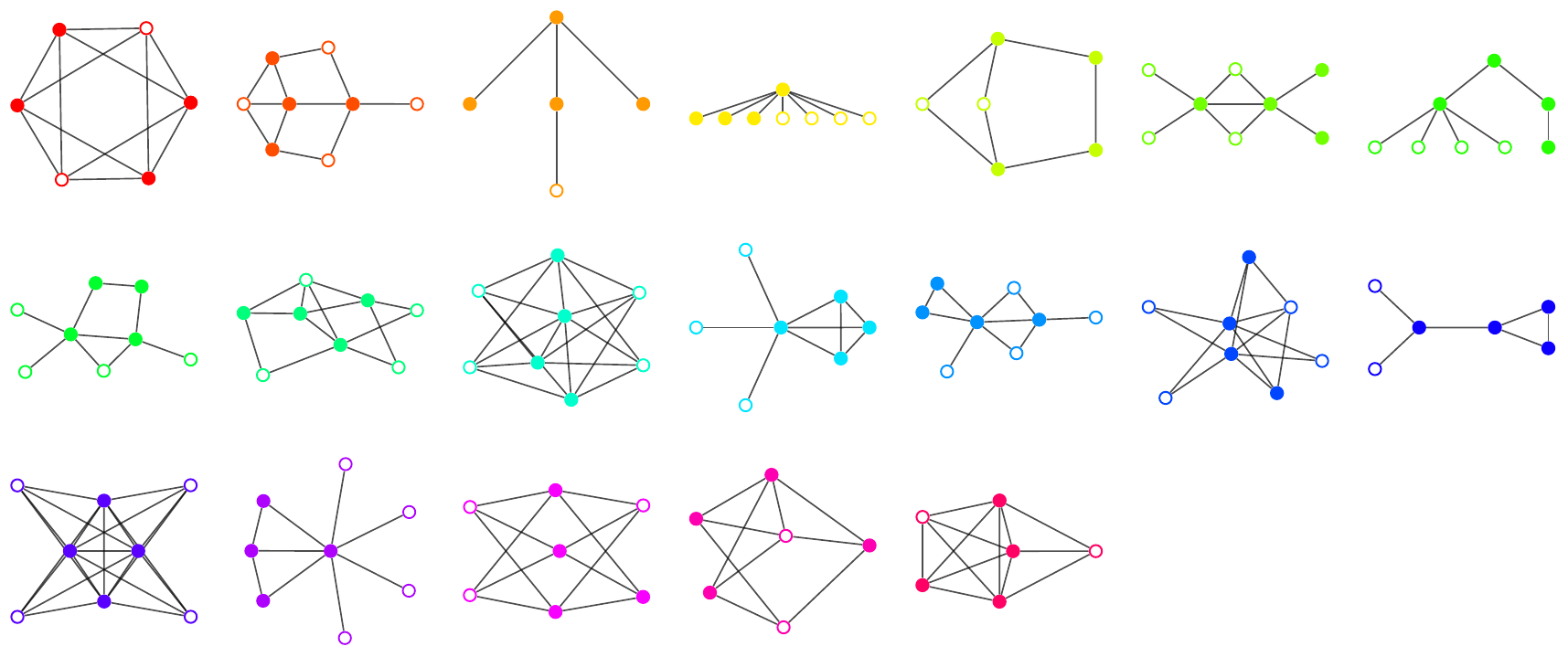}
	\caption{Color legend for all $k\sys=4$ codes in \Cref{fig:all-codes-B}.
		The system vertices are filled with solid color, the environment vertices in white.}
	\label{fig:all-codes-legend-4}
\end{figure}
\begin{figure}
	\centering
	\includegraphics[width=\linewidth]{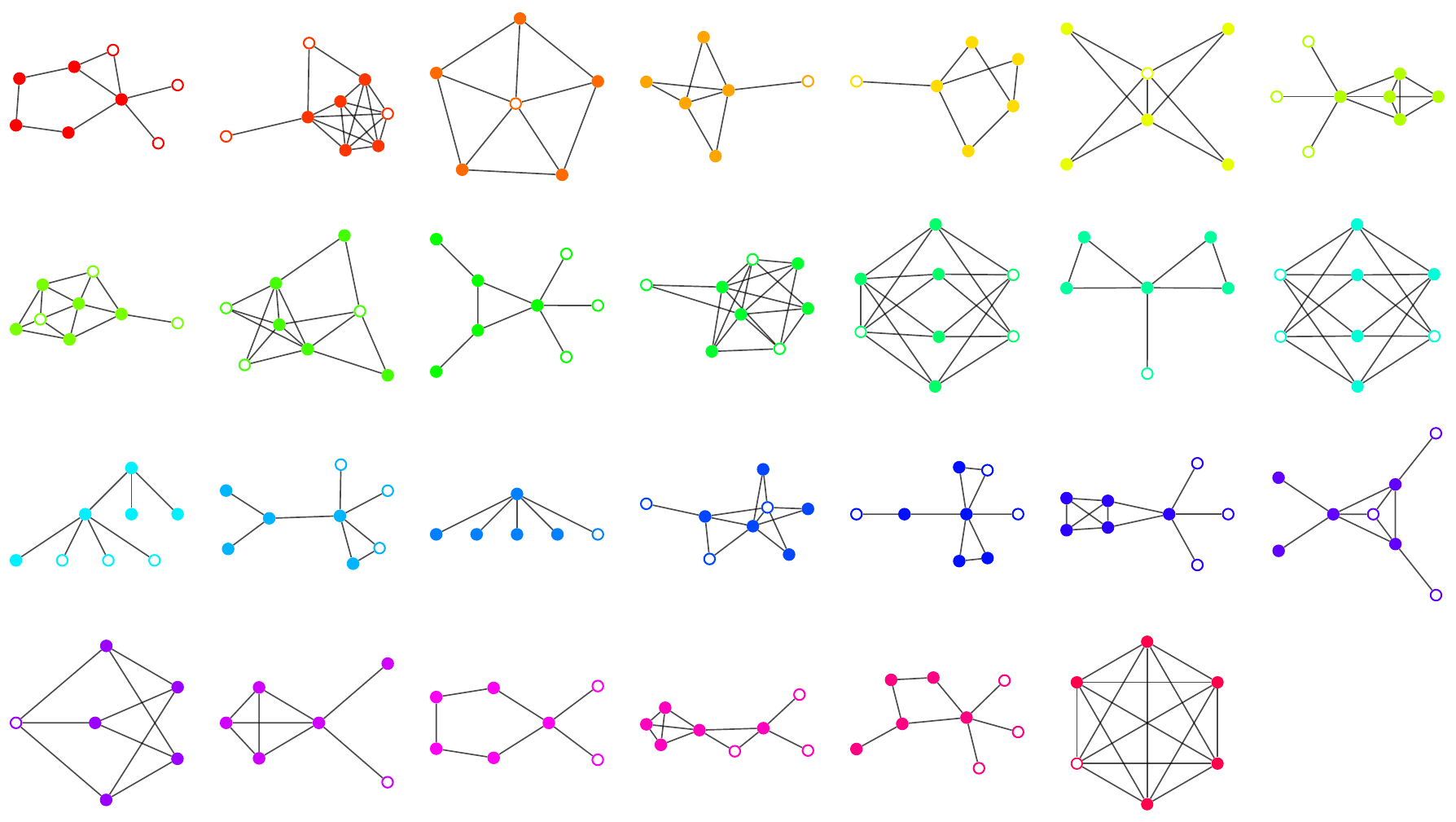}
	\caption{Color legend for all $k\sys=5$ codes with $k\env\leq 3$ in \Cref{fig:all-codes-B}.
		The system vertices are filled with solid color, the environment vertices in white.}
	\label{fig:all-codes-legend-5}
\end{figure}

\begin{table}[t]
	\centering
	\begin{tabular}{r|llllll}
		\toprule
		$k\sys$ & 2 & 3 & 4 & 5 & 6 & 7  \\
		\midrule
		$\#_\Gamma$ of graphs \cite{A001349} & 1 & 6& 112& 11117& 11716571& $>164 \times 10^9$ \\
		$\#_{\ket\Gamma}$ of graph states (\Cref{eq:graph-ct1}) & 42 & 2474 & 809830 & $>298\times 10^6$ & $>152\times 10^{12}$ & $ 100\times 10^{18}$ \\
		filtered graph states & 32 & 256 & 3072 & 43008 & 917504 &  27951104 \\ 
		\bottomrule
	\end{tabular}
	\caption{Number of graph states to analyse in the exhaustive search before and after filtering out redundancies.}
	\label{tab:graph-counts}
\end{table}

As part of our numerical study we also performed a brute-force search over all small graph states and collected their respective CI thresholds.
This is a challenging problem: the number of connected graphs on $k$ vertices $\#_\Gamma(k)$ grows exponentially in $k\sys$, see \Cref{tab:graph-counts}.
On top of that, any subset of the total vertices can act as system vertices.
This means that with $k\sys < k\tot \le 2k\sys$ vertices overall, there are
\begin{align}\label{eq:graph-ct1}
	\#_{\ket\Gamma} = \sum_{k=k\sys+1}^{2k\sys} \binom{k}{k\sys} \#_\Gamma(k)
\end{align}
different graph states to consider; as listed in \Cref{tab:graph-counts}, already for $k\sys=5$ this is more than $10^8$ states.

However, since $\Gamma\envtoenv$ does not factor into the expressions for $\lambda$ and $\lambda\sys$ in \Cref{eq:lambda-state-sum,eq:lambda-state-tr-sum}, we are overcounting: edges within the environment play no role.
Furthermore, permuting system and environment vertices in a synchronous fashion also leaves the CI invariant.
Finally, permutations \emph{only} within the environment---keeping the edges to the system vertices attached---are redundant as well.

These observations suggest the following counting:
\begin{enumerate}
	\item List all non-isomorphic system graphs on $k\sys$ vertices; call this set $S$.
	\item For each $\Gamma\sys\in S$, we iterate over all binary matrices of size $k\sys \times k\sys+1$; if the $i$\textsuperscript{th} row has an entry in column $j$, we draw an edge from the $i$\textsuperscript{th} system vertex to the $j$\textsuperscript{th} environment vertex.
	\item There are no edges within the environment only.
\end{enumerate}

This procedure can be optimized further by deleting isomorphic duplicates from the final list of graphs (note that even though we started with all non-isomorphic system graphs, we will generate redundant links between system and environment).
Yet even without this final optimization, the number of graphs to search over shrinks dramatically (cf.\ \Cref{tab:graph-counts}).

With this method, we managed to exhaustively search all graphs up to $k\sys=5$ system vertices, while restricting the overall vertex count $k\tot\le 8$---i.e.\ we performed a complete exhaustive search for $k\sys \leq 4$, and an exhaustive search for $k\sys=5$ and $k\env\leq 3$.
The constraining factor at this point is definitely memory, as e.g.\ the $\lambda$ and $\lambda\sys$ expressions take $\approx 3.5$ TB of memory in the latter case.

We have plotted the CI thresholds for our results for $k\sys\in\{2,3,4,5\}$ in \Cref{fig:all-codes-A,fig:all-codes-B}.
Noteworthy is that around the depolarizing channel the repetition codes perform best (given by both the complete and star graphs in \Cref{fig:all-codes-legend-3,fig:all-codes-legend-5}).
Deviating from unbiased noise, and depending on whether $ZY$ or $ZX$ noise dominates, the local Clifford invariance (which renders star and complete graphs equivalent) has to be broken; one of the choices performs better.

We note that since we restrict our search to connected graphs of one specific instance of $k\sys$, product codes---such as e.g.\ products of the optimal single-letter code---are not included in our search.
This explains why the code families as displayed can be worse than the single-letter threshold.

\end{document}